\documentclass[letterpaper,10pt,leqno]{article}
\usepackage{paper}
\allowdisplaybreaks

\usepackage{cleveref}

\usepackage{listings}
\usepackage{float}
\floatname{algoflt}{Protocol}
\floatstyle{ruled}
\newfloat{algoflt}{tbp}{loa}
\makeatletter
\lst@AddToHook{OnEmptyLine}{\addtocounter{lstnumber}{-1}\vspace{-0.5\baselineskip}}

\lstnewenvironment{algorithm}[2][htb]{%
\lstset{%
    mathescape=true,%
    tabsize=2,%
    numbers=none,%
    numberstyle=\tiny\dcmnumbersty,%
    numberblanklines=false,%
    showlines=false,%
    basicstyle=\ttfamily\footnotesize,%
    backgroundcolor={},%
    columns=fullflexible,%
    emph={%
        [1]if, else, then, and, or, for, do, while, return, exit, not, output, initialize, var %
        },
    emphstyle={%
        [1]\bfseries%
        },%
    escapeinside={/*}{*/},
    escapebegin={\ifmmode\else\hfill\color{gray}\small$\triangleright$ \fi},
    escapeend={\ifmmode\else\fi},
    }%
\setlength{\topsep}{0pt}%
\setlength{\parindent}{0pt}%
\@skiphyperreftrue%
\algoflt[#1]\caption{#2}%
}{\endalgoflt\@skiphyperreffalse}
\makeatother

\usepackage{xspace}

\newcommand{\PRAM}{{\textsc{PRAM}}\xspace}

\newcommand{\HYBRID}{{\textsc{HYBRID}}\xspace}

\newcommand{\HD}{{\textsc{HD}}\xspace}

\newcommand{\SetSSP}{{\textsc{SetSSP}}\xspace}
\newcommand{\LDD}{{\textsc{LDD}}\xspace}
\newcommand{\LDC}{{\textsc{LDC}}\xspace}

\newcommand{\CONGEST}{\ensuremath{\mathsf{CONGEST}}\xspace}

\newcommand{\pr}[1]{\ensuremath{\text{{\bf Pr}$\left[#1\right]$}}}
\newcommand{\E}[1]{\ensuremath{\text{{\bf E}$\left[#1\right]$}}}

\newcommand{\Cutoff}[1]{\ensuremath{\mathcal{C}(#1)}\xspace}
\newcommand{\Texp}[1]{\ensuremath{\mathsf{Texp}(#1)}\xspace}
\newcommand{\Cut}{\ensuremath{\mathsf{Cut}}\xspace}

\newcommand{\cblur}{\ensuremath{c_{\mathsf{blur}}}}
\newcommand{\cldd}{\ensuremath{c_{\mathsf{ldd}}}}

\newcommand{\gblurdash}{\ensuremath{\gamma'_{\mathsf{blur}}}}
\newcommand{\rblur}{\ensuremath{\rho_{\mathsf{blur}}}}
\newcommand{\kblur}{\ensuremath{K_{\mathsf{blur}}}}
\newcommand{\kexp}{\ensuremath{K_{\mathsf{Texp}}}}

\DeclareMathOperator*{\argmin}{arg\,min}
\DeclareMathOperator*{\argmax}{arg\,max}

\newcommand{\xmax}[1]{\ensuremath{{x}^{t}_{\mathsf{max}}}}
\newcommand{\xmin}[1]{\ensuremath{{x}^{t}_{\mathsf{min}}}}

\newenvironment{claim}[1]{\par\noindent\underline{Claim:}\space#1}{}

\usepackage{nicefrac}
\usepackage{mathtools}  % amsmath with extensions
\usepackage{amsfonts}  % (otherwise \mathbb does nothing)
\usepackage{dsfont}
\usepackage{xspace}
\usepackage{subfiles}
\usepackage{apxproof}
\usepackage{todonotes}
\usepackage{mdframed}
\usepackage[skins,most]{tcolorbox}
\usepackage{apxproof}
\usepackage{framed}

\usepackage{tikz}
\usepackage{standalone}
\usepackage{pgfplots}
\pgfplotsset{width=9cm,compat=1.8}
\usetikzlibrary{matrix,positioning, arrows}
\usetikzlibrary{decorations.pathreplacing}
\usetikzlibrary{decorations.pathmorphing}
\tikzset{snake it/.style={decorate, decoration=snake, thick}}
\usetikzlibrary{arrows,backgrounds,calc,trees}
\pgfdeclarelayer{background}
\pgfsetlayers{background,main}
\usetikzlibrary{graphs}% tikz package already loaded by 'tikz' option
\usetikzlibrary{shapes}
\usetikzlibrary{quotes}
\usetikzlibrary{intersections} 
\usetikzlibrary{matrix,positioning}
\usetikzlibrary{decorations.pathreplacing}
\usetikzlibrary{arrows,backgrounds,calc,trees}
\pgfdeclarelayer{background}
\pgfsetlayers{background,main}
\usetikzlibrary {math} 
\usepackage{xcolor}
\definecolor{indigo}{RGB}{0,0,167}
\definecolor{dred}{RGB}{193,27,45}

\newtheoremrep{alemma}[lemma]{Lemma}

\usepackage{xcolor}

\usepackage{booktabs} % For formal tables
\usepackage{tabulary}
\usepackage{threeparttable}
\usepackage{etoolbox}

\usepackage{thm-restate}
\usepackage{lipsum}
\usepackage{wrapfig}

\title{Distributed And Parallel Low-Diameter Decompositions for Arbitrary and Restricted Graphs}
\author{Jinfeng Dou, Thorsten Götte, Henning Hillebrandt, Christian Scheideler, Julian Werthmann}

\date{\today}
\begin{document}

\maketitle

\noindent {\textbf{Abstract:}} We consider the distributed and parallel construction of low-diameter decompositions with strong diameter for (weighted) graphs and (weighted) graphs that can be separated through $k \in \tilde{O}(1)$ shortest paths.
This class of graphs includes planar graphs, graphs of bounded treewidth, and graphs that exclude a fixed minor $K_r$.
We present algorithms in the PRAM, CONGEST, and the novel HYBRID communication model that are competitive in all relevant parameters.

Given $\mathcal{D} > 0$, our low-diameter decomposition algorithm divides the graph into connected clusters of strong diameter $\mathcal{D}$.
For a arbitrary graph, an edge $e \in E$ of length $\ell_e$ is cut between two clusters with probability $O(\frac{\ell_e\cdot\log(n)}{\mathcal{D} })$.
If the graph can be separated by $k \in \tilde{O}(1)$ paths, the probability improves to $O(\frac{\ell_e\cdot\log \log n}{\mathcal{D} })$.
In either case, the decompositions can be computed in $\tilde{O}(1)$ depth and $\tilde{O}(kn)$ work in the PRAM and $\Tilde{O}(1)$ time in the HYBRID model.
In CONGEST, the runtimes are $\tilde{O}(HD + \sqrt{n})$ and $\tilde{O}(HD)$ respectively.
All these results hold w.h.p.

Broadly speaking, we present distributed and parallel implementations of sequential divide-and-conquer algorithms where we replace exact shortest paths with approximate shortest paths.
In contrast to exact paths, these can be efficiently computed in the distributed and parallel setting [STOC '22].
Further, and perhaps more importantly, we show that instead of explicitly computing vertex-separators to enable efficient parallelization of these algorithms, it suffices to sample a few random paths of bounded length and the nodes close to them.
Thereby, we do not require complex embeddings whose implementation is unknown in the distributed and parallel setting.

\section{Introduction}

This paper considers the efficient construction of so-called (probabilistic) low-diameter decompositions (\LDD) for general and restricted graphs in the \CONGEST, \PRAM, and the \HYBRID model.
An \LDD of a (weighted) graph $G := (V,E,\ell)$ is a partition of $G$ into disjoint subgraphs $\mathcal{K}(G) := K_1, \ldots, K_{N}$ with $K_i:= (V_i,E_i)$.
We will refer to these subgraphs as \emph{clusters}.
In a partition, each node is contained in exactly one of these clusters.
An algorithm that constructs an \LDD is parameterized with a value $\mathcal{D} > 0$, which gives an upper bound on the clusters' diameter.
We say that a clustering has \emph{strong} diameter $\mathcal{D}$ if each $K_i$ has a diameter of at most $\mathcal{D}$. 
That is, between two nodes in $K_i$ there is a path of length (at most) $\mathcal{D}$ that only consists of nodes in $K_i$.
In contrast, an \LDD with a weak diameter creates possibly disconnected subgraphs with a diameter of $\mathcal{D}$ with respect to the original graph $G$.
Here, for all nodes $v,w \in K_i$, there is a path of length (at most) $\mathcal{D}$ from $v$ to $w$ that can use nodes outside of $K_i$.
We measure the decomposition quality by the number of edges that are \emph{cut} between clusters.
An edge $\{v,w\} \in E$ is cut if its two endpoints $v$ and $w$ are assigned to different clusters.
Thus, a good clustering algorithm only ensures that \emph{most} nodes are in the same cluster as their neighbors.
Without this constraint, a trivial algorithm could simply remove all edges from the graph and return clusters containing one node each.
There are several ways to count the edges that are cut.
In this work, we consider so-called \emph{probabilistic} decompositions.
Here, the probability for an edge $z = \{v,w\}$ to be cut between two clusters $K_i$ and $K_j$ with $i \neq j$ that contain its respective endpoint depends on its length $\ell_z$.
This means that \emph{short} edges are cut less likely than \emph{long} edges.
We use the following formal definition of \LDD's in the remainder of this work:
\begin{definition}[Probabilistic Low-diameter Decomposition (\LDD)]
\label{def:ldd}
    Let $G := (V,E,\ell)$ be a weighted graph and $\mathcal{D} > 0$ be a distance parameter. 
    Then, an algorithm computes a \LDD with quality $\alpha \geq 1$ creates a series of disjoint clusters $\mathcal{K} := K_1, K_2, \ldots, K_{N}$ with $K_i:= (V_i,E_i)$ where $V_1 \sqcup \ldots \sqcup V_N = V$.
    Futher, it holds:
    \begin{enumerate}
        \item Each cluster $K_i \in \mathcal{K}$ has a strong diameter of at most $\mathcal{D}$.
        \item Each edge $z := \{v,w\} \in E$ of length $\ell_z$ is cut between clusters with probability $O\left(\frac{\ell_z\cdot \alpha}{\mathcal{D}}\right)$.
%         \item \textbf{High Clustering Probability:} Each node is contained in a cluster with probability (at least) $\beta$.
     \end{enumerate}
% If all nodes $v \in V$ are added to some cluster, i.e., it holds $\beta=1$, we simply say that the \LDD has quality $\alpha$.
\end{definition}
\noindent If we do not require each node to be in a cluster, we use the related notion of a Low-Diameter \emph{Clustering} (\LDC). In a \LDC with quality $(\alpha,\beta)$ and strong diameter $\mathcal{D}$ we create a series of disjoint clusters $\mathcal{K} := K_1, K_2, \ldots, K_{N}$ with $K_i:= (V_i,E_i)$.
Each cluster $K_i \in \mathcal{K}$ has a strong diameter of at most $\mathcal{D}$ and each edge $z := \{v,w\} \in E$ of length $\ell_z$ is cut between clusters with probability (at most) $O\left(\frac{\ell_z\cdot \alpha}{\mathcal{D}}\right)$.
Further, each node is part of a cluster with probability at least $\beta$.
As we will see, \LDD's and \LDC's are nearly equivalent as we can build an \LDD from an \LDC by repeatedly applying it until all nodes are clustered. 

\medskip

\noindent It is known that the best possible value of $\alpha$ is $O(\log n)$ for general graphs. However, if the graph class is restricted, better decompositions are possible.
So, in addition to solutions for general graphs, this work will consider the restricted class of so-called $k$-path separable graphs that Abraham and Gavoille introduced in \cite{DBLP:conf/podc/AbrahamG06}.
Roughly speaking, a weighted graph $G = (G,E,\ell)$ is $k$-path separable if the \emph{recursive} removal of $k$ shortest paths results in connected components containing a constant fraction of the nodes.
Moreover, a graph $G = (G,E)$ is \emph{universally} $k$-path separable if it is $k$-path separable for \emph{every} weight function $\ell$.
Note that this makes universal $k$-path separability a topological property rather than one that depends on the weight function.
Abraham and Gavoille formalized this generic class of path separators as follows:
\begin{definition}[$k$-path separator{\cite{DBLP:conf/podc/AbrahamG06}}]
    A weighted graph $G := (V,E,\ell)$ with $n$ vertices is \emph{$k$-path separable}, if there exists a subset of vertices $S$, called a $k$-path separator, such that:
    \begin{enumerate}[topsep=0pt,itemsep=-1ex,partopsep=1ex,parsep=1ex]
        \item $S:=\mathcal{P}_0\cup \mathcal{P}_1\cup \mathcal{P}_2\cup\dots$, where each $\mathcal{P}_i := \{P_{i_1}, P_{i_2}, \ldots\}$ is a set of shortest paths in $G\backslash \bigcup_{j<i} \mathcal{P}_j$.
        \item The total number of paths in $S$ is at most $k$, i.e., $\sum_{\mathcal{P}_i \in S} |\mathcal{P}_i| \leq k$.
        \item Each connected component in $G \setminus S$ contains at most $\frac{n}{2}$ nodes.
    \item Either $G\setminus S$ is empty or $k$-path separable. 
    \end{enumerate}
If $G$ is $k$-path separable for any weight function, we call $G$ universally $k$-path separable.
We will sometimes abuse notation and use $P_{i_j} \in S$ when we refer to some path $P_{i_j} \in \mathcal{P}_i$ and $\mathcal{P}_i \in S$.
\label{def:$k$-path separator}
\end{definition}
In this work, we will only consider universally $k$-path separable graphs.
Given this definition, a natural question is which graph classes are $k$-path separable. 
Clearly, all graphs that have separators that consist of at most $\eta$ vertices are trivially universally $\eta$-path separable. 
This follows because every node can be seen as the shortest path to itself in any subgraph that it is contained in.
Notably, graphs of bounded tree-width $\tau$ have separators that consist of $\tau$ nodes \cite{DBLP:journals/jal/RobertsonS86} and are therefore universally $\tau$-path separable as under any cost function on the edges, nodes are shortest paths to themselves.
Further, due to Thorup and Zwick, it is known that planar graphs are universally $3$-path separable \cite{Thorup04, TZ01}.
More generally, Abraham and Gavoille showed that every graph $G := (V,E,\ell)$ that does not include a fixed clique minor $K_r$ is universally $k$-path separable where $k := f(r)$ depends only on $r$ and not the size of $G$ \cite{DBLP:conf/podc/AbrahamG06}.
These graph classes have received significant attention in the context of distributed algorithms ( cf. \cite{DBLP:conf/podc/GhaffariH16, DBLP:conf/soda/GhaffariH16, DBLP:conf/podc/GhaffariH21, DBLP:conf/wdag/GhaffariP17, DBLP:conf/spaa/IzumiKNS22}).  
Not only do they contain graph topologies that frequently appear in practice, they also have favorable properties for distributed algorithms as they exclude pathologic worst case topologies. 
As $k$-path separable graphs are a superclass of these restricted graphs, we derive novel results for planar graphs, graphs of bounded treewidth and most importantly graphs that exclude a fixed minor $K_r$.

\subsection{Applications of LDDs}
Having defined them, we can discuss the natural question of why we consider \LDD's in the first place. 
Simply put, \LDD's are part of an algorithm designer's toolkit for developing efficient divide-and-conquer algorithms, as they offer a generic way to decompose a given graph.
As such, creating \LDD's with low edge-cutting probability plays a significant role in numerous algorithmic applications.
In the following, we present a comprehensive list of applications (which we do not claim to be complete):

\noindent\textbf{Tree Embeddings:} For example, they can be used to embed graphs into trees, i.e., constructing so-called metric tree embeddings (MTE) and low-stretch spanning trees (LSSP).
In both, we map a given graph $G:= (V,E,\ell)$ into a randomly sampled tree $T = (V,E_T,\ell_T)$ in a way that preserves the path lengths of the original graph $G$ on expectation.
In an LSSP, it must hold ${E_T \subset E}$, i.e., we compute a subtree of $G$, while a MTE can use virtual edges not present in the original graph.
LSSPs and MTEs have proven to be a helpful design paradigm for efficient algorithms as many computational problems are significantly easier to solve on trees than on general graphs. 
For example, they have been used to construct distributed and parallel solvers for certain important classes of LPs \cite{DBLP:conf/spaa/BellochGKPT11, DBLP:journals/dc/AnagnostidesLHZG23, DBLP:conf/sirocco/Vos23, DBLP:journals/siamcomp/GhaffariKKLP18}, which can then be used solve \textsc{MaxFlow} and other complex problems. 
Nearly all constructions use \LDD's (or variants thereof) as subroutines; see \cite{DBLP:conf/focs/AbrahamBN08,AN19,DBLP:conf/wdag/BeckerEGL19, EEST08} for the construction of LSSPs and
\cite{Bar96,DBLP:conf/stoc/Bartal98,DBLP:conf/esa/Bartal04,DBLP:journals/jcss/FakcharoenpholRT04}
for MTEs.

\noindent\textbf{Light Spanners:} Another interesting avenue for \LDD's are so-called \emph{light spanners} first introduced by Elkin, Neimann, and Solomon in \cite{DBLP:journals/siamdm/ElkinNS15}.
Given a weighted graph $G = (V,E,\ell)$, a $t$-spanner is a subgraph $H \subset G$ that approximately preserves the distances between the nodes of the original graph by a factor $t$, i.e., for all pairs of node $v,w \in V$, it holds $d_H(v,w) \leq t \cdot d_G(v,w)$.
A spanner's lightness $L_H$ is the ratio between the weight of all its edges and the MST of graph $G$, i.e., it holds
    $L_H = \frac{\sum_{e \in H} \ell_e}{\sum_{e \in MST(G)} \ell_e}$
Thus, compared to the LSSPs and MTEs, light spanners have more freedom in choosing their edges as they are not required to be trees.
The lightness is (arguably) a natural quality measure for a spanner, especially for distributed computing. 
Consider, for example, an efficient broadcast scheme in which a node wants to spread a message to every other node in $G$. Further, suppose the edges' length as the cost of using this edge to transmit a message. 
If we send the messages along the edges of the light spanner $H$ of $G$, the total communication cost can be bound by its lightness, and the stretch bounds the cost of the path between any two nodes.
Despite this, there are almost no distributed algorithms that $t$-spanners with an upper bound on the lightness with \cite{DBLP:conf/podc/ElkinFN20} being a notable exception of stretch $O(k)$ with lightness $\Tilde{O}(n^{\frac{1}{k}})$ \footnote{$\Tilde{O}(\cdot)$ hides polylogarithmic factors. We stress here that we neither claim nor try to optimize these factors.} .
The complicating factor in the distributed construction is that lightness is a global measure.
Many previously known distributed spanner constructions like \cite{BS07} or \cite{MPV+15} are very local in their computations as they only consider each node's $t$-neighborhood with $t \in O(\log n)$.
Therefore, the lightness of the resulting spanners is unbounded despite them having few edges.
\LDD's are connected to light spanners through an observation by Neiman and Filtser in \cite{DBLP:journals/algorithmica/FiltserN22}:
They show that with black-box access to an \LDD algorithm with quality $\alpha$, 
one can construct an $O(\alpha)$-spanner with lightness $\Tilde{O}(\alpha)$ for any weighted graph that $G = (V,E,\ell)$. 
In addition to an algorithm that creates \LDD's for geometrically increasing diameters, they only require so-called $(\alpha,\beta)$-nets for which they already provided a distributed implementation in \cite{DBLP:conf/podc/ElkinFN20}.
Thus, finding better distributed algorithms \ LDDs, especially for restricted graphs, drectly improves the distributed construction of light spanners.

% \noindent\textbf{Compact Routing Schemes:} \emph{Routing schemes} are distributed algorithms that manage the forwarding of data packets between network devices. 
% Thus, developing efficient routing schemes is essential to enhance communication among multiple parties in distributed systems.
% It is known that \LDD's can be used to construct routing schemes with routing paths that are not significantly longer than the shortest paths in $G$.
% In particular, an \LDD with strong diameter and quality $\alpha$ directly enables us to construct a routing scheme where all paths only differ by a factor of $O(\alpha)$ from the true shortest paths. 
% We elaborate in this in Chapter, which is dedicated to finding such routing schemes.
% There, we present two efficient construction of routing schemes with the help of \LDD's. 

\subsection{A Meta-Model for Distributed and Parallel Graph Algorithms}
In this work, we present algorithms that can be implemented in several different models of computation for parallel and distributed systems. 
Notably, this includes \CONGEST and the \PRAM, which are the de-facto standard models for distributed and parallel computing, respectively. 
In addition to these two \emph{classic} models of computation, we will also use the recently established \HYBRID model. The \HYBRID model was introduced in \cite{DBLP:conf/soda/AugustineHKSS20} as a means to study distributed systems that leverage multiple communication modes. In \HYBRID, a node can send a message to all its neighbors in $G$ and also $O(\log n)$ nodes of its choice (as long as no node receives more than $O(\log n)$ messages). For a discussion of this model, we refer to \cite{DBLP:conf/soda/AugustineHKSS20,DBLP:phd/dnb/Schneider23}. 

Rather than exploiting these models' intricacies, we reduce both our problems to $\Tilde{O}(1)$ applications of a $(1+\epsilon)$-approximate shortest path algorithm with $\epsilon \in O\left(\nicefrac{1}{\log^2 n}\right)$ and some simple aggregations. 
To be precise, our algorithms are based on approximate set-source shortest paths (\SetSSP) and minor aggregations.
We formalize the class of approximate \SetSSP algorithms we use as follows:
\begin{definition}[Approximate \SetSSP with Virtual Nodes]
\label{def:sssp}
    Let $G := (V,E)$ be a weighted graph and let $G' := (V',E',w)$ with $V' := V \cup \{s_1, s_2, \ldots\}$ be the graph that results from adding $\Tilde{O}(1)$ virtual nodes to $G$.
    Each virtual node can have an edge of polynomially bounded weight to every virtual and non-virtual node in $G'$. 
    Finally, let $S \subset V'$ be an arbitrary subset of virtual and non-virtual nodes.
    Then, an algorithm that solves $(1+\epsilon)$-approximate set-source shortest path with virtual nodes computes the following:
    \begin{itemize}[topsep=0pt,itemsep=-1ex,partopsep=1ex,parsep=1ex]
        \item Each node $v \in V'\setminus S$ learns a predecessor $p_v \in N(v)$ on a path of length at most $(1+\epsilon)d(v,S)$ to some node in $S$ and marks the edge $\{v,p_v\}$.
        Together, all the marked edges imply an approximate shortest path tree $T$ rooted in set $S$.
        \item Each node $v \in V'$ learns its distance $d_T(v,S) \leq (1+\epsilon)d(v,S)$ to $S$ in tree $T$, i.e., its exact distance to $S$ in $T$ and its $(1+\epsilon)$-approximate distance to $S$ in $G'$.
    \end{itemize}
\end{definition}
\noindent Our second building block are so-called \emph{minor aggregations}, which were first introduced in \cite{DBLP:conf/podc/GhaffariH16,DBLP:conf/soda/GhaffariH16}.
Consider a network $G=(V, E)$ and a (possibly adversarial) partition of vertices into disjoint subsets $V_1, V_2, \ldots, V_N\subset V$, each of which induces a \emph{connected} subgraph $G[V_i]$. We will call these subsets \emph{minors}. 
Further, let each node $v \in V$ have private input $x_v$ of length $\Tilde{O}(1)$, i.e., a value that can be sent along an edge in $\Tilde{O}(1)$ rounds. 
Finally, we are given an aggregation function $\bigotimes$ like $\texttt{SUM},\texttt{MIN},\texttt{AVG},\dots$, which we want to compute in each part.
Then a \emph{minor aggregation} computes these functions in all minors $G[V_i]$.
Note that the diameter of these minors might be (much) larger than the diameter of $G$.
Therefore, the corresponding algorithm has to use edges outside of each minor to be efficient.

Both approximate shortest paths and minor aggregation can be efficiently implemented in restricted graphs. 
To be precise, it holds:

\begin{theorem}
\label{thm:agg_runtime}
Let ${G := (V,E,\ell)}$ be a (weighted) graph and let ${C}_1, \ldots, {C}_N$ be set of disjoint subgraphs of $G$.
Further, let $\mathcal{A}$ be an algorithm that is independently executed on each ${C}_1, \ldots, {C}_N$.
Suppose that $\mathcal{A}$ can be broken down into $\tau_s$ steps of ${(1+\epsilon)}$-approximate \SetSSP computations with ${\epsilon < 1}$ as defined in Definition \autoref{def:sssp} and $\tau_m$ minor aggregations.
Then, it holds:
\begin{itemize}
    \item On any graph $G$, $\mathcal{A}$ can be executed on all ${C}_1, \ldots, {C}_N$ in parallel in $\Tilde{O}\left((\epsilon^{{-2}}\cdot\tau_s + \tau_m)\cdot(\HD + \sqrt{n})\right)$ time in \CONGEST, in $\Tilde{O}\left(\epsilon^{{-2}}\cdot\tau_s + \tau_m\right)$ time in \HYBRID, and $\Tilde{O}\left(\epsilon^{{-2}}\cdot\tau_s + \tau_m\right)$ depth in the \PRAM, w.h.p.
    \item If $G$ is $k$-path separable, $\mathcal{A}$ can be executed on all ${C}_1, \ldots, {C}_N$ in parallel in $\Tilde{O}\left((\epsilon^{{-2}}\cdot\tau_s + \tau_m)\cdot k \cdot \HD\right)$ time in \CONGEST, in $\Tilde{O}\left(\epsilon^{{-2}}\cdot\tau_s + \tau_m\right)$ time in \HYBRID and $\Tilde{O}\left((\epsilon^{{-2}}\cdot\tau_s + \tau_m)\right)$ depth in the \PRAM, w.h.p.
\end{itemize}
\end{theorem}
% \begin{theorem}
% \label{thm:agg_runtime}
% Let $G := (V,E,\ell)$ be a (weighted) graph and let $\mathcal{A}$ be an algorithm that takes $G$ as an input.
% Suppose that $\mathcal{A}$ can be broken down into $\tau$ steps, s.t. in each step we either perform a minor aggregation or an $(1+\epsilon)$-approximate SSSP as defined in \autoref{def:sssp}.
% Then, algorithm $\mathcal{A}$ can be implemented with the following time bounds:
% \begin{itemize}[topsep=0pt,itemsep=-1ex,partopsep=1ex,parsep=1ex]
%     \item {\PRAM:} $\mathcal{A}$ can be implemented in $\Tilde{O}(\epsilon^{-2} \tau)$ depth and $\Tilde{O}(\tau m)$ work for any graph $G$.
%     \item {\CONGEST:} $\mathcal{A}$ can be implemented in $\Tilde{O}\left(\epsilon^{-2} \cdot \tau \cdot (\HD+\sqrt{n})\right)$ time for any graph $G$.
%     If $G$ is $k$-path separable, we require $\Tilde{O}\left(\tau \cdot k \cdot \HD\right)$ time, w.h.p.
%     \item {\HYBRID}\cite{DBLP:conf/soda/AugustineHKSS20}\textbf{:} $\mathcal{A}$ can be implemented in $\Tilde{O}\left(\epsilon^{-2} \cdot \tau\right)$ time for any graph $G$.
% \end{itemize}
% \end{theorem}
The bounds for approximate \SetSSP in \CONGEST and \PRAM follow from a recent breakthrough result \cite{DBLP:conf/stoc/RozhonGHZL22}. 
The runtimes for minor aggregation in \CONGEST and \PRAM are given in \cite{DBLP:conf/podc/GhaffariH21, DBLP:conf/stoc/HaeuplerWZ21}.
In particular, in the \CONGEST model, if the graph excludes a fixed clique minor $K_r$, the the runtime of a minor aggregation is within $\Tilde{O}(r\HD)$, where $\HD$ is the graph's hop diameter \cite{DBLP:conf/podc/GhaffariH21}.
For general graphs on the other hand, the runtime is $\Tilde{O}(\HD + \sqrt{n})$ \cite{DBLP:conf/stoc/HaeuplerWZ21}.
Note that $k$-path separable graphs exclude $K_{4k+1}$ as minor and thereofore the bound of $\Tilde{O}(\HD)$ applies.
This fact was shown in \cite{DBLP:conf/faw/DiotG10}.
Finally, \cite{DBLP:phd/dnb/Schneider23} showed the bounds for both \SetSSP and minor aggregation in the \HYBRID model. 

Note that we are mainly interested in algorithm that require $\Tilde{O}(1)$ aggregations and $(1+\epsilon)$-approximate \SetSSP computations with $\epsilon \in \Omega(\nicefrac{1}{\log^c n})$ as these algorithm can then be efficiently implemented in all three models.

\subsection{Related Work}
Despite this vast number of applications for \LDD's on \emph{weighted} graphs, research on \LDD's in a distributed setting has mostly focused on the unweighted case, producing many efficient algorithms in this regime. 
In particular, the research \cite{DBLP:journals/dc/LeviMR21,DBLP:conf/podc/ChangS22, DBLP:conf/podc/Chang23, DBLP:journals/tcs/ElkinN22, DBLP:conf/soda/0001GHIR23} in the \CONGEST model focused on so-called \emph{network decompositions} that enable fast algorithms for local problems like MIS, Coloring, or Matching.
In principle, these algorithms could be applied to weighted graphs and produce \LDD's. 
However, their runtime depends on the diameter $\mathcal{D}$ of the resulting clusters, which might be much larger than the hop diameter (or even $n$) if the graph is weighted.
We give a more detailed overview of these algorithms in Section \ref{sec:relatedwork}.
Two notable exceptions, however, explicitly consider weighted graphs and are closely related to our results: The work of Becker, Emek, and Lenzen \cite{BeckerEL20} creates \LDD's of quality $O(\log n)$ with weak diameter\footnote{Actually, \cite{BeckerEL20} proves a stronger property of the diameter. Although the diameter is weak, the number of nodes outside of a cluster that are part of the shortest path between two nodes of a cluster is limited. For many applications, this is sufficient and just as good as a strong diameter. } for general graphs. 
We note that $O(\log n)$ is the best quality we can hope for in an \LDD due to a result by Bartal\cite{Bar96}.
The algorithm requires $\Tilde{O}(\HD + \sqrt{n})$ in the \CONGEST model, which is optimal as each distributed \LDD construction in a weighted graph requires $\Omega(\HD + \sqrt{n})$ time \cite{ghaffari_universally-optimal_2022} as we can derive approximate shortest paths from it.
Just as our algorithm, \cite{BeckerEL20} consists of $\Tilde{O}(1)$ ${(1+\epsilon)}$-approximate \SetSSP computations with $\epsilon \in O\left(\nicefrac{1}{\log^2 n}\right)$.
Further, there is the work of Rozhon, Elkin, Grunau, \textcircled{r} Haeupler \cite{DBLP:conf/focs/RozhonEGH22}, which makes two significant improvements compared to \cite{BeckerEL20}. They present a decomposition with strong diameter (instead of weak), and their construction is deterministic (instead of randomized).
Conversely, they have a slightly worse quality of only $O(\log^3 n)$.
Again, the algorithm consists of $\Tilde{O}(1)$ ${(1+\epsilon)}$-approximate \SetSSP computations with $\epsilon \in O\left(\nicefrac{1}{\log^2 n}\right)$.
For restricted graphs like planar, bounded treewidth, or minor-free graphs, we are unaware of a 
distributed algorithm explicitly designed for weighted graphs.

\subsection{Our Contribution \& Structure of This Paper}

In a nutshell, prior works \cite{BeckerEL20} and \cite{DBLP:conf/focs/RozhonEGH22} have shown that computing \LDD's for general graphs can be broken down into $\Tilde{O}(1)$ applications of your favorite \SetSSP algorithm in your favorite model of computation.
We extend this innovative insight in two ways by improving the algorithm of \cite{BeckerEL20} and showing that for $k$-path seperable graph one can compute \LDD's of near optimal quality with only approximate shortest paths.

To be precise, first, we show that we can achieve the (asymptotically) best possible quality of $O(\log n)$ for general through approximate \SetSSP and minor aggregation. 
Our first main theorem is the following.
\begin{theorem}[LDDs for General Graphs]
\label{thm:clustering_general}
    %\label{lemma:kr_clustering}
Let ${\mathcal{D}>0}$ be an arbitrary distance parameter and ${G:=(V,E,\ell)}$ be a (possibly weighted) undirected graph.
Then, there is an algorithm that creates an \LDD of $G$ with strong diameter $\mathcal{D}$ and quality $O(\log n)$. The algorithm can be implemented with $\Tilde{O}(1)$ minor aggregations and $\Tilde{O}(1)$  $(1+\epsilon)$-approximate \SetSSP computations where $\epsilon \in O\left(\nicefrac{1}{\log^{2}n} \right)$.
\end{theorem}
\noindent 
By \autoref{thm:agg_runtime}, this implies that the algorithm can, w.h.p., be implemented in $\Tilde{O}(\HD+\sqrt{n})$ time in \CONGEST, and $\Tilde{O}(1)$ depth in the \PRAM and $\Tilde{O}(1)$ time in \HYBRID.
Thus, our algorithm is currently the best randomized \LDD construction in \CONGEST for general weighted graphs.
It has same runtime as \cite{DBLP:conf/focs/RozhonEGH22} and \cite{BeckerEL20}, creates clusters of \emph{strong} diameter, and has (asymptotically) optimal quality of $O(\log n)$.
We achieve this improvement through a more fine-grained analysis of the well-known exponential-delay clustering first used by Miller et al. \cite{MPV+15}
that was already used by Becker, Lenzen, and Emek in \cite{BeckerEL20}. 
In fact, we only present a little addition to their existing algorithm to ensure the strong diameter.

\medskip

\noindent Second, for $k$-path separable graphs with $k \in \Tilde{O}(1)$, we present an algorithm with almost exponentially better quality.
More precisely, our second main theorem is the following.
\begin{theorem}[LDDs for $k$-Path Seperable Graphs]
\label{thm:clustering_k_path}
Let ${\mathcal{D}>0}$ be an arbitrary distance parameter and ${G:=(V,E,\ell)}$ be a (possibly weighted) $\Tilde{O}(1)$-path separable graph.
Then, there is an algorithm that creates an \LDD of $G$ with strong diameter $\mathcal{D}$ and quality $O(\log\log n)$. The algorithm can be implemented with $\Tilde{O}(1)$ minor aggregations and $\Tilde{O}(1)$  $(1+\epsilon)$-approximate \SetSSP computations where $\epsilon \in O\left(\nicefrac{1}{\log^{2}n} \right)$.
\end{theorem}
Thus, by the same reasoning as above, the algorithm can, w.h.p., be implemented in $\Tilde{O}(\HD)$ time in \CONGEST, and $\Tilde{O}(1)$ depth in the \PRAM and $\Tilde{O}(1)$ time in \HYBRID.
However, when comparing this result to other LDD constructions for restricted graphs, the situation is more nuanced than before.
Recall that each universally $k$-path separable graph excludes $K_{4k+1}$ as minor, and thus, we need to compare ourselves to algorithms for $K_{r}$-free graphs.
There is a distributed algorithm \cite{DBLP:conf/podc/Chang23} by Chang for unweighted graphs that present \LDD's in $K_r$-free graphs with quality $O(r)$ in the \CONGEST model.
While our bound of $O(\log\log n)$ is exponentially better than that for general graphs, it is still far from this bound as it depends on $n$ and not only on $r$.
On the flip side, our algorithm has an (asymptotically) optimal runtime of $\Tilde{O}(\HD)$ where the hidden factors depend on $r$.
The other distributed algorithm is tailored to small clusters, as the runtime is polynomial in the cluster's diameter $\mathcal{D}$. 
Therefore, our algorithm is faster for large diameters $\mathcal{D}$, which arguably trades off its worse cutting probability.
Also note that our techniques are completely different that those of \cite{DBLP:conf/podc/Chang23}.
We elaborate more on this in \autoref{sec:relatedwork}.

Our main technical novelty for proving \autoref{thm:clustering_k_path} is a weaker form of a path separator that can be constructed without computing an embedding of $G$. 
Prior algorithm for $k$-path separable graphs, e.g., \cite{DBLP:conf/podc/AbrahamG06} often employ a so-called shortest decomposition, we recursively compute $k$-path separators until the graph is empty.
The computation of these separators, however, requires precomputing a complex embedding, which we cannot afford in a distributed or parallel setting.
Instead, we take a different approach to completely avoid the (potentially expensive) computation of an embedding.
We show that by sampling approximate shortest paths (in a certain way),
we can obtain what we will call a weak $\mathcal{D}$-separator.
A weak $\mathcal{D}$-separator $S'$ only ensures that in the graph $G \setminus S$, all nodes have at most a constant fraction of nodes in distance $\mathcal{D}$.
In particular, in contrast to a \emph{classical} separator, the graph $G \setminus S'$ might still be connected.
We strongly remark that, in doing considering weak separators, we sacrifice some of useful properties of a traditional seperator.
Nevertheless, this only results in some additional polylogarithmic factors in the runtime our applications.

\medskip

Due to the many technical results, this paper takes the form of an extended abstract. 
In the main part of the paper, we present our algorithms, state our main lemmas, and sketch their proofs.
More detailed descriptions and the full analysis including all proofs can be found in the appendix.
That being said, the main part of the paper is structured as follows:
In \autoref{sec:ldd_general} we present the algorithm behind \autoref{thm:clustering_general}.
In doing so, we also present several useful intermediate results that we will use in later sections.
Then, in \autoref{sec:weak_separators} we present our novel technique to compute weak separators.
In \autoref{sec:k_path_cluster}, we combine our insights into clustering algorithms from \autoref{sec:ldd_general} and findings on the construction of separators from \autoref{sec:weak_separators} to develop the algorithm behind \autoref{thm:clustering_k_path}.

\section{Strong Diameter LDDs for General Graphs}
\label{sec:ldd_general}

In this section, we present the algorithm behind \autoref{thm:clustering_general}. 
The description is divided into three sections/subroutines.
First, in \autoref{sec:pseudopadded}, we present a generic clustering algorithm. 
Second, in \autoref{sec:ldc} we combine this algorithm with a technique from \cite{BeckerEL20} that cuts all edges with the correct probability.
However, we may not add each node to a cluster.
Finally, in \autoref{sec:ldd}, we present a generic technique 
that recursively applies the algorithm from \autoref{sec:ldc} to cluster all nodes while asymptotically preserving the cut probability of an edge.

\subsection{Pseudo-Padded Decompositions Using Approximate Shortest Paths}
\label{sec:pseudopadded}

We begin with a crucial technical theorem that will build the foundation of most our results.
A key component in this construction is the use of truncated exponential variables. 
In particular, we will consider exponentially distributed random variables truncated to the $[0,1]$-interval.
Loosely speaking, a variable is truncated by resampling it until the outcome is in the desired interval.
In the following, we will always talk about variables that are \emph{truncated to $[0,1]$-interval} when we talk about truncated variables.
The density function for a truncated exponential distribution with parameter $\lambda>1$ is defined as follows:
\begin{definition}[Truncated Exponential Distribution]
We say a random variable $X$ is truncated exponentially distributed with parameter $\lambda$ if and only if its density function is $f(x) := \frac{\lambda \cdot e^{-x\lambda}}{1-e^{-\lambda}}$.
    We write $X \sim \mathsf{Texp}(\lambda)$. Further, if $X \sim \mathsf{Texp}(\lambda)$ and $Y := \mathcal{D} \cdot X$, we write $Y \sim \mathcal{D} \cdot \mathsf{Texp}(\lambda)$.
\end{definition}
The truncated exponential distribution is a useful tool for decompositions that has been extensively used in the past \cite{AGGNT19, Filtser19, DBLP:conf/spaa/MillerPX13}.
Using a truncated exponential distribution and $(1+\epsilon)$-approximate \SetSSP computations, we prove a helpful auxiliary result, namely:
\begin{theorem}[Pseudo-Padded Decomposition for General Graphs]\label{thm:generalpartition}
\label{theorem:generalpartition}
    Let $\mathcal{D}>0$ be a distance parameter, $\epsilon$ be an error parameter, $G:=(V,E,\ell)$ a (possibly weighted) undirected graph, and $\mathcal{X} \subseteq V$ be a set of possible cluster centers.
    Suppose that for each node $v \in V$, the following two properties hold:
    \begin{itemize}
        \item \textbf{Covering Property:} There is at least one $x \in \mathcal{X}$ with $d_G(v,x) \leq \mathcal{D}$.
        \item \textbf{Packing Property:} There are at most $\tau$ centers  $x' \in \mathcal{X}$ with $d_G(v,x') \leq 6\mathcal{D}$.
    \end{itemize}
    Then, for $\epsilon\in o(\nicefrac{1}{\log\tau})$ there is an algorithm that computes a series of connected clusters $\mathcal{K} = K_1, \ldots, K_N$ with strong diameter $4(1+\epsilon)\mathcal{D}$ where for all nodes $v \in V$ and all $\epsilon \leq \gamma \leq \frac{1}{32}$, it holds:
    \begin{align*}
        \pr{B(v,\gamma\mathcal{D}) \subset K(v)} \geq e^{-\Theta((\gamma+\epsilon)\log\tau)} - O(\nicefrac{1}{n^c}).
    \end{align*}
    Here, $K(v)$ denotes the cluster that contains $v$.
    The algorithm can be implemented with one ${(1+\epsilon)}$ approximate \SetSSP computation and $\Tilde{O}(1)$ minor aggregations.
\end{theorem}
Technically, this algorithm is a generalization of the algorithm in \cite{Filtser19} that is based on \textbf{exact} shortest path computations.
This algorithm is itself derived from \cite{DBLP:conf/spaa/MillerPVX15}.
Our algorithm replaces all these exact computations through $(1+\epsilon)$-approximate calculations.
The main analytical challenge is carrying the resulting error $\epsilon$ through the analysis to obtain the bounds in the theorem.
The same approach was already used by Becker, Lenzen, and Emek \cite{BeckerEL20}.
However, our analysis is more fine-grained w.r.t. to the impact of the approximation parameter $\epsilon$.

In the following, we give the high-level idea behind the construction.
An intuitive way to think about the clustering process from \cite{DBLP:conf/spaa/MillerPVX15, Filtser19} is as follows: 
Each center $x$ draws value $\delta_x \sim \mathcal{D}\cdot\mathsf{Texp}(2+2\log\tau)$ wakes up at time $\mathcal{D}-\delta_x$.
Then, it begins to broadcast its identifier. 
The spread of all centers is done in the same unit tempo. 
A node $v$ joins the cluster of the \emph{first} identifier that reaches it, breaking ties consistently.
If we had $O(\mathcal{D})$ time, we could indeed implement it exactly like this.
However, this approach is infeasible for a general $\mathcal{D} \in \tilde{\Omega}(1)$ in weighted graphs as $\mathcal{D}$ may be arbitrarily large.
Instead, we will model this intuition using a virtual super source $s$ and  shortest path computations.
For each center $x \in \mathcal{X}$, we independently draw a value $\delta_x \sim \mathcal{D}\cdot\mathsf{Texp}(2+2\log\tau)$
from the truncated exponential distribution with parameter $2+2\log(\tau)$.
We assume that $\tau$ (or some upper bound thereof) is known to each node.
Then, we add a virtual source $s$ with a weighted virtual edge $(s,x)$ to each center $x \in \mathcal{X}$ with weight $w_x := (\mathcal{D}- \delta_x)$.
Then, we compute an $1+\epsilon$-approximate \SetSSP from $s$ with $\epsilon \leq \frac{1}{40\log\tau}$. 
Any node joins the cluster of the \emph{last} center on its (approximate) shortest path to $s$.

With exact shortest paths, this construction would preserve intuition.
Any node whose shortest path to $s$ contains center $x$ as its last center on the path to $s$ would have been reached by $x$'s broadcast first. 
The statement is not that simple with approximate distances as the approximation may introduce a \emph{detour}, so the center that minimizes $(\mathcal{D} - \delta_x) + d(x,v)$ may not actually cluster a node $v \in V$.
In particular, two endpoints of a short edge $\{v,w\}$ may end up in different clusters although the same center minimizes the distance to both.
The approximation error for $v$ could be large, while for $w$, it is very low, which can cause undesired behavior.
The nodes will only be added to same cluster, if its center's head start (which is determined by the random distance to the source) is large enough to compensate for the \emph{approximation errors}.
This, however, depends on the approximation and \emph{not} an the length of the edge we consider.
Nevertheless, if the error $\epsilon$ is small enough, we get similar clustering guarantees that are good enough for our purposes. 
In particular, if the difference in distance between the two closest centers is larger than $O(\epsilon\mathcal{D})$, the clustering behaves very similar to its counterpart with exact distances.
In our analysis, we show this by carefully carry the approximation error through the analysis.
We do not introduce any fundamental new techniques, but simply carry the approximation error all the way through the analysis of the corresponding clustering process with exact distances presented in \cite{Filtser19}. 

In addition to \autoref{thm:generalpartition} itself, this also allows us to observe the positive correlation between nodes that are on the shortest path to some cluster center. 
This closer look also allows us to prove the technical fact about the pseudo-padded decompositions:
\begin{lemma}
\label{lemma:special}
    Let $\mathcal{K}$ be a partition computed by the algorithm from \autoref{thm:generalpartition}. Then, there a constant $c \geq 1$ such that for given node $u \in K_i$ with probability $(\nicefrac{1}{2})$ the following holds: There is path $P_u \coloneqq (u = v_1 , \ldots, v_N = x_i)$ of length at most $2(1+\frac{1}{40\log \tau})\mathcal{D}$ to the cluster center $x_i$ of $K_i$ where for each $v_j \in P_u$, it holds $B\left(v_j, \frac{\mathcal{D}}{c\log \tau}\right) \subseteq K_i$.
\end{lemma}
This lemma states that if a single node $v \in V$ is padded, all nodes on a short path to its cluster center are likely padded as well. 
While this initially sounds very technical, it roughly translates to the following: Consider a clustering $\mathcal{K} = K_1, \ldots, K_N$. Then for a suitably chosen $\rho \coloneqq \frac{\mathcal{D}}{c \cdot \log \tau}$, in each cluster $\mathcal{K}_i = (V_i,E_i)$ there is a constant fraction of nodes $V'_i \subseteq V$ in the distance $\rho$ to their closest node in a neighboring cluster $K_j \neq K_i$ \textbf{and} for any two $v,w \in V'_i$ there is a path of length $6\mathcal{D}$ (via the cluster center) that only consists of nodes in $V_i'$.
This insight will be \emph{crucial} for our next algorithm.

\medskip

\noindent A more detailed description and the full proof can be found in \autoref{sec:padded_appendix}.

\subsection{Strong LDCs from Pseudo-Padded Decompositions}
\label{sec:ldc}

In this section, we present 
% an 
% algorithm that creates an \LDC of strong diameter $\mathcal{D}$.
% To achieve this, we will generalize an algorithm by Becker, Emek, and Lenzen \cite{BeckerEL20}.
% Our main result is 
a generic algorithm that creates an \LDC with strong diameter $\mathcal{D}$ of quality $\left(O(\log \tau), \nicefrac{1}{2}\right)$ if the number of nodes that can be centers of clusters has been sufficiently sparsed out, such that each node can only be in one of $\tau$ clusters.
We show that it holds:
\begin{theorem}[A Generic Clustering Theorem]
\label{thm:genericldd}
%\label{thm:clustering_general}
    Let $\mathcal{D}>0$ be a distance parameter, $G:=(V,E,\ell)$ a (possibly weighted) undirected graph, and $\mathcal{X} \subseteq V$ be a set of marked nodes.
    Suppose that for each node $v \in V$, the following two properties hold:
    \begin{itemize}
        \item \textbf{Covering Property:} There is at least one $x \in \mathcal{X}$ with $d(v,x) \leq \mathcal{D}$.
        \item \textbf{Packing Property:} There are at most $\tau$ centers  $x' \in \mathcal{X}$ with $d(v,x') \leq 6\mathcal{D}$.
    \end{itemize}
Then, there is an algorithm that creates an \LDC of strong diameter $8\mathcal{D}$ with quality $\left(O\left(\nicefrac{\log\tau}{\mathcal{D}}\right), \nicefrac{1}{2}\right)$.
The algorithm can be implemented with $\Tilde{O}(1)$ minor aggregations and $\Tilde{O}(1)$  $(1+\epsilon)$-approximate \SetSSP computations where $\epsilon \in O\left(\nicefrac{1}{\log^{2}n} \right)$.
\end{theorem}
The algorithm uses the so-called \emph{blurry ball growing} (BBG) technique.
This is perhaps the key technical result that Becker, Emek, and Lenzen introduced in \cite{BeckerEL20}.
It provides us with the following guarantees:
\begin{restatable}[Blurry Ball Growing (BBG), cf. \cite{BeckerEL20, DBLP:conf/focs/RozhonEGH22}]{lemma}{bbg}
\label{lemma:bbg}
Given a subset $S \subseteq V$ and an arbitrary parameter $\rho > 0$, an algorithm for BBG outputs a superset $S' \supseteq S$ with the following properties
\begin{enumerate}
    \item An edge $\{v,w\} \in E$ of length $\ell_{(v,w)}$ is cut with probability $\pr{v \in S', w \not\in S'} \leq O\left(\frac{\ell_{(v,w)}}{\rho}\right)$.
    \item For each node $v \in S'$, it holds $d(v,S) \leq \frac{\rho}{1-\alpha} \leq 2\rho$ where $\alpha \in O\left(\nicefrac{\log \log n}{\log n}\right)$.
\end{enumerate}
BBG can be implemented using $\Tilde{O}(1)$ $(1+\epsilon)$ approximate \SetSSP computations with $\epsilon \in O\left((\nicefrac{\log\log n}{\log n})^2\right)$.
\end{restatable}
This technique allows us to create clusters with a low probability of cutting an edge while only having access to approximate shortest paths. 
Note that in \cite{DBLP:conf/focs/RozhonEGH22} the dependency on $\epsilon$ was improved to $O\left(\nicefrac{1}{\log n}\right)$. 
The paper also presents a deterministic version of blurry ball growing.
However, we will only use the probabilistic version introduced above.

Given this definition, we now give an overview of the algorithm that creates an \LDC:
In the following, define $\rho \coloneqq \frac{\mathcal{D}}{c \cdot \log \tau}$ where $c$ is the constant from \autoref{lemma:special}.
The algorithm first computes a pseudo-padded decomposition $\mathcal{K} = K_1, \ldots, K_N$ of strong diameter $8\mathcal{D}$ using the algorithm of \autoref{theorem:generalpartition}.
To do this, we choose the given set $\mathcal{X}$ of marked nodes as the set of cluster centers.

The following steps are executed in each cluster $K_i$ in parallel.
Within each cluster $K_i$, we determine a so-called \emph{inner cluster} $K'_i \subseteq K_i$ of strong diameter $6\mathcal{D}$ where each node $v \in V'_i$ has distance (at least) $\rho' \coloneqq \rho/2$ to the closest node in different cluster $K_j \neq K_i$. 
These inner clusters can be determined via two (approximate) \SetSSP computations:
First, all nodes calculate the $2$-approximate distance to the closest node in a different cluster.
To do this, we create a virtual supersource $s$, and all nodes $v \in K_i$ with an edge to a node $w \in K_j$ create a virtual edge of length $\ell_{v,w}$ to $s$.
If they have several such neighbors, they choose the shortest edge.
Then, we perform a $2$-approximate \SetSSP from $s$ on the resulting virtual graph.
The edges $\{v,w\}$ between clusters will not be considered in the \SetSSP.
We then mark all nodes where this $2$-approximate distance exceeds $2\rho'$ as \emph{active nodes}.
Next, we consider the graph $G'$ induced by the active nodes. 
Using a $(1+\frac{1}{40\log\tau})$-approximate \SetSSP calculation, the active nodes compute if they have a path of length at most $3\mathcal{D}$ to their cluster center in $G'$.
If so, they add themselves to the inner cluster $K'_i$ of $K_i$.
Recall that their exact distance to the next cluster is at least $\rho'$ as the paths are $2$-approximate.
Further, for any two nodes $v,w \in K'_i$ there is a path of length at most $6\mathcal{D}$ via the cluster center and so the inner cluster has a strong diameter of $6\mathcal{D}$.
Thus, this procedure \emph{always} results in an inner cluster $K'_i$ with the desired properties.
In the third and final stage of the algorithm, the actual clusters are computed.
To this end, the \emph{inner clusters} $\mathcal{K}' \coloneqq \bigcup_{i=1}^N K'_i$ are used as the input set to the blurry ball growing (BBG) procedure from \autoref{lemma:bbg}. 
In particular, we choose the parameter for the BBG to be $\rho'' \coloneqq \rho'/2$ and compute the superset $S(\mathcal{K}') \coloneqq \mathsf{blur}(\mathcal{K}',\rho'')$. 
Now define final clustering $\mathcal{C} = C_1, \ldots, C_N$ where $C_i \coloneqq K_i \cap S(\mathcal{K}')$. 
In other words, the cluster $C_i$ contains the inner cluster $K'_i$ and all nodes from $K_i$ added by the BBG process. 

The resulting clustering fulfills all three properties required by \autoref{thm:genericldd}.
For the cutting probability, note that the distance from an inner cluster to a different cluster is $\rho'$, and BBG only adds nodes in distance smaller than $2\rho'' = \rho'$.
Thus, all edges with an endpoint in cluster $K_i$ are only cut by the blurring process and not by the initial pseudo-padded decomposition.
Thus, the choice of $\rho''$ implies that edges are cut with correct probability of $O(\ell/\rho'') = O(\frac{\ell \cdot \log\tau}{\mathcal{D}})$ by \autoref{lemma:bbg}.
The (strong) diameter of the cluster is at most $8\mathcal{D}$ as each node is in the distance at most $\rho \leq \mathcal{D}$ to a node from an inner cluster, and all nodes in an inner cluster are in the distance $6\mathcal{D}$ to one another.

Therefore, it only remains to show that a nodes is clustered with prob. $\nicefrac{1}{2}$.
We can prove this via \autoref{lemma:special}.
For all nodes $v \in V$ the following holds with probability $(\nicefrac{1}{2})$:
Node $v \in K_i$ \textbf{and} all nodes on the path $P_v$ of length $2(1+\frac{1}{40\log\tau})\mathcal{D}$ to its cluster center $x_i$ are distance at least $\rho$ to the closest node in a different cluster.
Therefore, all these nodes will mark themselves active.
As the $2$-approximate \SetSSP only overestimates, these nodes compute a distance more than $\rho = 2\rho'$, which causes them to be active.
Moreover, this implies that the path $P_v$ is fully contained $G'$ the graph induced by active nodes.
As the $P_u$ is of length $2(1+\frac{1}{40\log\tau})\mathcal{D}$ and ends in the cluster center $x_i$, the $(1+\frac{1}{40\log\tau})$-approximate \SetSSP computation will find a path of length at most $(1+\frac{1}{40\log\tau})2(1+\frac{1}{40\log\tau})\mathcal{D} \leq 3\mathcal{D}$ in $G'$.
Therefore, with a probability of at least $\nicefrac{1}{2}$, node $v$ and all nodes on the path to the cluster center will be added to the inner cluster $K'_i$.
Thus, \autoref{thm:genericldd} follows as half of all nodes are in an inner cluster and are therefore guaranteed to be in some cluster $C_i$.

\medskip

\noindent A more detailed description and the full proof can be found in \autoref{sec:appendix_ldc}.

\subsection{Strong LDDs from Strong LDCs}
\label{sec:ldd}

In this section, we show that we can create an \LDD with quality ${O}(\log n)$ for any graph within $\Tilde{O}(1)$ minor aggregations and $(1+\epsilon)$-approximate \SetSSP's.
Note that, at first glance, the theorem follows almost directly from the \autoref{thm:genericldd}. 
Choose $\mathcal{X} = V$ and $\mathcal{D}' = 8\mathcal{D}$ and apply the theorem to a graph $G$ to see it.
We obtain a clustering with diameter $\mathcal{D}'$ if each node has one node in distance $\mathcal{D}'/8$ and at most $\tau$ nodes in distance $\mathcal{D}'$.
Clearly, for every possible choice of parameter $\mathcal{D}'$, each node has at least one marked node in distance $\mathcal{D}'/8$, namely itself.
Further, for every possible choice of parameter $\mathcal{D}'$, each node has at most $n$ marked node in distance $\mathcal{D}$ because there are only $n$ nodes in total.
Thus, by choosing all nodes as centers, we obtain a clustering with quality $O(\log n)$ and a diameter we can choose freely. 
While the diameter and edge-cutting probability of the resulting clustering are (asymptotically) correct, we can only guarantee that at least half of the nodes are clustered on expectation.
To put it simply, we do not have a partition as required.
However, we can reapply our algorithm to the graph induced by unclustered nodes.
This clusters half of the remaining nodes on expectation.
It is easy to see that, if we continue this for $O(\log n)$ iterations, we obtain the desired partition, w.h.p.
Further, this will not significantly affect the edge-cutting probability, as each edge that is not cut will be added to a cluster with constant probability.
Therefore, there cannot be too many iterations where the edge can actually be cut.
Thus, applying the algorithm until all nodes are clustered intuitively produces the required \LDD of quality $O(\log n)$.
We formalize and prove this intuition in the following lemma.
Not that we state it in a more general fashion to reuse it the next section.
\begin{lemma}
\label{lemma:folkloreldd}
Let $G = (V,E,\ell)$ be a weighted graph.
Let $\mathcal{A}$ be an algorithm that for each subgraph of $G' \subseteq G$ can create clusters $K_1, \ldots, K_N$ where each edge is cut between clusters with probability at most $\alpha$, and each node is added to a cluster with probability at least $\beta$.
Then, recursively applying $\mathcal{A}$ to $G$ until all its nodes are in a cluster creates an \LDD of quality $O(\alpha\beta^{-2})$.
The procedure requires $O(\beta^{-1}\log n)$ applications of $\mathcal{A}$, w.h.p. 
\end{lemma}

\noindent The full proof can be found in \autoref{sec:appendix_ldd}. Combining \autoref{lemma:folkloreldd} with \autoref{thm:genericldd} then yields \autoref{thm:clustering_general}.

\section{Weak Separators from Approximate Shortest Paths}
\label{sec:weak_separators}

We now move from arbitrary graphs to $k$-path seperable graph.
In this section, we consider the distributed computation of separators for $k$-path separable graphs.
Recall that a separator $S$ is a \emph{subgraph} of $G$ such that in the induced subgraph $G \setminus S$ that results from removing $S$ from $G$, every connected component only contains a constant fraction of $G$'s nodes.
Separators are a central tool in designing algorithms for restricted graph classes as they give rise to efficient divide-and-conquer algorithms. 
Our algorithms are no exception. 
In particular, we will compute the following noevel type of separators.
The following construct is of the main building blocks of both our routing and our clustering results: 
\begin{restatable}[Weak $\kappa$-Path $(\mathcal{D},\epsilon)$-Separator]{definition}{weaksep}
\label{def:weak_k_path_separator}
    Let $G := (V,E,\ell)$ be a weighted graph, $\mathcal{D} > 1$ be an arbitrary distance parameter, and $\epsilon > 0$ be an approximation parameter.
    Then, we call the set $S(\mathcal{D},\epsilon) := \left(\mathcal{P}_1, \ldots, \mathcal{P}_\kappa\right)$ with $\mathcal{P}_i := (P_i,B_i)$ a weak $\kappa$-path separator, if it holds:
    \begin{enumerate}[topsep=0pt,itemsep=-1ex,partopsep=1ex,parsep=1ex]
        \item Each $P_{i} \in \mathcal{P}_i$ is a (approximate) shortest path in $G\setminus \bigcup_{j=1}^{i-1} \mathcal{P}_j$ of length at most $4\mathcal{D}$.
        \item Each $B_i \subseteq B_G(P_{i},\epsilon\mathcal{D})$ is a set of nodes surrounding path $P_i$.
        \item For all $v \in (V\setminus \bigcup_{j=1}^{\kappa} \mathcal{P}_i)$ it holds $|B_{G\setminus S}(v,\mathcal{D})| \leq (\nicefrac{7}{8})\cdot n$.
        %The graph $G\setminus \bigcup_{j=1}^{\kappa} \mathcal{P}_i$ is either empty, or 
    \end{enumerate}
\end{restatable}
Note that this definition is generic enough also to include vanilla $k$-path separators (where each set $B_i$ only contains the path itself and no further nodes) and \emph{traditional} vertex separators (where each path is a single node). 
Thus, for certain $k$-path separable graphs, we can already use existing algorithms to compute these separators in distributed and parallel models\footnote{For planar graphs, we can use the algorithm provided in \cite{DBLP:conf/stoc/LiP19} to compute a separator that consists of $4$ paths.  While for graphs of bounded tree-width $\tau$, we can use  \cite{DBLP:conf/spaa/IzumiKNS22} to obtain a separator that consists of $\tau$ nodes.}.
For general $k$-path separable graphs, however, we need new techniques.

\subsection{Why Computing Separators for $k$-Path Separable Graphs is Hard}
Before, we go into details, let us first review some background on separators that motivates why we will mainly work with the weaker concept.
To this end, we first review prior results and identify possible difficulties. 
Starting with positive results, it \emph{is} possible to construct path separators for planar graphs in a distributed setting. In fact, the state-of-the-art algorithms for computing DFS trees in \cite{DBLP:conf/wdag/GhaffariP17} or computing a graph's (weighted) diameter in \cite{DBLP:conf/stoc/LiP19} in planar graphs construct a path separator for planar graphs.
The algorithm presented in these papers constructs a path separator in $\Tilde{O}(\HD)$ time, w.h.p.
However, while the existence of an efficient algorithm for a planar graph gives hope, it turns out to be very tough for other $k$-path separable graphs.
The constructions in both \cite{DBLP:conf/wdag/GhaffariP17} and \cite{DBLP:conf/stoc/LiP19} heavily exploit the fact that a planar embedding of a (planar) graph can be computed in $\Tilde{O}(\HD)$ time in \CONGEST due to Ghaffari and Haeupler in \cite{DBLP:conf/soda/GhaffariH16}.
Given a suitable embedding, the ideas could be generalized.
Thus, finding a suitable embedding could be a way to extend this algorithm to more general $k$-separable graphs.
A good candidate for such an embedding can be found in \cite{DBLP:conf/podc/AbrahamG06}.
Abraham and Gavoille present an algorithm that computes a $f(r)$-path separator for $K_r$-free graphs \cite{DBLP:conf/podc/AbrahamG06}.
Their algorithm relies on a complex graph embedding algorithm by Seymour and Robertson\cite{DBLP:journals/jal/RobertsonS86}.
On a very high level, the embedding divides $G$ into subgraphs that can almost be embedded on a surface similar to planar graph. 
By \emph{almost}, we mean that in each subgraph, there is only small number of non-embeddable parts (with special properties) that need to be handled separately. 
The concrete number of these parts depends only on $r$.
Given such an embedding, Abraham and Gavoille show that one can efficiently compute $k$-path separator.
Sadly, we have little hope that the algorithm that computes the embedding can be efficiently implemented in a distributed (or even parallel) context.
It already requires $n^{O(r)}$ time in the sequential setting.
Moreover, it requires sophisticated techniques that can not be trivially sped up by running them in parallel. 
Thus, it remains elusive (to the authors at least) how to implement this algorithm in any distributed or even parallel model.

\subsection{Why Computing Weak Separators for $k$-Path Separable Graphs is Easier}
Now that we have established why it is difficult to compute $k$-path separators, we can proceed to the initially promised weak separators.
Our main insight is that instead of computing an embedding in a distributed and parallel setting, we show that we get very similar guarantees using a weaker separator that can be computed without the embedding.
Thus, we take a different approach to completely avoid the (potentially expensive) computation of an embedding.
We strongly remark that, in doing so, we sacrifice some of the separator's useful properties.
However, in our applications, this only results in some additional polylogarithmic factors in the runtime.
To be precise, we prove the following lemma:
\begin{theorem}[Weak Separators for $k$-Path Separable Graphs]
\label{thm:distributed_weak_separator}
Consider a weighted $k$-path separable graph $G := (V,E,\ell)$ with weighted diameter smaller than $\mathcal{D}$. 
For $\epsilon \geq 0$, there is an algorithm that constructs a weak $O(\epsilon^{-1}\cdot k \cdot \log n)$-path $(\mathcal{D},\epsilon)$-separator, w.h.p.
The algorithm uses $O(\epsilon^{-1}\cdot k \cdot \log n)$ minor aggregations and $O(\epsilon^{-1}\cdot k \cdot \log n)$ 
$2$-approximate shortest path computations, w.h.p.
\end{theorem}

Surprisingly, the algorithm behind the lemma can be stated in just a few sentences.
The core idea is to iteratively sample approximate shortest paths and nodes close to these paths and remove them from the graph until all remaining nodes have few nodes in distance $\mathcal{D}$, w. h.p.
More precisely, we start with a graph $G_0 = G$.
The algorithm proceeds in $T \in O(\epsilon^{-1} \cdot k\cdot \log n)$ steps where in step $t$, we consider graph $(V_t, E_t) \coloneqq G_t \subset G_{t-1}$. 
A single step $t$ works as follows:
First, we pick a node $v \in V_{t}$ uniformly at random.
Then, we compute a $2$-approximate shortest path tree $T_v$ from $w$. 
Let $W_v \subset V_t$ be set of all nodes in distance at most $4\cdot\mathcal{D}$ to $v$ in $T_v$.
We pick a node $w \in W_v$ uniformly at random and consider the corresponding path $P_t = (v, \ldots, w)$.
Using $P_t$ as a set-source, we compute a $2$-approximate shortest path tree $T_{P_t}$ from $P_t$. 
Now denote $B_{P_t}$ to be the nodes in the distance at most $\epsilon\cdot\mathcal{D}$ to $P_t$ in $G$ (and \textbf{not} $G_t$).
Finally, we remove $P_t$ and $B_{P_t}$ from $G_t$ to obtain $G_{t+1}$.
Note that, due to the approximation guarantee, this removes all nodes in distance at least $\nicefrac{\epsilon}{2}\cdot\mathcal{D}$ in $G$.
Further, all steps can be computed using $2$-approximate paths (and few aggregation to pick the random nodes) and are completely devoid of any complex embedding.

\begin{figure}[ht]
     \begin{subfigure}[t]{0.48\textwidth}
         \centering
         \includegraphics[width=\textwidth]{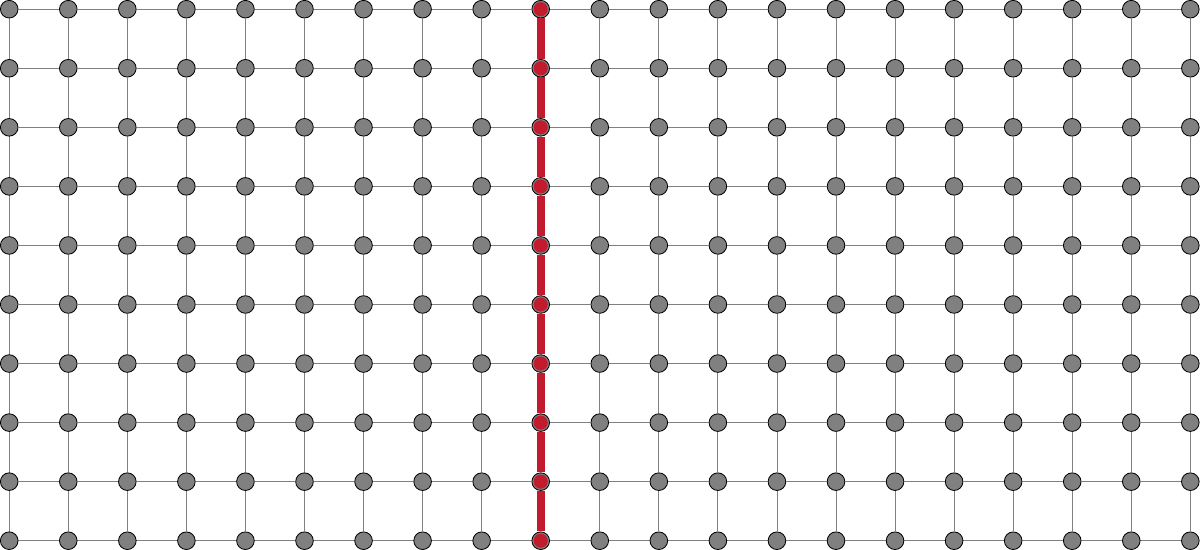}
         \caption{Marked in red is a separator path unknown to the algorithm.}
     \end{subfigure}
     \hfill
     \begin{subfigure}[t]{0.48\textwidth}
         \centering
         \includegraphics[width=\textwidth]{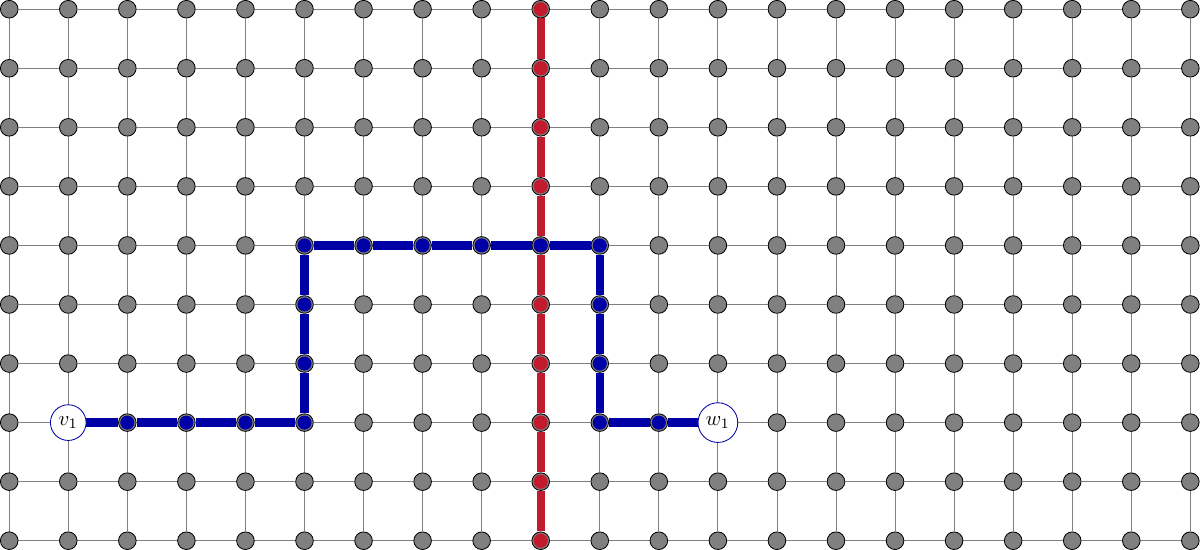}
         \caption{When sampling a random path of bounded length, the probability of crossing the separator is constant.}
     \end{subfigure}
\bigskip
        \begin{subfigure}[t]{0.48\textwidth}
         \centering
         \includegraphics[width=\textwidth]{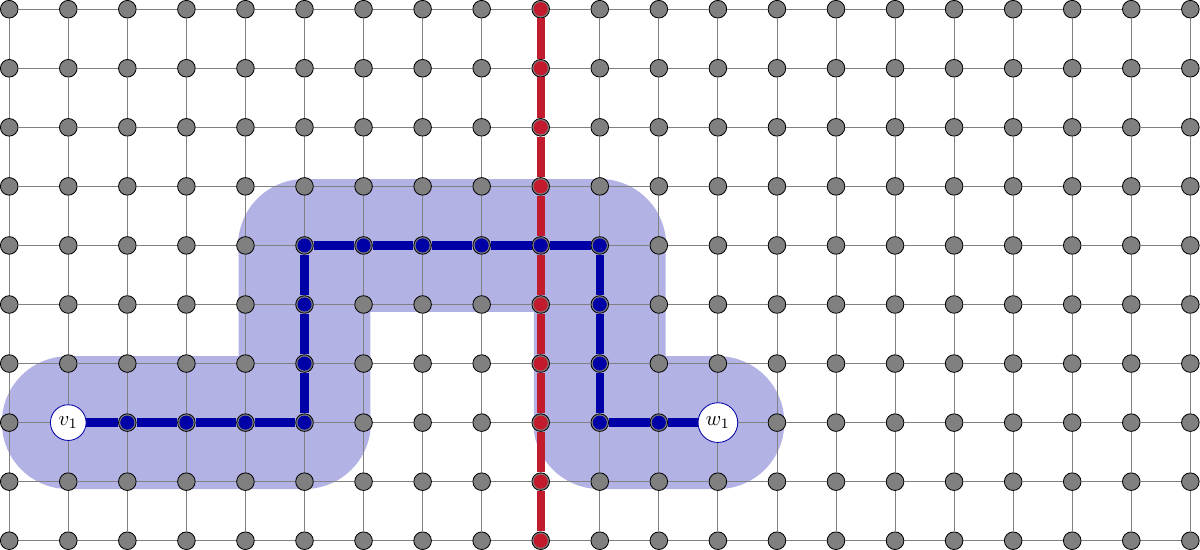}
         \caption{We remove nodes in distance $\epsilon\mathcal{D}$ around the path. 
         This also removes an $\epsilon\mathcal{D}$-fraction of the separator path.}
     \end{subfigure}
    \hfill
     \begin{subfigure}[t]{0.48\textwidth}
         \centering
         \includegraphics[width=\textwidth]{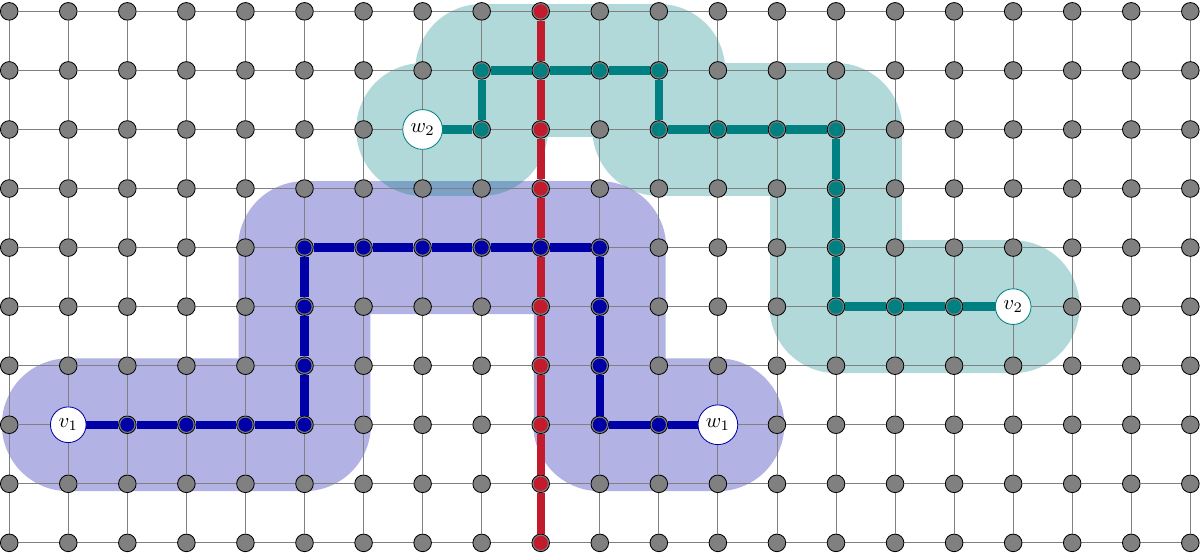}
         \caption{If the residual graph is not well-separated, the probability of crossing (what remains of) the separator is still constant.}
     \end{subfigure}
        \caption{We choose a two-dimensional mesh for illustration. While the operations in Figures \textbf{B} and \textbf{C} are straightforward to prove, the core of our analysis we will be the claim we make in Figure \textbf{D}.}
        \label{fig:sep_algo}
\end{figure}

\noindent Why does this work?
As our main analytical tool, we consider the following potential:
\begin{align*}
    \Phi_t := \left|\left\{ v \in V_t \bigm\vert |B_{G_t}(v,2\mathcal{D})| \geq (\nicefrac{7}{8})\cdot n \right\}\right|
\end{align*}
Further, for brevity, we write $\Phi_t(v) = 1$ if and only if $|B_{G_t}(v,2\mathcal{D})| \geq (\nicefrac{7}{8})\cdot n$ and $\Phi_t(v) = 0$, otherwise. Note that $\Phi_t := \sum_{v \in V} \Phi_t(v)$.

\medskip

\noindent We make the following important observation: 
As soon as this potential drops below $(\nicefrac{7}{8})\cdot n$, we have more than $(\nicefrac{1}{8})\cdot n$ nodes with less than $(\nicefrac{7}{8})\cdot n$ nodes in distance $2\mathcal{D}$.
This implies that no node can have more than $(\nicefrac{7}{8})\cdot n$ nodes in distance $\mathcal{D}$.
Suppose for contradiction that there is a node $v \in V_t$ with more than $(\nicefrac{7}{8})\cdot n$ nodes in distance $\mathcal{D}$.
By the triangle inequality, all these nodes would be in distance $2\mathcal{D}$ to each other.
Thus, the potential would be bigger than $(\nicefrac{7}{8})\cdot n$ as there are more than $(\nicefrac{7}{8})\cdot n$ nodes with more than $(\nicefrac{7}{8})\cdot n$ nodes in distance $\mathcal{D}$. This is a contradiction and therefore the nodes sampled until this point are the desired weak separator. 

Therefore, we want to bound the number of steps until the potential drops.
For this, we will first show the following useful fact about $k$-path separators, namely
\begin{restatable}{lemma}{boundary}
\label{lemma:boundary_simple}
    Let $G := (V, E, w)$ be a weighted {$k$-path separable} graph of $n$ nodes with (weighted) diameter $\mathcal{D}$.
    Then, there exists a set $\mathcal{B} = \{P_1, \ldots, P_\kappa\}$ with $\kappa \leq k$ simple paths of length $32\mathcal{D}$ such that for all $v \in V$, it holds $B_{G \setminus \mathcal{B}}(v,4\mathcal{D}) \leq \left(\nicefrac{3}{4}\right)\cdot n$.    
\end{restatable}
In other words, a subset of bounded length paths intersects with a constant fraction of all paths our algorithm can potentially sample. 
The proof requires copious use of the triangle inequality and the pigeonhole principle and can be found in \autoref{sec:appendix_separator}.

Now consider the event that the algorithm samples a path $P_t$ that \emph{crosses} some path of the separator $\mathcal{B}$, i.e., the path we sample contains (at least) one node from one of the $k$ paths in $\mathcal{B}$.
In this case, we remove a subpath of length $\epsilon\mathcal{D}$ from this path when we remove the set $B_{P_t}$ around the $P_t$.
This follows because we determine the set $B_t$ with respect to the original graph $G$ and not $G_t$.
Thus, after sampling $O(\epsilon^{-1} \cdot \mathcal{D})$ paths that cross $\mathcal{B}$, we must have completely removed $\mathcal{B}$ from $G$.
If $\mathcal{B}$ is removed, each node has at most $(\nicefrac{3}{4})\cdot n$ nodes in distance $\mathcal{D} \leq 2\mathcal{D}$ per definition and we are done.

Thus, we will bound the time until \textbf{either} the potential drops \textbf{or} $O(\epsilon^{-1} \cdot k)$ paths cross $\mathcal{B}$.
In the following, denote the event that the path from $v_t$ to $w_t$ crosses $\mathcal{B}$ as $v_t \rightsquigarrow_{\mathcal{B}} w_t$.
By definition of $\mathcal{B}$, for each node, the set of nodes reachable via $\mathcal{B}$ is at least size $(\nicefrac{1}{4})\cdot n$.
Thus, for all nodes $v \in V_t$ with $\Phi_t(v) = 1$, at least $(\nicefrac{1}{4}) \cdot n - (\nicefrac{1}{8})\cdot n \geq (\nicefrac{1}{8})\cdot n$ of these nodes must still be in distance $2\mathcal{D}$.
All these nodes might be chosen as the path's endpoint $P_t$ if we pick $v$ as the starting point as the $2$-approximate distance is at most $4\mathcal{D}$. 
So, if we sample one of them, we cross $\mathcal{B}$ as all paths of length at most $4\mathcal{D}$ cross $\mathcal{B}$.
Thus, in a configuration with potential $\Phi_t \geq (\nicefrac{7}{8})\cdot n$, the probability sample a path that crosses $\mathcal{B}$ is at least
\begin{align}
   \pr{v_t \rightsquigarrow_{\mathcal{B}} w_t} \geq \pr{\Phi_t(v_t) = 1}\cdot\pr{v_t \rightsquigarrow_{\mathcal{B}} w_t \mid \Phi_t(v_t) = 1} \geq \frac{1}{n} \cdot \Phi_t \cdot \frac{1}{n} \cdot (\nicefrac{1}{8})\cdot n \geq \frac{7}{64} \geq \frac{1}{16}
\end{align}
Thus, if the potential has not dropped after $T$ steps, the expected number of crossings is $(\nicefrac{1}{16})\cdot T$.
Using standard Chernoff-like tail estimates, we can show that for any $T \in \Omega(\log n)$, this would imply that at least $\frac{T}{32}$ paths crossed $\mathcal{B}$, w.h.p.
Thus, within $T \in O(\epsilon^{-1} \cdot k\cdot \log n)$ steps, we either reached a step with low potential or had sufficiently many crossings to remove $\mathcal{B}$.
In either case, we are done as each node had less than $\nicefrac{7}{8}\cdot n$ nodes in distance $\mathcal{D}$.
As up until this point we sampled $T \in O(\epsilon^{-1} \cdot k\cdot \log n)$ paths of length $4\mathcal{D}$ and all nodes in distance $\epsilon\mathcal{D}$ around these paths, this set must be the desired weak separator.

% For the extension to general $k$-path separable graphs, we cannot guarantee that all paths in the separator have length $\mathcal{D}$.
% Therefore, we cannot make this exact arguement.
% However, through copious use of the pigeon hole principle, we can show that each separator $S$ has subset $\mathcal{B} \subseteq S$ of (at most) $k$ paths of length (at most) $4\mathcal{D}$ that is a hitting set for all paths (of length at most $4\mathcal{D}$) between a linear sized set of nodes $V' \subseteq V$ with $(\nicefrac{7}{8})\cdot n$.
% As the length of the paths is bounded, they can be \emph{covered} by $O(\epsilon^{-1} \cdot k)$ paths sampled by our algorithm.

\medskip 

\noindent A visualization of the algorithm can be found in \autoref{fig:sep_algo}. A more detailed description and the full proof can be found in \autoref{sec:appendix_separator}.

\section{Low-Diameter Decompositions for $k$-path Separable Graphs}
\label{sec:k_path_cluster}

In this section, we prove \autoref{thm:clustering_k_path}.
Our algorithm will combine the insights we gathered for \LDD's and \LDC's in general graphs with our novel separator construction.
Our main technical result is the following proposition:
\begin{proposition}[A Clustering Theorem for Restricted Graphs]
    \label{lemma:kr_clustering}
    Let ${\mathcal{D}>0}$ be an arbitrary distance parameter and ${G:=(V,E,\ell)}$ be a (possibly weighted) $\Tilde{O}(1)$-path separable graph. Then, there is an algorithm that creates a series of disjoint clusters $\mathcal{K} := K_1, K_2, \ldots, K_{N}$ with $K_i:= (V_i,E_i)$.
    Futher, it holds:
    \begin{enumerate}
        \item Each cluster $K_i \in \mathcal{K}$ has a strong diameter of at most $\mathcal{D}$.
        \item Each edge $z := \{v,w\} \in E$ of length $\ell_e$ is cut between clusters with probability $O\left(\frac{\ell_z\cdot (\log k + \log\log n)}{\mathcal{D}}\right)$.
        \item Each node $v \in V$ is added to some cluster $K_i \in \mathcal{K}$ with probability at least $\nicefrac{1}{2}$.
\end{enumerate}
The algorithm can be implemented in $\Tilde{O}(k)$ minor aggregations and $\Tilde{O}(k)$ ${(1+\epsilon)}$-approximate \SetSSP computations with $\epsilon \in O\left(\nicefrac{1}{\log^2 n}\right)$.
\end{proposition}
Following this proposition, we can create an LDC for $k$-path separable graphs.
Crucially, the proposition is sufficient to prove \autoref{thm:clustering_k_path} for $\Tilde{O}(1)$-separable graphs:
For $k \in \Tilde{O}(1)$ one application of the algorithm creates a clustering with strong diameter $\mathcal{D}$ and cutting probability $O\left(\frac{\ell_z\cdot (\log\log n)}{\mathcal{D}}\right)$.
Now recall that universally $k$-path separable graphs are closed under minor-taking.
Therefore, the graph $G \setminus \mathcal{K}$ is also $k$-path separable, and we can apply the algorithm to $G \setminus \mathcal{K}$ with the same guarantees to cluster at least half of the remaining nodes.
By \autoref{lemma:folkloreldd}, after $O(\log n)$ recursive applications, all nodes are clustered, and we obtain an \LDD with quality $O(\log\log n)$ as required.
Further, for $k \in \Tilde{O}(1)$ one iteration of the algorithm can, w.h.p., be implemented in $\Tilde{O}(1)$ minor aggregations and $\Tilde{O}(1)$ ${(1+\epsilon)}$-approximate \SetSSP computations with $\epsilon \in O\left(\nicefrac{1}{\log^2 n}\right)$. 
Thus, $O(\log n)$ recursive applications are within the required time complexity of \autoref{thm:clustering_k_path}. 

\medskip

% Suppose we want to construct an \LDD $\mathcal{K} = K_1, \ldots, K_N$ with strong diameter $\mathcal{D}$ for some graph $G=(V,E)$.
Thus, for the remainder, we will focus on \autoref{lemma:kr_clustering}. From a high-level perspective, the algorithm behind \autoref{lemma:kr_clustering} proceeds in two distinct phases: the \emph{backbone phase} and the \emph{refinement phase}.
In the backbone phase, we create a special partition $\mathcal{K}[\mathcal{B}] := K(\mathcal{B}_1), \ldots, K(\mathcal{B}_N)$ of $G$.
Each node $v \in K(\mathcal{B}_i)$ is in distance at most $\mathcal{D}_{BC} < \mathcal{D}$ to some subset $\mathcal{B}_i \subset K_i$ that consist of $\Tilde{O}(k)$ (possibly disjoint) paths of bounded length of $O\left(\mathcal{D}\log^2\right)$.
These paths may not necessarily be short\textbf{est} paths in $G$ or even $K(\mathcal{B}_i)$, only their length is bounded.
We call these subsets $\mathcal{B}_1, \ldots, \mathcal{B}_N$ the \emph{backbones} of the clusters.
In contrast to an \LDD, the clusters do \textbf{not} have simple nodes $v \in V$ as centers. 
As a result, the diameter of the resulting clusters might be much larger than $\mathcal{D}$.
Moreover, a single cluster might not even be connected.
In the second phase, we turn the resulting clustering into a \LDC with connected clusters of strong diameter $\mathcal{D}$.
We do so by computing an \LDC $\mathcal{K}_i = K_{i_1}, \ldots, K_{i_N}$ of strong diameter $\mathcal{D}$ in each cluster $K(\mathcal{B}_i)$.
Here, we exploit the special structure of the backbone cluster and choose a carefully chosen subset of nodes from each backbone path and possible cluster centers such that no more than $\Tilde{O}(k)$ centers can cluster a given node.
These centers will then be used as an input to \autoref{thm:genericldd} to refine each backbone cluster to a set of connected clusters of strong diameter $\mathcal{D}$.

This algorithm is inspired by the works of Abraham, Gavoille, Gupta, Neiman, and Talwar \cite{DBLP:conf/stoc/AbrahamGGNT14,AGGNT19} and the subsequent improvement of their algorithm Filtser\cite{Filtser19}.
Surprisingly, we do not require any fundamental changes to obtain an efficient distributed algorithm.
Both algorithms essentially follow the same two-phase structure sketched above.
They first constructed a partition around paths that are sampled in a certain way, and then they refined each partition individually.
For reference, Abraham, Gavoille, Gupta, Neiman, and Talwar dubbed these subsets around which the partitions are constructed \emph{skeletons} while Filtser used the term $r$-\emph{core}. However, their key properties are (almost) the same. We use the different term \emph{backbones} because we construct them slightly differently.
Our contribution when comparing our algorithm to \cite{DBLP:conf/stoc/AbrahamGGNT14,AGGNT19, Filtser19} is that we show a) that the exact shortest path can be replaced by approximate paths and the blurry ball growing procedure from \autoref{lemma:bbg} and, therefore, can be implemented efficiently and b) that we can our weak separator from \autoref{sec:weak_separators} to parallelize the algorithm and archive a logarithmic recursion depth efficiently. 

In the remainder of this section, we sketch the backbone phase in \autoref{sec:backbone} and the refinement phase in \autoref{sec:refinement}.
Together, the lemmas presented in these sections will prove \autoref{lemma:kr_clustering}.

\subsection{The Backbone Phase}
\label{sec:backbone}
In this section, we sketch how the backbone clustering is computed in the first phase of the algorithm.
Formally, we define the clustering constructed in the first phase as follows:
\begin{definition}[Backbone Clustering]
Let $\mathcal{D}>0$ be a distance parameter and $G:=(V,E,\ell)$ a (possibly weighted) graph.
Then, a $(\alpha, \beta, \kappa)$-backbone clustering is a series of disjoint subgraphs $K(\mathcal{B}_1), \ldots, K(\mathcal{B}_N)$ with $K(\mathcal{B}_i) = (V_i,E_i)$ where $V_1 \sqcup \ldots \sqcup V_N = V$.
Further, the following three conditions hold: 
\begin{enumerate}
        \item Each node $v \in K(\mathcal{B}_i)$ is in distance at most $\mathcal{D}$ to $\mathcal{B}_i$.
        \item  An edge of length $\ell$ is cut between clusters with probability $O\left(\frac{\ell\cdot \alpha}{\mathcal{D}}\right)$.
        \item Each backbone $\mathcal{B}_i$ consists of at most $\kappa$ paths of length $O(\beta\mathcal{D})$.
\end{enumerate}
\end{definition}
\noindent Given this definition, we  show the the following lemma:
\begin{restatable}[Backbone Clustering]{lemma}{backbone}
\label{lemma:backbone_clustering_k_path}
Let $\mathcal{D}>0$ be a distance parameter and $G:=(V,E,\ell)$ a (possibly weighted) $k$-path separable graph.
Then, there is an algorithm that creates an $(\alpha,\beta,\kappa)$-backbone clustering with $\alpha \in O(\log{(k\log n)})$, $\beta \in O(\log^2)$, $\kappa \in O(k\log^2 n)$ and pseudo-diameter $\mathcal{D}$ for $G$.
The algorithm can be implemented with $\Tilde{O}(1)$ minor aggregations and $(1+\epsilon)$-approximate \SetSSP computations where $\epsilon \in O\left(\nicefrac{1}{\log^{2}n} \right)$.
\end{restatable}

$\mathcal{K}[\mathcal{B}]$ can be constructed by a recursive divide-and-conquer algorithm executed in parallel on all connected components of $G$.
The algorithm works as follows:
Let $G_t \subset G$ be the graph at the beginning of the $t^{th}$ recursive step where initially $G_0 = G$. 
At the beginning of each recursive step, we create an \LDD $C_1, \ldots, C_{N'}$ of diameter ${\mathcal{D}}' \in O(\mathcal{D}\log^2 n)$ of $G_t$.
We use the algorithm from \autoref{thm:clustering_general} for this.
The edges between two clusters $C_i$ and $C_j$ are removed and will not be considered in the following iterations.
We denote the graph that results from removing all these edges as $G_t'$.
Then, in each $C_i$, we compute a weak $\Tilde{O}(k)$-path $({\mathcal{D}'},\epsilon)$-separator $\mathcal{S_i}$ with $\epsilon \in O(\nicefrac{1}{\log^2 n})$ using the algorithm from \autoref{thm:distributed_weak_separator}.
Recall that we defined this separator to be the union of $O(\epsilon^{-1} \cdot k)$ paths of length $4\cdot\mathcal{D}'$ and some nodes in distance $\epsilon\mathcal{D}'$ to these paths.
The $O(\epsilon^{-1} \cdot k)$ paths in $\mathcal{S}_i$ will be the backbone of a cluster.
To add further nodes to the cluster, we proceed in two steps:
First, we draw $X_i \sim \Texp{4 \log(\log n)}$ and compute a $(1+\epsilon)$-approximate \SetSSP from $\mathcal{S}_i$.
Then, we mark all nodes at a distance at most $(1+\epsilon)X_i\cdot\mathcal{D}_{BC}/4$ from $\mathcal{S}_i$.
We denote this set as $\kexp(\mathcal{S}_i)$.
Then apply the BBG procedure of \autoref{lemma:bbg} with parameter $\rblur \coloneqq \frac{\mathcal{D}_{BC}}{\cblur \cdot \log\log n}$ from $\kexp(\mathcal{S}_i)$. 
Here,  $\cblur > 8$ is a suitably chosen constant.
This results in a superset $\kblur(\mathcal{S}_i) \subseteq C_i$ for each subgraph $C_i$.
Finally, we add all sets $\kblur(\mathcal{S}_i)$ from all clusters $C_i$ to our clustering $\mathcal{K}[\mathcal{B}]$ and remove them from the graph.

We repeat this process with the remaining graph $G_{t+1} := G'_t \setminus \bigcup \kblur(\mathcal{S}_i)$ until all subgraphs are empty.
As we remove a weak $\mathcal{D}'$-separator in each step \textbf{and} then create an \LDD with strong diameter $\mathcal{D}'$, 
the size of the largest connected components shrinks by a constant factor in each recursion.
Thus, we require $O(\log n)$ recursions overall until all components are empty.

\medskip

\noindent In each cluster $\kblur(\mathcal{S}_i)$, we define the short paths $P_1, \ldots, P_\kappa$ in $\mathcal{S}_i$ to be the backbones.
To prove that the resulting clustering is a proper backbone clustering, we must analyze its pseudo-diameter and edge-cutting probability.
Recall that for an appropriate choice of $\epsilon \in O(\nicefrac{1}{\log^2 n})$, all nodes in $\mathcal{S}_i$ are in distance $\epsilon\mathcal{D}' \leq \nicefrac{\mathcal{D}_{BC}}{4}$ to any path.
Further als $(1+\epsilon) \leq 2$ and $X_i \cdot \nicefrac{\mathcal{D}_{BC}}{4} \leq \nicefrac{\mathcal{D}_{BC}}{4}$, all nodes in $\kexp(\mathcal{S}_i)$ are in distance $(1+\epsilon) \cdot X_i \cdot \nicefrac{\mathcal{D}_{BC}}{4} \leq \nicefrac{\mathcal{D}_{BC}}{2}$ to $\mathcal{S}_i$.
Finally, the distance of any node in $\kblur(\mathcal{S}_i)$ to any node in $\kexp(\mathcal{S}_i)$ is at most $2\rblur \leq \nicefrac{\mathcal{D}_{BC}}{4}$.
Thus, the distance from each node $u \in \kblur(\mathcal{S}_i)$ to any of the paths in the separator is at most
\begin{align}
    {d}(u, \mathcal{B}_i) \leq {d}(v, \kexp(S_i)) + {\mathbf{max}_{v \in \kexp(S_i)}}  {d}(v, S_i) + \mathbf{max}_{w \in S_i}    d(w,\mathcal{B}_i)
     \leq \nicefrac{\mathcal{D}_{BC}}{4} + 2\cdot\nicefrac{\mathcal{D}_{BC}}{4} + \nicefrac{\mathcal{D}_{BC}}{4}
     = 4\nicefrac{\mathcal{D}_{BC}}{4} = \mathcal{D}_{BC}
\end{align}
Thus, only the edge-cutting probability remains to be shown. 
For each edge $z \in E$, two operations can cut an edge in each iteration.
The LDD from \autoref{thm:clustering_general} and the BBG procedure from \autoref{lemma:bbg}.
First, consider the LDD. As we pick $\mathcal{D}' \in O(\mathcal{D}\log^2 n)$, a single application of \autoref{thm:clustering_general} cuts an edge with probability $O\left(\nicefrac{\ell_z\log n}{\mathcal{D}'}\right) = O\left(\nicefrac{\ell_z}{\mathcal{D}\log n}\right)$.
Thus, as we apply it at most $O(\log n)$ times, the probability sums up to $O\left(\nicefrac{\ell_z}{\mathcal{D}}\right)$ by the union bound.
The other possibility is that $z$ is cut by the BBG.
By our choice of $\rblur$, the probability for the ball cut to $z$ is $O\left(\nicefrac{\ell_z\log(k\log n)}{\mathcal{D}}\right)$.
Therefore, we cannot simply use the union bound to sum up the probabilities over all $O(\log n)$ iterations.
Here, the intermediary step where we construct $\kexp(\mathcal{S}_i)$ comes into play.
We can exploit that $\kexp{(\mathcal{S}_i)}$ must be \emph{close} to either endpoint of $z$, i.e., within distance at most $2\rblur$, in order for $\mathsf{blur}(\kexp{(\mathcal{S}_i)}, \rblur)$ to cut it.
The distance between $\Texp{\mathcal{S}_i}$ and $z$ is determined by the exponentially distributed variable $X_i \cdot \mathcal{D}_{BC}/4$.
Using the properties of the truncated exponential distribution, we can show that the probability of 
$\Texp{\mathcal{S}_i}$ being close to $z$ without adding the endpoints of $z$ to $\Texp{\mathcal{S}_i}$ is constant for our choice of parameters.
Thus, the edge is safe from ever being cut before the blur procedure is even executed.
This follows from the properties of the truncated exponential distribution.
Thus, on expectation, there will only be a constant number of tries before $z$ is either cut or safe.
As we will see in the full analysis, this is sufficient for our probability bound.

Finally, verifying that each algorithm step requires $\Tilde{O}(k)$ approximate \SetSSP computations and minor aggregations is easy. 
We use three subroutines from \autoref{thm:clustering_general}, \autoref{thm:distributed_weak_separator}, and \autoref{lemma:bbg}, which all require at most $\Tilde{O}(k)$ approximate \SetSSP computations and minor aggregations on a $k$-path separable graph. 
In addition, the algorithm only uses one more \SetSSP; thus, the runtime follows.

\noindent  A more detailed description and the full proof can be found in \autoref{sec:appendix_backbone}.

\subsection{The Refinement Phase}
\label{sec:refinement}

For the second phase, we show a generic lemma that allows us to turn any backbone clustering into an \LDC.
To be precise, it holds:
\begin{restatable}{lemma}{refinement}
\label{lemma:backbone_clustering}
Suppose we have an algorithm $\mathcal{A}$ that creates a $(\alpha,\beta,\kappa)$-backbone clustering with pseudo-diameter $\mathcal{D}_{BC}$ of a weighted graph $G$. 
 Then, there is an algorithm that creates a series of disjoint clusters $\mathcal{K} := K_1, K_2, \ldots, K_{N}$ with $K_i:= (V_i,E_i)$.
    Futher, it holds:
    \begin{enumerate}
        \item  Each cluster $K_i \in \mathcal{K}$ has a strong diameter of at most $16\cdot\mathcal{D}_{BC}$.
        \item  Each edge $z := \{v,w\} \in E$ of length $\ell_e$ is cut between clusters with probability $O\left(\frac{\ell_z\cdot (\alpha+\log\kappa\beta)}{\mathcal{D}}\right)$.
        \item Each node $v \in V$ is added to some cluster $K_i \in \mathcal{K}$ with probability at least $\nicefrac{1}{2}$.
\end{enumerate}
The algorithm requires one application of $\mathcal{A}$, $\Tilde{O}(\kappa)$ minor aggregations, and $\Tilde{O(1)}$ approximate \SetSSP computations with $\epsilon \in O(\nicefrac{1}{\log^2 n})$.
\end{restatable}

The algorithm begins a constructing backbone clustering $K(B_1), \ldots, K(B_N)$ using algorith $\mathcal{A}$.
The following steps are executed in parallel in each cluster $K(\mathcal{B}_i)$:
On each path $P \in \mathcal{B}_i$ in the backbone, mark a subset of nodes in distance (at most) $\mathcal{D}_{BC}$ to each other and use them as potential cluster centers.
This can be done with one shortest path computation and two aggregations on each path:
First, we let each node compute the distance to the first node of the path $v_1 \in P$.
A simple \SetSSP computation can do this; as there is only one path, the algorithm returns an exact result.
Then, we let each node locally compute its \emph{distance class} $\lceil \frac{d_P(v,v_1)}{\mathcal{D}_{BC}}\rceil$.
Finally, we mark the first node in each distance class.
A node can determine whether it is the first node by aggregating the distance class of both its neighbors in two aggregations.
Having marked the nodes, we use them as centers $\mathcal{X}_i$ in the algorithm of \autoref{thm:genericldd} and compute an \LDC in $K(\mathcal{B}_i)$.
To determine the gurantees of this clustering, we need to determine the \emph{covering} and \emph{packing} properties of these centers.
Each node $v \in K(\mathcal{B}_i)$ has at least one marked node in distance $2\mathcal{D}_{BC}$ by construction:
It must have a one path $P_j \in \mathcal{B}_i$ in distance $\mathcal{D}_{BC}$ and the next marked nodes on the path can also only be in distance $\mathcal{D}_{BC}$ to the node closest to $v$.
Further, as a path of length $\beta\mathcal{D}_{BC}$ has $\beta$ distance classes, we mark at most $\beta$ nodes per path.
As we limit the number of paths in $\mathcal{B}_i$ to $\kappa$ and the length of each path to $\beta\mathcal{D}_{BC}$, the total number of  centers is limited to $\beta \cdot \kappa$.
Thus, each node in $K(\mathcal{B_i})$ is covered by one center in distance $2\mathcal{D}_{BC}$ and \emph{packed} by at most $\beta \cdot \kappa$ centers in distance $6\cdot(2\mathcal{D}_{BC})$.
Following \autoref{thm:genericldd}, this gives us an LDC with strong diameter $8\cdot (2\mathcal{D}_{BC})$ and quality $O(\log \beta\kappa)$.
Finally recall that an edge can either be cut by the initial construction of the backbone clustering or by this procedure.
By the union bound, an edge is cut with probability $O(\frac{\alpha + \log \beta\kappa}{\mathcal{D}})$ by either stage.
As the algorithm consists of computing a backbone clustering, one \SetSSP and one minor aggregations in each path to mark the nodes, and an application of \autoref{thm:genericldd}, the proclaimed runtime follows.

Combing this result with \autoref{lemma:backbone_clustering_k_path}, this creates the desired \LDC with strong diameter $16\mathcal{D}_{BC}$ of quality $O(\log k + \log\log n)$ in each cluster.
In particular, an edge is cut with probability $O(\frac{\ell\cdot\log k\log n}{\mathcal{D}})$ in either of the two phases.
This proves the \autoref{lemma:kr_clustering} when choosing $\mathcal{D}_{BC} = \mathcal{D}/16$.

\medskip

\noindent  A more detailed description and the full proof can be found in \autoref{sec:appendix_refinement}.

\bibliography{bib}

\begin{thebibliography}{}

\end{thebibliography}


\newcommand{\etalchar}[1]{$^{#1}$}
\begin{thebibliography}{MPVX15b}

\bibitem[ABN08]{DBLP:conf/focs/AbrahamBN08}
Ittai Abraham, Yair Bartal, and Ofer Neiman.
\newblock Nearly tight low stretch spanning trees.
\newblock In {\em 49th Annual {IEEE} Symposium on Foundations of Computer Science, {FOCS} 2008, October 25-28, 2008, Philadelphia, PA, {USA}}, pages 781--790. {IEEE} Computer Society, 2008.

\bibitem[AG06]{DBLP:conf/podc/AbrahamG06}
Ittai Abraham and Cyril Gavoille.
\newblock Object location using path separators.
\newblock In Eric Ruppert and Dahlia Malkhi, editors, {\em Proceedings of the Twenty-Fifth Annual {ACM} Symposium on Principles of Distributed Computing, {PODC} 2006, Denver, CO, USA, July 23-26, 2006}, pages 188--197. {ACM}, 2006.

\bibitem[AGG{\etalchar{+}}14]{DBLP:conf/stoc/AbrahamGGNT14}
Ittai Abraham, Cyril Gavoille, Anupam Gupta, Ofer Neiman, and Kunal Talwar.
\newblock Cops, robbers, and threatening skeletons: padded decomposition for minor-free graphs.
\newblock In David~B. Shmoys, editor, {\em Symposium on Theory of Computing, {STOC} 2014, New York, NY, USA, May 31 - June 03, 2014}, pages 79--88. {ACM}, 2014.

\bibitem[AGG{\etalchar{+}}19]{AGGNT19}
Ittai Abraham, Cyril Gavoille, Anupam Gupta, Ofer Neiman, and Kunal Talwar.
\newblock Cops, robbers, and threatening skeletons: Padded decomposition for minor-free graphs.
\newblock {\em {SIAM} J. Comput.}, 48(3):1120--1145, 2019.

\bibitem[AGMW10]{DBLP:journals/mst/AbrahamGMW10}
Ittai Abraham, Cyril Gavoille, Dahlia Malkhi, and Udi Wieder.
\newblock Strong-diameter decompositions of minor free graphs.
\newblock {\em Theory Comput. Syst.}, 47(4):837--855, 2010.

\bibitem[AHK{\etalchar{+}}20]{DBLP:conf/soda/AugustineHKSS20}
John Augustine, Kristian Hinnenthal, Fabian Kuhn, Christian Scheideler, and Philipp Schneider.
\newblock Shortest paths in a hybrid network model.
\newblock In {\em Proceedings of the 2020 {ACM-SIAM} Symposium on Discrete Algorithms, {SODA} 2020, Salt Lake City, UT, USA, January 5-8, 2020}, pages 1280--1299. {SIAM}, 2020.

\bibitem[ALH{\etalchar{+}}23]{DBLP:journals/dc/AnagnostidesLHZG23}
Ioannis Anagnostides, Christoph Lenzen, Bernhard Haeupler, Goran Zuzic, and Themis Gouleakis.
\newblock Almost universally optimal distributed laplacian solvers via low-congestion shortcuts.
\newblock {\em Distributed Comput.}, 36(4):475--499, 2023.

\bibitem[AN19]{AN19}
Ittai Abraham and Ofer Neiman.
\newblock Using petal-decompositions to build a low stretch spanning tree.
\newblock {\em {SIAM} J. Comput.}, 48(2):227--248, 2019.

\bibitem[AP90]{AP90}
Baruch Awerbuch and David Peleg.
\newblock Sparse partitions (extended abstract).
\newblock In {\em 31st Annual Symposium on Foundations of Computer Science, St. Louis, Missouri, USA, October 22-24, 1990, Volume {II}}, pages 503--513. {IEEE} Computer Society, 1990.

\bibitem[Bar96]{Bar96}
Yair Bartal.
\newblock Probabilistic approximations of metric spaces and its algorithmic applications.
\newblock In {\em 37th Annual Symposium on Foundations of Computer Science, {FOCS} '96, Burlington, Vermont, USA, 14-16 October, 1996}, pages 184--193. {IEEE} Computer Society, 1996.

\bibitem[Bar98]{DBLP:conf/stoc/Bartal98}
Yair Bartal.
\newblock On approximating arbitrary metrices by tree metrics.
\newblock In Jeffrey~Scott Vitter, editor, {\em Proceedings of the Thirtieth Annual {ACM} Symposium on the Theory of Computing, Dallas, Texas, USA, May 23-26, 1998}, pages 161--168. {ACM}, 1998.

\bibitem[Bar04]{DBLP:conf/esa/Bartal04}
Yair Bartal.
\newblock Graph decomposition lemmas and their role in metric embedding methods.
\newblock In Susanne Albers and Tomasz Radzik, editors, {\em Algorithms - {ESA} 2004, 12th Annual European Symposium, Bergen, Norway, September 14-17, 2004, Proceedings}, volume 3221 of {\em Lecture Notes in Computer Science}, pages 89--97. Springer, 2004.

\bibitem[BEGL19]{DBLP:conf/wdag/BeckerEGL19}
Ruben Becker, Yuval Emek, Mohsen Ghaffari, and Christoph Lenzen.
\newblock Distributed algorithms for low stretch spanning trees.
\newblock In Jukka Suomela, editor, {\em 33rd International Symposium on Distributed Computing, {DISC} 2019, October 14-18, 2019, Budapest, Hungary}, volume 146 of {\em LIPIcs}, pages 4:1--4:14. Schloss Dagstuhl - Leibniz-Zentrum f{\"{u}}r Informatik, 2019.

\bibitem[BEL20]{BeckerEL20}
Ruben Becker, Yuval Emek, and Christoph Lenzen.
\newblock {Low Diameter Graph Decompositions by Approximate Distance Computation}.
\newblock In Thomas Vidick, editor, {\em 11th Innovations in Theoretical Computer Science Conference (ITCS 2020)}, volume 151 of {\em Leibniz International Proceedings in Informatics (LIPIcs)}, pages 50:1--50:29, Dagstuhl, Germany, 2020. Schloss Dagstuhl--Leibniz-Zentrum fuer Informatik.

\bibitem[BGK{\etalchar{+}}11]{DBLP:conf/spaa/BellochGKPT11}
Guy~E. Blelloch, Anupam Gupta, Ioannis Koutis, Gary~L. Miller, Richard Peng, and Kanat Tangwongsan.
\newblock Near linear-work parallel sdd solvers, low-diameter decomposition, and low-stretch subgraphs.
\newblock SPAA '11, page 13–22, New York, NY, USA, 2011. Association for Computing Machinery.

\bibitem[BS07]{BS07}
Surender Baswana and Sandeep Sen.
\newblock A simple and linear time randomized algorithm for computing sparse spanners in weighted graphs.
\newblock {\em Random Structures \& Algorithms}, 30(4):532--563, 2007.

\bibitem[Cha23]{DBLP:conf/podc/Chang23}
Yi{-}Jun Chang.
\newblock Efficient distributed decomposition and routing algorithms in minor-free networks and their applications.
\newblock In Rotem Oshman, Alexandre Nolin, Magn{\'{u}}s~M. Halld{\'{o}}rsson, and Alkida Balliu, editors, {\em Proceedings of the 2023 {ACM} Symposium on Principles of Distributed Computing, {PODC} 2023, Orlando, FL, USA, June 19-23, 2023}, pages 55--66. {ACM}, 2023.

\bibitem[CS22]{DBLP:conf/podc/ChangS22}
Yi{-}Jun Chang and Hsin{-}Hao Su.
\newblock Narrowing the {LOCAL-CONGEST} gaps in sparse networks via expander decompositions.
\newblock In Alessia Milani and Philipp Woelfel, editors, {\em {PODC} '22: {ACM} Symposium on Principles of Distributed Computing, Salerno, Italy, July 25 - 29, 2022}, pages 301--312. {ACM}, 2022.

\bibitem[DG10]{DBLP:conf/faw/DiotG10}
Emilie Diot and Cyril Gavoille.
\newblock Path separability of graphs.
\newblock In Der{-}Tsai Lee, Danny~Z. Chen, and Shi Ying, editors, {\em Frontiers in Algorithmics, 4th International Workshop, {FAW} 2010, Wuhan, China, August 11-13, 2010. Proceedings}, volume 6213 of {\em Lecture Notes in Computer Science}, pages 262--273. Springer, 2010.

\bibitem[dV23]{DBLP:conf/sirocco/Vos23}
Tijn de~Vos.
\newblock Minimum cost flow in the {CONGEST} model.
\newblock In Sergio Rajsbaum, Alkida Balliu, Joshua~J. Daymude, and Dennis Olivetti, editors, {\em Structural Information and Communication Complexity - 30th International Colloquium, {SIROCCO} 2023, Alcal{\'{a}} de Henares, Spain, June 6-9, 2023, Proceedings}, volume 13892 of {\em Lecture Notes in Computer Science}, pages 406--426. Springer, 2023.

\bibitem[EEST08]{EEST08}
Michael Elkin, Yuval Emek, Daniel~A. Spielman, and Shang{-}Hua Teng.
\newblock Lower-stretch spanning trees.
\newblock {\em {SIAM} J. Comput.}, 38(2):608--628, 2008.

\bibitem[EFN20]{DBLP:conf/podc/ElkinFN20}
Michael Elkin, Arnold Filtser, and Ofer Neiman.
\newblock Distributed construction of light networks.
\newblock In {\em {PODC} '20: {ACM} Symposium on Principles of Distributed Computing, Virtual Event, Italy, August 3-7, 2020}, pages 483--492. {ACM}, 2020.

\bibitem[EN22]{DBLP:journals/tcs/ElkinN22}
Michael Elkin and Ofer Neiman.
\newblock Distributed strong diameter network decomposition.
\newblock {\em Theor. Comput. Sci.}, 922:150--157, 2022.

\bibitem[ENS15]{DBLP:journals/siamdm/ElkinNS15}
Michael Elkin, Ofer Neiman, and Shay Solomon.
\newblock Light spanners.
\newblock {\em {SIAM} J. Discret. Math.}, 29(3):1312--1321, 2015.

\bibitem[Fil19]{Filtser19}
Arnold Filtser.
\newblock {On Strong Diameter Padded Decompositions}.
\newblock In Dimitris Achlioptas and L{\'a}szl{\'o}~A. V{\'e}gh, editors, {\em Approximation, Randomization, and Combinatorial Optimization. Algorithms and Techniques (APPROX/RANDOM 2019)}, volume 145 of {\em Leibniz International Proceedings in Informatics (LIPIcs)}, pages 6:1--6:21, Dagstuhl, Germany, 2019. Schloss Dagstuhl--Leibniz-Zentrum fuer Informatik.

\bibitem[FN22]{DBLP:journals/algorithmica/FiltserN22}
Arnold Filtser and Ofer Neiman.
\newblock Light spanners for high dimensional norms via stochastic decompositions.
\newblock {\em Algorithmica}, 84(10):2987--3007, 2022.

\bibitem[FRT04]{DBLP:journals/jcss/FakcharoenpholRT04}
Jittat Fakcharoenphol, Satish Rao, and Kunal Talwar.
\newblock A tight bound on approximating arbitrary metrics by tree metrics.
\newblock {\em J. Comput. Syst. Sci.}, 69(3):485--497, 2004.

\bibitem[FT03]{FT03}
Jittat Fakcharoenphol and Kunal Talwar.
\newblock An {Improved} {Decomposition} {Theorem} for {Graphs} {Excluding} a {Fixed} {Minor}.
\newblock In Gerhard Goos, Juris Hartmanis, Jan van Leeuwen, Sanjeev Arora, Klaus Jansen, José D.~P. Rolim, and Amit Sahai, editors, {\em Approximation, {Randomization}, and {Combinatorial} {Optimization}.. {Algorithms} and {Techniques}}, volume 2764, pages 36--46. Springer Berlin Heidelberg, Berlin, Heidelberg, 2003.
\newblock Series Title: Lecture Notes in Computer Science.

\bibitem[GGH{\etalchar{+}}23]{DBLP:conf/soda/0001GHIR23}
Mohsen Ghaffari, Christoph Grunau, Bernhard Haeupler, Saeed Ilchi, and V{\'{a}}clav Rozhon.
\newblock Improved distributed network decomposition, hitting sets, and spanners, via derandomization.
\newblock In Nikhil Bansal and Viswanath Nagarajan, editors, {\em Proceedings of the 2023 {ACM-SIAM} Symposium on Discrete Algorithms, {SODA} 2023, Florence, Italy, January 22-25, 2023}, pages 2532--2566. {SIAM}, 2023.

\bibitem[GH16a]{DBLP:conf/podc/GhaffariH16}
Mohsen Ghaffari and Bernhard Haeupler.
\newblock Distributed algorithms for planar networks {I:} planar embedding.
\newblock In George Giakkoupis, editor, {\em Proceedings of the 2016 {ACM} Symposium on Principles of Distributed Computing, {PODC} 2016, Chicago, IL, USA, July 25-28, 2016}, pages 29--38. {ACM}, 2016.

\bibitem[GH16b]{DBLP:conf/soda/GhaffariH16}
Mohsen Ghaffari and Bernhard Haeupler.
\newblock Distributed algorithms for planar networks {II:} low-congestion shortcuts, mst, and min-cut.
\newblock In Robert Krauthgamer, editor, {\em Proceedings of the Twenty-Seventh Annual {ACM-SIAM} Symposium on Discrete Algorithms, {SODA} 2016, Arlington, VA, USA, January 10-12, 2016}, pages 202--219. {SIAM}, 2016.

\bibitem[GH21]{DBLP:conf/podc/GhaffariH21}
Mohsen Ghaffari and Bernhard Haeupler.
\newblock Low-congestion shortcuts for graphs excluding dense minors.
\newblock In Avery Miller, Keren Censor{-}Hillel, and Janne~H. Korhonen, editors, {\em {PODC} '21: {ACM} Symposium on Principles of Distributed Computing, Virtual Event, Italy, July 26-30, 2021}, pages 213--221. {ACM}, 2021.

\bibitem[GKK{\etalchar{+}}18]{DBLP:journals/siamcomp/GhaffariKKLP18}
Mohsen Ghaffari, Andreas Karrenbauer, Fabian Kuhn, Christoph Lenzen, and Boaz Patt{-}Shamir.
\newblock Near-optimal distributed maximum flow.
\newblock {\em {SIAM} J. Comput.}, 47(6):2078--2117, 2018.

\bibitem[GP17]{DBLP:conf/wdag/GhaffariP17}
Mohsen Ghaffari and Merav Parter.
\newblock Near-optimal distributed {DFS} in planar graphs.
\newblock In Andr{\'{e}}a~W. Richa, editor, {\em 31st International Symposium on Distributed Computing, {DISC} 2017, October 16-20, 2017, Vienna, Austria}, volume~91 of {\em LIPIcs}, pages 21:1--21:16. Schloss Dagstuhl - Leibniz-Zentrum f{\"{u}}r Informatik, 2017.

\bibitem[GZ22a]{ghaffari_universally-optimal_2022}
Mohsen Ghaffari and Goran Zuzic.
\newblock Universally-{Optimal} {Distributed} {Exact} {Min}-{Cut}.
\newblock In {\em Proceedings of the 2022 {ACM} {Symposium} on {Principles} of {Distributed} {Computing}}, pages 281--291, Salerno Italy, July 2022. ACM.

\bibitem[GZ22b]{GZ22}
Mohsen Ghaffari and Goran Zuzic.
\newblock Universally-{Optimal} {Distributed} {Exact} {Min}-{Cut}.
\newblock In {\em Proceedings of the 2022 {ACM} {Symposium} on {Principles} of {Distributed} {Computing}}, pages 281--291, Salerno Italy, July 2022. ACM.

\bibitem[HWZ21]{DBLP:conf/stoc/HaeuplerWZ21}
Bernhard Haeupler, David Wajc, and Goran Zuzic.
\newblock Universally-optimal distributed algorithms for known topologies.
\newblock In Samir Khuller and Virginia~Vassilevska Williams, editors, {\em {STOC} '21: 53rd Annual {ACM} {SIGACT} Symposium on Theory of Computing, Virtual Event, Italy, June 21-25, 2021}, pages 1166--1179. {ACM}, 2021.

\bibitem[IKNS22]{DBLP:conf/spaa/IzumiKNS22}
Taisuke Izumi, Naoki Kitamura, Takamasa Naruse, and Gregory Schwartzman.
\newblock Fully polynomial-time distributed computation in low-treewidth graphs.
\newblock In Kunal Agrawal and I{-}Ting~Angelina Lee, editors, {\em {SPAA} '22: 34th {ACM} Symposium on Parallelism in Algorithms and Architectures, Philadelphia, PA, USA, July 11 - 14, 2022}, pages 11--22. {ACM}, 2022.

\bibitem[KPR93]{KPR93}
Philip~N. Klein, Serge~A. Plotkin, and Satish Rao.
\newblock Excluded minors, network decomposition, and multicommodity flow.
\newblock In S.~Rao Kosaraju, David~S. Johnson, and Alok Aggarwal, editors, {\em Proceedings of the Twenty-Fifth Annual {ACM} Symposium on Theory of Computing, May 16-18, 1993, San Diego, CA, {USA}}, pages 682--690. {ACM}, 1993.

\bibitem[LMR21]{DBLP:journals/dc/LeviMR21}
Reut Levi, Moti Medina, and Dana Ron.
\newblock Property testing of planarity in the {CONGEST} model.
\newblock {\em Distributed Comput.}, 34(1):15--32, 2021.

\bibitem[LP19]{DBLP:conf/stoc/LiP19}
Jason Li and Merav Parter.
\newblock Planar diameter via metric compression.
\newblock In Moses Charikar and Edith Cohen, editors, {\em Proceedings of the 51st Annual {ACM} {SIGACT} Symposium on Theory of Computing, {STOC} 2019, Phoenix, AZ, USA, June 23-26, 2019}, pages 152--163. {ACM}, 2019.

\bibitem[MPVX15a]{MPV+15}
Gary~L. Miller, Richard Peng, Adrian Vladu, and Shen~Chen Xu.
\newblock Improved parallel algorithms for spanners and hopsets.
\newblock In {\em Proc. of the 27th ACM symposium on Parallelism in Algorithms and Architectures (SPAA)}, pages 192--201, 2015.

\bibitem[MPVX15b]{DBLP:conf/spaa/MillerPVX15}
Gary~L. Miller, Richard Peng, Adrian Vladu, and Shen~Chen Xu.
\newblock Improved parallel algorithms for spanners and hopsets.
\newblock In {\em Proceedings of the 27th {ACM} on Symposium on Parallelism in Algorithms and Architectures, {SPAA} 2015, Portland, OR, USA, June 13-15, 2015}, pages 192--201. {ACM}, 2015.

\bibitem[MPX13]{DBLP:conf/spaa/MillerPX13}
Gary~L. Miller, Richard Peng, and Shen~Chen Xu.
\newblock Parallel graph decompositions using random shifts.
\newblock In Guy~E. Blelloch and Berthold V{\"{o}}cking, editors, {\em 25th {ACM} Symposium on Parallelism in Algorithms and Architectures, {SPAA} '13, Montreal, QC, Canada - July 23 - 25, 2013}, pages 196--203. {ACM}, 2013.

\bibitem[REGH22]{DBLP:conf/focs/RozhonEGH22}
V{\'{a}}clav Rozhon, Michael Elkin, Christoph Grunau, and Bernhard Haeupler.
\newblock Deterministic low-diameter decompositions for weighted graphs and distributed and parallel applications.
\newblock In {\em 63rd {IEEE} Annual Symposium on Foundations of Computer Science, {FOCS} 2022, Denver, CO, USA, October 31 - November 3, 2022}, pages 1114--1121. {IEEE}, 2022.

\bibitem[RGH{\etalchar{+}}22]{DBLP:conf/stoc/RozhonGHZL22}
V{\'{a}}clav Rozhon, Christoph Grunau, Bernhard Haeupler, Goran Zuzic, and Jason Li.
\newblock Undirected (1+\emph{{\(\epsilon\)}})-shortest paths via minor-aggregates: near-optimal deterministic parallel and distributed algorithms.
\newblock In Stefano Leonardi and Anupam Gupta, editors, {\em {STOC} '22: 54th Annual {ACM} {SIGACT} Symposium on Theory of Computing, Rome, Italy, June 20 - 24, 2022}, pages 478--487. {ACM}, 2022.

\bibitem[RS86]{DBLP:journals/jal/RobertsonS86}
Neil Robertson and Paul~D. Seymour.
\newblock Graph minors. {II.} algorithmic aspects of tree-width.
\newblock {\em J. Algorithms}, 7(3):309--322, 1986.

\bibitem[Sch00]{ScheidelerHabil}
Christian Scheideler.
\newblock {\em Probabilistic Methods for Coordination Problems}.
\newblock Habilitation, Universit{\"a}t Paderborn, Heinz Nixdorf Institut, Theoretische Informatik, 2000.
\newblock ISBN 3-931466-77-9.

\bibitem[Sch23]{DBLP:phd/dnb/Schneider23}
Philipp Schneider.
\newblock {\em Power and limitations of hybrid communication networks}.
\newblock PhD thesis, University of Freiburg, Freiburg im Breisgau, Germany, 2023.

\bibitem[Tho04]{Thorup04}
Mikkel Thorup.
\newblock Compact oracles for reachability and approximate distances in planar digraphs.
\newblock {\em J. ACM}, 51(6):993–1024, nov 2004.

\bibitem[TZ01]{TZ01}
Mikkel Thorup and Uri Zwick.
\newblock Compact routing schemes.
\newblock In Arnold~L. Rosenberg, editor, {\em Proceedings of the Thirteenth Annual {ACM} Symposium on Parallel Algorithms and Architectures, {SPAA} 2001, Heraklion, Crete Island, Greece, July 4-6, 2001}, pages 1--10. {ACM}, 2001.

\end{thebibliography}

\newpage

\appendix

\section{More Related Work}
\label{sec:relatedwork}

In the following, we review the related work for low-diameter decompositions in both sequential and distributed models.
First, we note that there are many different notions of decompositions in the literature, and the research is not restricted to the low-diameter decompositions that follow Definition \ref{def:ldd}.
Instead, they come in many shapes and forms with different guarantees for the clusters they produce.
Nevertheless, different types of decompositions are often related to each other and/or their computations use similar techniques.
In order to give a good overview on how our algorithms perform, we will give a short introduction to some commonly used notions of decompositions.
We begin with a generic definition that provides the basis for all forthcoming types of decomposition.
A decomposition of a weighted graph $G := (V,E,\ell)$ with distance parameter $\mathcal{D}$ is a series of subgraphs $K_1,K_2, \ldots$ of $G$, which we will call clusters.
Each node $v \in V$ is contained in (at least) one of these clusters.
Further, each cluster has a diameter of $\mathcal{D}$. 
More precisely, we say that the cluster has a \emph{strong} diameter $\mathcal{D}$ if the distance within the cluster is $\mathcal{D}$.
In this case, the path between the nodes only uses edges where both endpoints are contained in the cluster.
Otherwise, if the distance between two nodes is only $\mathcal{D}$ if we are allowed to consider all edges of $G$, we say that the cluster has a \emph{weak} diameter $\mathcal{D}$. 
Note that a cluster of weak diameter $\mathcal{D}$ is not necessarily connected.
Given this underlying definition, we now present three types of decompositions.
First, there is the deterministic counterpart to our probabilistic \LDD's.
Here, we do not bound the probability that a specific edge of length $\ell$ is cut, but instead --- since it is deterministic --- we count the overall number of edges of length $\ell$ that are cut.
It holds:
\begin{definition}[Deterministic Low-Diameter Decomposition]
A deterministic low-diameter decomposition with diameter $\mathcal{D}$ of a weighted graph $G := (V,E,\ell)$ with $m$ edges is a partition of $G$ into subgraphs $K_1, K_2, \ldots$ of diameter $\mathcal{D}$, where each node is contained in exactly one cluster.
We say that decomposition has quality $\alpha$, the number of edges of length $\ell$ with endpoint in different clusters is bound by $m\cdot\frac{\alpha\cdot\ell}{\mathcal{D}}$.
\end{definition}
Note that $m\cdot\frac{\alpha\cdot\ell}{\mathcal{D}}$ is exactly the \emph{expected} number of cut edges in probabilistic \LDD of quality $\alpha$. 
Thus, with constant probability, the probabilistic version roughly cuts the same number of edges.
However, without further information about the specific random choices made by the corresponding algorithm, a probabilistic \LDD could cut many more edges.
While this can be mitigated with standard probability amplification techniques, i.e., executing the algorithm $O(\log n)$ and picking the iteration that cute the fewest edges, a deterministic \LDD does not have this problem in the first place.
In a deterministic \LDD, we have the guarantee that no matter what happens in the execution of the algorithm, less than $m\cdot\frac{\alpha\cdot\ell}{\mathcal{D}}$ edges of length $\ell$ are cut.
This makes these decompositions very useful when we do not have the time or resources for $O(\log n)$ repetitions.
A typical example is distributed algorithms for local problems, where we aim for sublogarithmic runtimes.
As a final remark before getting to the next type of decomposition, in unweighted graphs, some authors use slight notation when describing \LDD's.
While the definition above focuses on the clusters' diameter, some works emphasize the number of cut edges.
That is, they define it as a decomposition that cuts $\epsilon\cdot m$ edges and creates a cluster of diameter $\mathcal{D} := O(\epsilon^{-1})$.

The next class of decompositions we consider in this section are so-called \emph{padded} decompositions.
While an \LDD only gives guarantees for single edges and node pairs, in a padded decomposition, we have a guarantee that \textbf{all} nodes in a certain distance are contained in the same cluster.
This makes them more versatile in many scenarios.
Formally, they are defined as follows:
\begin{definition}[Padded Decomposition]
A $(\beta,\gamma_{\mathsf{max}})$-padded decomposition with diameter $\mathcal{D}$ of a weighted graph $G := (V,E,\ell)$ is a partition of $G$ into subgraphs $K_1, K_2, \ldots$ of diameter $\mathcal{D}$, where each node is contained in exactly one cluster.
For each $v \in V$, let $K(v)$ denote the cluster that contains $v$.
Then, for each $\gamma \leq \gamma_{\mathsf{max}}$, it holds
\begin{align}
    \pr{B(v,\gamma\mathcal{D}) \subseteq K(v)} \geqslant e^{-\beta\gamma}.
\end{align}
\end{definition}
Typically, the parameter $\gamma_{\mathsf{max}}$ is in the magnitude of $\nicefrac{C}{\beta}$ for some $C \geq 1$, so we get a guarantee for balls of diameter $\nicefrac{C\cdot \mathcal{D}}{\beta}$ or smaller.
While other values of $\gamma_{\mathsf{max}}$ are not forbidden by the definition, we are not aware of any work where $\gamma_{\mathsf{max}}$ is smaller than $O(\nicefrac{1}{\beta})$.
With this in mind, note that the guarantees of a padded decomposition are provably stronger than the guarantees of a \LDD, i.e., a $(\beta, \nicefrac{C}{\beta})$-padded decomposition implies a (probabilistic) \LDD of quality $\alpha$.
To verify this, consider an edge $(v,w) \in E$ of length $\ell = \gamma\mathcal{D}$ and suppose we construct a padded decomposition. 
By definition of the padded decomposition, the ball $B(v,\gamma\mathcal{D})$ is fully contained in $K(v)$ with probability $e^{-\beta\gamma} \approx 1-\beta\gamma = 1-\frac{\beta\cdot\ell}{\mathcal{D}}$.
In this case, the edge is not cut and thus, any padded decomposition is also a low-diameter decomposition. 

The final type of decomposition that we introduce is \emph{so-called} neighborhood or sparse covers (both names appear in the literature). 
As with padded decomposition, the goal is to create clusters, s.t., for every node $v \in V$ there is a cluster that contains the full ball $B(\gamma\mathcal{D})$.
In contrast to the two previous concepts, a node can now be in more than one cluster to achieve  this goal.
However, the number of clusters that contain a given node should be kept small.
Formally, they are defined as follows.
\begin{definition}[Neighborhood/Sparse Cover]
A $(\gamma,s)$-neighborhood cover with diameter $\mathcal{D}$ of a weighted graph $G := (V,E,\ell)$ is a partition of $G$ into subgraphs $K_1, K_2, \ldots$ of diameter $\mathcal{D}$, where each node is contained in at most $s$ clusters.
Then, for each node $v \in V$ there is at least one cluster $K$ that contains $B(v,\gamma\mathcal{D})$.
\end{definition}
The paramter $s$ is sometimes called the \emph{degree} of the cover, while $(\nicefrac{1}{\gamma})$ is called the \emph{diameter blowup}.
It is easy to establish a connection between padded decompositions and neighborhood covers.
Suppose that we have $(\beta,\nicefrac{C}{\beta})$-padded decomposition for parameters $\beta, C \geq 1$.
Then, w.h.p., for every $c \leq C$ we can construct a $(\nicefrac{c}{\beta},O(e^c\cdot\log n))$-neighborhood cover through $O(e^c\cdot\log n))$ executions of the padded decomposition algorithm.
If we choose $\gamma = \nicefrac{c}{\beta}$, the ball $B(v,\gamma\mathcal{D})$ is contained in $K(v)$ with probability at least $e^{-c}$ after each execution.
This follows directly from the definition of a padded decomposition.
Thus, after $s := O(e^c \cdot \log n)$ independent executions of the padded decomposition algorithm, each node is in $s$ clusters and, w.h.p., there is one cluster that contains $B(v,\gamma\mathcal{D})$.
Therefore, the difficulty in studying neighborhood covers is finding a construction beats this simple trick.
In particular, for $K_r$-free graphs, there is an emphasis on finding covers with diameter blowup $O(r)$ and sublogarithmic degree.

\subsection{Decompositions in Sequential Models}

In the following, we provide an overview of seminal contributions and developments in low-diameter decompositions with a focus on minor-free graphs and general graphs. The main insights are summarized in Table \ref{tab:overview_stochastic_de}. 
We begin with the results for general graphs.
\paragraph*{\textbf{General Graphs:}} 
There are several algorithms that construct \LDD with the optimal quality $O(\log n)$ for general graphs.
In an often cited paper, Bartal \cite{Bar96} establishes \LDD's with quality $O(\log n)$ in general graphs with $n$ vertices. 
The techniques used in that paper bear the greatest resemblance to our work:
In each step, an arbitrary unclustered node $v_i$ and a (roughly) exponentially distributed diameter $\mathcal{D}_i$ are chosen.
Then, all nodes in the distance $\mathcal{D}_i$ to $v_i$ are added to a cluster.
More precisely, the paper shows that there exists a constant $c$ such that, for any $\mathcal{D} > 0$ and $k \geq 1$, any $n$-vertex graph can be partitioned into partitions of diameter $\mathcal{D}$, such that pairs of vertices with a distance at most $\frac{\mathcal{D}}{ck}$ are clustered together with a probability of at least $n^{-\frac{1}{k}}$.
Bartal \cite{Bar96} further proves a lower bound of $O(\log n)$ for the approximation of the metric of any graph. This can be extend to show an \LDD for a general graph must have a cutting probability of at least $O(\frac{\log n \ell}{\mathcal{D}})$ for every edge of length $\ell$ and clusters with diameter $\mathcal{D}$. 
Thus, Bartal's construction is optimal for general graphs.
As mentioned earlier, there are other types of clustering for general graphs.
In their seminal paper \cite{AP90}, Awerbuch and Peleg present a neighborhood cover for general graphs with stretch $4k-1$ and degree $2kn^{1/k}$.
That is, each node is in $2kn^{1/k}$ clusters of strong diameter $\Delta$ and every node's $\frac{\mathcal{D}}{4k-1}$-neighborhood is in one of these clusters.
As noted by Fitsler in \cite{Filtser19}, the algorithm in \cite{AP90} creates $O(k\cdot n^{1/k}  )$ partitions with strong diameter $\mathcal{D}$, 
Suppose, we sample a single partition from \cite{AP90} uniformly at random.
Then, any edge $(v,w)$ of length $\frac{\mathcal{D}}{4k-1}$ is preserved with probability at least $\Omega\left(\frac{1}{k} \cdot n^{-1/k}\right)$.
Thus, it also is weaker form of the padded decompositon and also weaker than an \LDD as each edge is node's $\frac{\mathcal{D}}{4k-1}$-neighborhood is cut with the same probability. 
Recently, Fitsler \cite{Filtser19} improved this type of decompoition 

\paragraph*{\textbf{$K_r$-free Graphs:}} Now, we move the results for graphs that exclude a fixed minor $K_r$.
Starting chronologically, Klein, Plotkin, and Rao \cite{KPR93} demonstrate that every minor-free graph of the form $K_r$ admits a weak decomposition scheme with padding parameter ${O}\left(r^3\right)$ for all distances $\Omega\left(1\right)$.
This directly implies a weak diameter \LDD with quality ${O}\left(r^3\right)$.
Later, Fakcharoenphol and Talwar \cite{FT03} improved the padding parameter to ${O}\left(r^2\right)$ with weak diameter. 
Finally, Abraham et al.\cite{AGGNT19} prove that $K_r$  minor free graphs admit weak (${O}\left(r\right)$,$\Omega\left(\frac{1}{r}\right)$)-padded decomposition scheme. 
Considering strong diameters, Abraham et al. \cite{DBLP:journals/mst/AbrahamGMW10} present a strong $2^{O(r)}$ \LDD for $K_r$ minor-free graphs. 
The state-of-the-art algorithm was presented by Filtser by building on the result of Abraham, Gavoille, Gupta, Neiman, and Talwar \cite{AGGNT19}.
Filtser showed that this also holds for strong decompositions for distance values smaller than $O(1/r)$.
The currently best algorithm is an algorithm presented in \cite{Filtser19} that achieves $\alpha \in O(r)$.
For graphs of bounded genus $g$ and treewidth $\tau$, \cite{AGGNT19} presents algorithms with padding parameter $O(\log g)$ and $O(\log \tau)$.
These algorithms are conjectured to be optimal.
However, for general $K_r$-free graphs, an algorithm with padding $O(\log r)$ remains elusive.

\begin{table}[ht]
\onehalfspacing
\footnotesize
\centering
\begin{threeparttable}
\begin{tabular}{@{}llllll@{}}
\toprule
\bfseries Ref. & \bfseries Quality & \bfseries Type & \bfseries Diameter &  \bfseries Comment  \\
\midrule
\cite{KPR93} & ${O}\left(r^3\right)$  & Padded & Weak & $K_r$ minor free\\
\cite{FT03} & ${O}\left(r^2\right)$  & Padded & Weak & $K_r$ minor free\\
\cite{AGGNT19} & ${O}\left(r\right)$ & Padded & Weak & $K_r$ minor free\\
\cite{DBLP:journals/mst/AbrahamGMW10}   &$O(e^r)$ & \LDD & Strong\tnote{d} & $K_r$ minor free\\
\cite{AGGNT19} & ${O}\left(r^2\right)$   & Padded & Strong  & $K_r$ minor free\\
\cite{AGGNT19} & ${O}\left(\log \tau\right)$   & Padded & Strong  & Treewidth $\tau$\\
\cite{AGGNT19} & ${O}\left(\log g\right)$   & Padded & Strong  & Genus $g$\\
\cite{Filtser19} & ${O}\left(r\right)$   & Padded & Strong & $K_r$ minor free\\
\midrule
\cite{Bar96} & ${O}\left(\log n\right)$  & Padded & Strong\tnote{d}& General graphs\\
\bottomrule
\end{tabular} 
\end{threeparttable}
\caption{An overview of the related work on decomposition schemes for various graph families in the sequential model. 
 }
\label{tab:overview_stochastic_de}
\end{table}

\subsection{Decompositions in \CONGEST}

In the following, we provide an overview of seminal contributions and developments in low-diameter decompositions with a focus on minor-free graphs and general graphs. The main insights are summarized in Table \ref{tab:overview_stochastic_de}. 

% Dou, Götte, Hillebrandt, Scheideler, and Werthmann\cite{DouGHSW23p} introduce an effect technique "padded decomposition" to compute a compact routing scheme problems with low stretch in planar graphs.

\paragraph*{\textbf{General Graphs:}} 
The landscape of \LDD algorithms for general \textbf{weighted} graphs is very comprehensible.
As mentioned before, Becker, Emek and Lenzen \cite{BeckerEL20} construct a low-diameter decomposition using blurry ball growing technique we reused for our construction to build a decomposition of quality $O(\log n)$ of the graph. 
The decomposition is weak but has the additional property that each edge is in at most $O(\log n)$ clusters.
This places it between weak and strong decompositions as they show that for many applications, their notion is sufficient.
Further, Rozhon, Elkin, Grunau and Haeupler \cite{DBLP:conf/focs/RozhonEGH22} make two significant improvements compared to \cite{BeckerEL20}. They present a decomposition with strong diameter (instead of weak), and their construction is deterministic (instead of randomized).
Conversely, their quality is only $O(\log^3 n)$, i.e., their algorithms cut more edges than the algorithm of \cite{BeckerEL20}.
Just as our algorithm, both algorithms only rely on approximate \SetSSP computations with $\epsilon \in O(1/\log^3 n)$.
Thus, they can be implemented on the \CONGEST, \PRAM, and \HYBRID models of computation.

% Abraham, Gavoille, Gupta, Neiman, and Talwar \cite{AGGNT19} demonstrate that for any H-minor-free graph, a low-diameter decomposition with $HD = O(\epsilon^{-1})$ exists, where the hidden constant in $O(\cdot)$ depends solely on $HD$. The optimal inverse linear dependence $HD=O(\epsilon^{-1})$ on $\epsilon$ is evident when considering cycle graphs. For the cluster's diameter, define as a decomposition that cuts $\epsilon\cdot m$ edges and creates a cluster of diameter $\mathcal{D} := O(\epsilon^{-1})$.
\paragraph*{\textbf{$K_r$-free graphs:}}  
We now shift our focus to $K_r$-free graphs.
For the following comparisons, we consider unweighted graphs that exclude $K_r$ as a minor.
Note that our algorithm also works (and was, in fact, designed for) weighted graphs.
Recall that our algorithm always has a complexity of $\tilde{O}(r\cdot \HD)$ for any goal diameter $\mathcal{D}$.
However, in the unweighted case, this could be decreased to $\Tilde{O}(\mathcal{D})$.
As our algorithm works with any SSSP algorithm, we can use a faster one if $\mathcal{D} \leq \HD$.
In the unweighted case, we can simply use the trivial BFS algorithm that computes the distance to all nodes in the distance $r$ in $r$ rounds. For simplicity and easier comparison, we omit the factors that depend on the graph’s degree in the following.

We will now compare our work to three existing \CONGEST algorithms that construct similar clusterings.
First, Levi, Medina, and Ron \cite{DBLP:journals/dc/LeviMR21} design a distributed algorithm for $K_r$-free graphs that computes a \LDD with $\mathcal{D} = \epsilon^{-O(1)}$ that cuts $\epsilon \cdot n$ edges in $\epsilon^{-O(1)} \cdot O(\log n)$ rounds. 
This is not immediately comparable to our bounds, so we must look at the implications of this result.
Recall that each $K_r$-free graph has $O(r\cdot n)$ edges.
Therefore, if we cut $\epsilon \cdot n \in O\left(\frac{\epsilon}{r} \cdot |E| \right)$ edges, we cut a $O(\nicefrac{\epsilon}{r})$-fraction.
The diameter of the cluster is $\mathcal{D} = \epsilon^{-O(1)} = \epsilon^{-c}$.
For simplicity, we assume that $e^{-c} \geq r$ and the constant $c$ hidden in the exponent's $O(1)$ is larger than $3$.
Thus, if measured in terms of $\mathcal{D}$, the term $\frac{\epsilon}{r}$ becomes
\begin{align*}
    \frac{\epsilon}{r} = \frac{\mathcal{D}\epsilon}{\mathcal{D} \cdot r} = \frac{\epsilon^{-c}\epsilon^{c \cdot 1/c}}{\mathcal D \cdot r}= \frac{\epsilon^{-c (1- 1/c)}}{\mathcal D \cdot r} = \frac{\mathcal{D}^{1-1/c}}{\mathcal D \cdot r}  
\end{align*}
By this calculation, we cut a $O\left(\frac{\mathcal{D}^{-1/c}}{r} \cdot |E| \right)$-fraction of the edges.
Now assume that we pick a random edge $z$; the probability of picking an edge that is cut is $O\left(\frac{\mathcal{D}^{-1/c}}{r}\right)$.
Given each edge's length is $\ell_z = 1$, the probability can be expressed as $O\left(\frac{\mathcal{D}^{1-1/c}\ell_z}{\mathcal{D}\cdot r}\right)$.
Thus, we can view the process as \LDD with quality $O(\mathcal{D}^{1-1/c}/r)$, which means that the quality declines with increasing $\mathcal{D}$.
For any $\mathcal{D} \in o(\log\log n)$, this algorithm cuts fewer edges than ours, but for larger values of $\mathcal{D}$, our algorithm prevails.
Similarly, for $\mathcal{D} \in \Tilde{O}(1)$, the algorithm of \cite{DBLP:journals/dc/LeviMR21} is faster, but for larger values, our algorithm is faster.

Improving this result of Levi, Medina, and Ron, Chang and Su \cite{DBLP:conf/podc/ChangS22} show that a low-diameter decomposition with $\mathcal{D} = O(\epsilon^{-1})$ can be computed in $\epsilon^{-O(1)}\log^{O(1)}n$ rounds with high probability and $\epsilon^{-O(1)}2^{O(\sqrt{\log n \log \log n})}$ rounds deterministically in the CONGEST model for any $K_r$-free.
Chang later improved the deterministic runtime to $O(\mathcal{D} \log^*n + \mathcal{D}^5)$.
In a direct comparison, we perform worse as our algorithm would cut a $O(\epsilon\cdot\log\log n)$-fraction of the edges in this scenario. 
However, the comparison of the runtime is more nuanced.
the algorithm of \cite{DBLP:conf/podc/ChangS22} is strictly better for small diameters $\mathcal{D} \in \Tilde{O}(1)$ as it is faster, deterministic, and, as far as we can tell, has no galactic constants. 
On the other hand, our algorithm is asymptotically faster for large diameters (but also cuts more edges). 

The techniques used in \cite{DBLP:conf/podc/ChangS22} differ significantly from ours.
They first perform a so-called expander decomposition to create clusters of high conductance.
To be precise, their clusters have diameter $O(\epsilon^{-1})$ and cut an $\epsilon$-fraction of the edges.
Then, they show that each cluster has a high-degree node that can gather all edges of the cluster.
This node can locally execute a sequential decomposition algorithm and inform the edges if they are cut.
Thus, their quality is good as the best sequential clustering algorithm, which currently stands as $O(r)$ due to \cite{Filtser19}.
Compared to that, the benefit of our approach is that a) it also works on weighted graphs and b) is model-agnostic, so we are not restricted to CONGEST.
This universality is bought, however, by cutting more edges than \cite{DBLP:conf/podc/ChangS22}.

% We went in a completely different direction, using a different structural insight.
% We argue that if we sequentially sample a sufficient number of random paths of length $D$ (and some nodes close to them), after $O(f(r))$ iterations (where $f(r)$ is a large constant that depends on $r$), each remaining node has only a constant fraction of all initial nodes in the distance $D$.
% We then build the clusters "around" these paths.
% To do this, we need to adapt and apply delicate changes to some previously known fully distributed clustering algorithms in addition to the structural insight.
% While we do not reinvent these algorithms, we require many small changes to their analysis to make them work for our purposes.

\begin{table}[ht]
\onehalfspacing
\footnotesize
\centering
\begin{threeparttable}
\begin{tabular}{@{}llllcl@{}}
\toprule
\bfseries Ref. & \bfseries Quality & \bfseries Diameter & \bfseries Runtime & \bfseries Weighted&\bfseries Comment  \\
\midrule
\cite{DBLP:journals/dc/LeviMR21} & ${O}\left(\mathcal{D}^{1-1/c}/r\right)$ & Strong& $O\left(\mathcal{D}^{O(1)}\right)$ &$\times$ & $K_r$-free\\
\cite{DBLP:conf/podc/ChangS22} & ${O}\left(r\right)$ & Strong& $O\left(\mathcal{D}^{O(1)}\right)$ &$\times$ & $K_r$-free\\
Thm. \ref{thm:clustering_k_path} & ${O}\left(\log\log n\right)$ & Strong& $\Tilde{O}(r\cdot\HD)$ &$\checkmark$& $K_r$-free\\
\midrule
\cite{BeckerEL20} & ${O}\left(\log n\right)$  & Weak& $\Tilde{O}(\HD + \sqrt{n}) $&$\checkmark$& General graphs\\
\cite{DBLP:conf/focs/RozhonEGH22} & ${O}\left(\log^3 n\right)$  & Strong& $\Tilde{O}(\HD + \sqrt{n}) $&$\checkmark$& General graphs\\
Thm. \ref{thm:clustering_general} & ${O}\left(\log n\right)$ & Strong& $\Tilde{O}(\HD + \sqrt{n})$ &$\checkmark$& General graphs\\
\bottomrule
\end{tabular} 
\end{threeparttable}
\caption{An overview of the related work on \LDD's for various graph families in the \CONGEST model. 
 }
\label{tab:overview_decomp_congest}
\end{table}

\newpage

\section{Full Analysis of \autoref{theorem:generalpartition} }
\label{sec:padded_appendix}

We now describe the algorithm promised by Theorem \ref{theorem:generalpartition} in more detail.
Note that we will express the algorithm in terms of approximate \SetSSP computations and minor aggregations.
The algorithm works in the following $4$ synchronized steps:
\begin{mdframed}
    \paragraph*{\textbf{(Step 1) Draw random distances:}} In the first step, for each center $x \in \mathcal{X}$, we independently draw a value $\delta_x\in [0,\mathcal{D}]$
    from the truncated exponential distribution with parameter $2+2\log(\tau)$.
    We assume that $\tau$ (or some upper bound thereof) is known to each node.
    Formally, we have
\begin{align*}
    \delta_x \sim \mathcal{D}\cdot\mathsf{Texp}(2+2\log\tau)
\end{align*}
We call $\delta_x$ the offset parameter of center $x$. 
Note that $\mathsf{Texp}(2+2\log\tau)$ returns a continuous random variable.
So, we implicitly round the value to the next value that can be encoded in $c \log n$ bits for some $c \geq 1$ in order to send it around in a single message.
Therefore, $\delta_x$ differs from $\mathcal{D}\cdot\mathsf{Texp}(2+2\log\tau)$ by at most $\nicefrac{\mathcal{D}}{n^c}$.
    \paragraph*{\textbf{(Step 2) Add a (virtual) super-source $s$:}} In next step, we add a virtual node $s$ to $G$ and obtain the graph $G_s := (V \cup \{s\}, E \cup E_s, w_s)$.
    The source $s$ has a weighted virtual edge $\{s,x\} \in E_s$ to each center $x \in \mathcal{X}$ with weight $w_x := (\mathcal{D}- \delta_x)$.
    All other weights and edges remain unchanged.
    \paragraph*{\textbf{(Step 3) Construct an approximate shortest-path tree $T_s$:}} 
    Next, we perform an $(1+\epsilon)$-approximate SSSP algorithm in $G_s$ from node $s$.
    Recall that our $\epsilon > 0$ is the desired error parameter for the pseudo-padded decomposition.
    Let $T$ be the (approximate) shortest-path tree computed by the SSSP algorithm, s.t., it holds:
    $$
        d_T(v,s) \leq (1 + \epsilon)d_{G_s}(v,s)
    $$
    \paragraph*{\textbf{(Step 4) Join Cluster of Closest Center:}}
    Finally, a vertex $v$ joins the cluster of the center $x \in \mathcal{X}$, s.t.,
    \begin{align*}
        x := \arg\min_{x' \in \mathcal{X}}\{\mathcal{D}-\delta_{x'}+d_T(x',v) \}
    \end{align*}
    It is easy to verify that this is the first center on the path from $s$ to $v$.
    Thus, it only remains to inform each node about their clusters identifier.
    This can be done via the $\mathsf{AncestorSum}$ primitive, which acts as a broadcast if only the root has an input value.
    Using this primitive, each center $x\in\mathcal{X}$ broadcasts its identifier and its distance to $s$ to each node in its subtree $T(x)$. 
    Finally, note that this construction ensures that for all pairs $v,w \in C_x$, there is a path in $C_x$ that connects them.
    \end{mdframed}

\subsection{Analysis}
\label{sec:general_clustering_proof}

In this section, we will prove Theorem \ref{theorem:generalpartition}.
To this end, we must show that the algorithm a) can be implemented with \emph{one} approximate set-source \SetSSP computation, b) creates clusters of strong diameter $(1+\epsilon)4\mathcal{D}$, and c) fulfills the padding guarantee from Theorem \ref{theorem:generalpartition}.
We begin by proving the proclaimed complexity, arguably easiest to show.
This can be derived directly from the description of the algorithm.
The algorithm requires only a single approximate shortest-path computation as Steps $1, 2$, and $4$ are purely local operations.

Therefore, it only remains to show the other two properties.
We begin with the strong diameter and show that for all vertices within a cluster, there is a path of length at most $(1+\epsilon)4\mathcal{D}$ \emph{within} the cluster.
Formally, we want to show:
\begin{lemma}[Cluster Diameter]
\label{lemma:strong_diameter}
	Every non-empty cluster $C_x$ created by the algorithm has strong diameter at most $(1+\epsilon)4\mathcal{D}$, i.e., for two vertices $v,w \in C_x$, it holds:
	\begin{align*}
	    d_{C_x}(v,w) \leq 4(1+\epsilon)\mathcal{D}
	\end{align*}
\end{lemma}
For the proof, we first recall how the clusters are built.
Each vertex joins a cluster $C_x$ of some center $x \in \mathcal{X}$ if and only if there is a path in the approximate shortest path tree $T_s$ where $x$ is the last hop before the root $s$.
Thus, each vertex has a path to $x$ that is fully contained in the cluster.
That means, for two vertices $v,w \in C_x$, there is a path from $v$ to $w$ via $x$ that is entirely contained in the cluster as all edges are undirected. 
This implies that:
\begin{align}
    d_{C_x}(v,w) \leq 2\cdot\mathbf{max}_{u \in C_x} d_{T_s}(u,x) \label{eqn:maxx}
\end{align}
Therefore, it remains to bound the maximal distance between a node and its cluster center.
Here, we see that it holds for all $x \in \mathcal{X}$:
\begin{claim}
\label{claim:disttocenter}
    $\mathbf{max}_{u \in C_x} d_{T_s}(u,x) \leq (1+\epsilon)2\mathcal{D}$
\end{claim}
\begin{proof}
Consider node $u \in V$ that maximizes the distance to $x$ in $T_s$.
As $x$ is the penultimate node on the path to $s$ in $T_s$, the distance to $x$ is upper bounded by the distance to $s$.
It holds:
\begin{align}
    d_{T_s}(s,u) &:= w(s,x) +  d_{T_s}(x,u)
    := (\mathcal{D} - \delta_{x}) +  d_{T_s}(x,u)
    \geq d_{T_s}(x,u) \label{ineq:s_geq_x}
\end{align}
The inequality follows because $\delta_{x} \in [0,\mathcal{D}]$ per construction and therefore $(\mathcal{D} - \delta_{x})$ is non-negative.
Further, by the covering property required in Theorem \ref{theorem:generalpartition}, we know that there is at least one center $x' \in \mathcal{X}$ in the distance at most $\mathcal{D}$ to $v$.
Note that the distance between $x'$ and $s$ is also at most $\mathcal{D}$ as $\delta_{x'} \geq 0$ by definition.
Thus, by the triangle inequality, it holds:
\begin{align*}
    d_G(s,u) &\leq d_G(s,x') +  d_G(x',u)
     := (\mathcal{D} - \delta_{x'}) +  d_G(x',u)
    \leq \mathcal{D} + \mathcal{D} = 2\mathcal{D}. \label{ineq:2Delta}
\end{align*}
Combining these two observations with the approximation guarantee from the underlying SSSP gives us the following upper bound on the distance to the cluster center:
\begin{align*}
    d_{T_s}(x,u) &\leq d_{T_s}(s,u) & \rhd \textit{ By Ineq.}\, \eqref{ineq:s_geq_x}\\
    &\leq  (1+\epsilon)d_G(s,u) & \rhd \textit{ As } d_T(s,u) \leq (1+\epsilon)\cdot d(s,u)\\
    &\leq (1+\epsilon)2\mathcal{D}  & \rhd \textit{ By Ineq.}\, \eqref{ineq:2Delta}
\end{align*}
This concludes the proof!    
\end{proof}
Together with Inequality \eqref{eqn:maxx}, this claim directly implies Lemma \ref{lemma:strong_diameter} and our algorithm is guaranteed to create clusters of strong diameter $4(1+\epsilon)\mathcal{D}$.
We continue with the padding property promised by Theorem \ref{theorem:generalpartition}.
To this end, we need to show that all nodes that are \emph{close} to a given vertex $v \in V$ are in the same cluster with some non-trivial probability.
In fact, we will show a slightly stronger statement that we will reuse in later algorithms.
We show:

\begin{restatable}[Padding Property]{lemma}{padding}
\label{lemma:padding_property}
Consider node $v\in V$ clustered by center $x_v \in \mathcal{X}$ and a parameter $\epsilon \leq \gamma\le\frac18$ . 
Further, let $P_{x_v} = (v_1, \ldots, v_\ell)$ be the \textbf{exact} shortest path from $x_v$ to $v$.
Then, for each $v_i \in P_{x_v}$, the ball $B_G(v_i,\gamma \mathcal{D})$ is fully contained in $C_{x_v}$ with probability at least 
\begin{align*}
    \pr{\bigcup_{v_i \in P_{x_v}}B_G(v_i,\gamma \mathcal{D}) \subset C_{x_v}} \geq e^{-16(\gamma+5\epsilon)\log\tau}-\frac{160\log\tau}{n^c}
\end{align*}
\end{restatable}

First, for simpler notation, we introduce the term $d(s,x,v)$, which is a shorthand for the distance between $v$ and $s$ via the center $x \in \mathcal{X}$ in graph $G_s$. 
It holds:
\begin{align*}
    d(s,x,v) := (\mathcal{D}-\delta_x) + d(x,v).
\end{align*}
Note that this definition is \emph{only} based on the properties of the input graph $G$ and the random distances $\mathcal{D}$.
Further, we let $N_v$ be the set of centers $x \in \mathcal{X}$ in the distance at most $6\mathcal{D}$ to node $v \in V$.
Note that these include the centers which can potentially cluster $v$ and for which there is a non-zero probability that $C_{x}$ intersects $B(v, \gamma\mathcal{D})$. 
In other words, these are all the centers of all clusters that (potentially) neighbor the cluster that contains $v$.
This can be verified by the following argument: 
Let $v_i \in P_{x_v}$ an node on the shortest path from $v$ to $x_v$ and let $x \in \mathcal{X}$ be a center that can potentially cluster a node $w \in B(v_i,\gamma\mathcal{D})$.
Any node $v_i$ on the path is in distance at most $(1+\epsilon)2\mathcal{D}$ to $v$;
any node $w \in B(v_i,\gamma\mathcal{D})$ is in distance at most $(\nicefrac{1}{8})\cdot\mathcal{D}$ to $w$ as $\gamma \leq \nicefrac{1}{8}$; and finally every center $x$ that can cover $w$ is in distance $(1+\epsilon)2\mathcal{D}$ to $x$.
Summing all these distances up and using that $\epsilon \leq \gamma \leq \nicefrac{1}{8}$ gives us:
\begin{align*}
   d_G(v,x) &\leq d_G(v,w) + d_G(w,b) + d_G(w,x)\\
   &\leq (1+\epsilon)2\mathcal{D} + \gamma\mathcal{D} + (1+\epsilon)2\mathcal{D}\\
\text{As }\epsilon \leq \gamma \leq \nicefrac{1}{8}:\\
   &\leq (1+\nicefrac{1}{8})2\mathcal{D} + (\nicefrac{1}{8})\mathcal{D} + (1+\nicefrac{1}{8})2\mathcal{D}\\
   &\leq 6\mathcal{D}
\end{align*}
Recall that by the \emph{packing property} in Theorem \ref{theorem:generalpartition}, we have $|N_v|\leq \tau$.
We continue with the definition of the central concept of our proof, the random shift $\Upsilon_i$ for each center $x_i \in N_v$.
It is defined as follows:
\begin{definition}[Random Shift]
\label{def:random_shift}
    Let $v \in V$ be a node and let $N_v := \{x \in \mathcal{X} \mid d(x,v) \leq 6\mathcal{D}\}$.
    For each $x \in N_v$, define the random shift $\Upsilon_v^{(x)}$ as follows:
    \begin{align*}
        \Upsilon_v^{(x)} &:= (\mathcal{D}-\delta_x) + d_{G}(x,v) - \min_{x' \in N_v \setminus\{x\}} \left((\mathcal{D}-\delta_{x'}) + d_{G}(x',v)\right)
    \end{align*}
\end{definition}
These values will help us to quantify how close a center is to node $v \in V$ compared to the other centers.
From the definition of $\Upsilon_v^{(x)}$, it follows that for all $x' \in \mathcal{X}\setminus\{x\}$, it holds:
\begin{claim}
\label{claim:random_shift}
    For $x' \in N_v$, it holds $d_G(v,x,s) \leq d_G(v,x',s) +\Upsilon_v^{(x)}$.
\end{claim}
\begin{proof}
Using the definitions of $d_G(v,x,s)$ and $\Upsilon_v^{(x)}$, we get:
  \begin{align*}
	 d_G(v,x,s) &= (\mathcal{D}-\delta_x) + d(x,v) \\
    &=(\mathcal{D}-\delta_{x'}) + d(x',v) + \left((\mathcal{D}-\delta_x) + d(x,v) - \left((\mathcal{D}-\delta_{x'}) + d(x',v)\right)\right)\\
    &\leq d_G(v,x',s) + \left((\mathcal{D}-\delta_x) + d(x,v) - \min_{x'' \in \mathcal{X}\setminus\{x\} }\left\{ (\mathcal{D}-\delta_{x''}) + d(x'',v)\right\}\right)\\
    &:= d_G(v,x',s) +\Upsilon_v^{(x)}
\end{align*} 
This proves the claim.
\end{proof}
We furthermore need the following lemma.
\begin{lemma}
\label{lemma:upsusmaller}
    Let $v \in V$ be a node and let $P_x$ be the \textbf{exact} shortest path to center $x \in \mathcal{X}$ in $G$.
    Further, let $\mathcal{X}' \subseteq \mathcal{X}$ be the set of centers that can potentially cluster a node $u \in P$.
    Define the values $\Upsilon_v^{(x)}$ and $\Upsilon_u^{(x)}$ as follows:
    \begin{align}
        \Upsilon_v^{(x)} &:= (\mathcal{D}-\delta_x) + d_{G}(x,v) - \min_{x' \in \mathcal{X}'\setminus\{x\}} \left((\mathcal{D}-\delta') + d_{G}(x',v)\right)\\
        \Upsilon_u^{(x)} &:= (\mathcal{D}-\delta_x) + d_{G}(x,u) - \min_{x' \in \mathcal{X}'\setminus\{x\}} \left((\mathcal{D}-\delta') + d_{G}(x',u)\right)
    \end{align}
    Then, it holds $\Upsilon_u^{(x)} \leq \Upsilon_v^{(x)}$.
\end{lemma}
\begin{proof}
First, for convenience, we define: 
\begin{align}
    x_u := \argmin_{x' \in \mathcal{X}'\setminus\{x\}}\left((\mathcal{D}-\delta') + d_{G}(x',v)\right).
\end{align}
Next, we add $\left(d_{G}(u,v) -  d_{G}(u,v)\right) = 0$ to $\Upsilon_u^{(x)}$ and obtain:  
\begin{align*}
    \Upsilon_u^{(x)} &= (\mathcal{D}-\delta_i) + d_{G}(x_i,u) - \left((\mathcal{D}-\delta_{u}) + d_{G}(x_u,u)\right)\\
    &= (\mathcal{D}-\delta_i) + d_{G}(x_i,u) - \left((\mathcal{D}-\delta_{u}) + d_{G}(x_u,u)\right) + \left(d_{G}(u,v) -  d_{G}(u,v)\right)\\
    &= (\mathcal{D}-\delta_i) + d_{G}(x_i,u) + d_{G}(u,v) - \left((\mathcal{D}-\delta_{u}) + d_{G}(x_{u},u) + d_{G}(u,v)\right)
\end{align*}
Now, we use the fact that $u$ is an ancestor of $v$ in the \emph{exact} shortest path $P$.
By this fact, it holds:
\begin{align}
    d_G(x_i,v) &= d_G(x_i,v) + d_G(u,v) \label{eqn:triangle_1}
\end{align}
Otherwise, there would be a shorter path from $v$ to $x$.
Therefore,
\begin{align*}
    \Upsilon_u^{(x)} \underset{\eqref{eqn:triangle_1}}{=} (\mathcal{D}-\delta_i) + d_{G}(x_i,v) - \left((\mathcal{D}-\delta_{u}) + d_{G}(x_{u},u) + d_{G}(u,v)\right)
\end{align*}
By the triangle inequality, we therefore obtain the following inequality:
\begin{align}
     d_{G}(x_{u},u) + d_{G}(u,v) \geq d_{G}(x_{u},v) \label{eqn:triangle_2}
\end{align}
Putting this insight back into the formula gives us:
\begin{align*}
    \Upsilon_u^{(x)} \underset{\eqref{eqn:triangle_2}}{\leq} (\mathcal{D}-\delta_i) + d_{G}(x_i,v) - \left((\mathcal{D}-\delta_{u}) + d_{G}(x_{u},v) \right)
\end{align*}
By definition, it holds:
\begin{align}
\label{eqn:mindef}
    \left((\mathcal{D}-\delta_{u}) + d_{G}(x_{u},v) \right) \geq \min_{x' \in \mathcal{X}'\setminus\{x\}} \left((\mathcal{D}-\delta_{u}) + d_{G}(x',v)\right)
\end{align}
Thus, we have:
\begin{align*}
    \Upsilon_u^{(x)} \underset{\eqref{eqn:mindef}}{\leq} (\mathcal{D}-\delta_i) + d_{G}(x_i,v) - \min_{x' \in \mathcal{X}'\setminus\{x\}} \left((\mathcal{D}-\delta_{u}) + d_{G}(x',v)\right) = \Upsilon_v^{(x)}
\end{align*}
This proves the lemma.
\end{proof}

Note that, when using exact distances, the center that minimizes $\Upsilon_v^{(x)}$ will add $v$ to its cluster.
However, this is not necessarily true when using approximate distances, as the approximation error can make another center appear closer.
Nevertheless, we can show that for a big enough difference $\Upsilon_v^{(x)}$, the center $x$ clusters node $v$ and, more importantly, also clusters all nodes close to $v$.
More precisely, it holds:
\begin{lemma}
\label{lemma:x_cluster_B}
Let $x \in N_v$ be a center in $\mathcal{X}$ in the distance at most $6\mathcal{D}$ to $v$.
Further, let $P := (v_1, \ldots, v_\ell)$ with $v_1 = x$ and $v_\ell = v$ be the \textbf{exact} shortest path from $x$ to $v$ in $G$.

\medskip

Consider a node $v_i \in P$ and a node $u \in B(v_i, \gamma\mathcal{D})$.
Suppose it holds 
\begin{align}
    \Upsilon_v^{(x)} < -(2\gamma+10\epsilon)\mathcal{D}.
\end{align}
Then this implies $v_i, u \in C_x$.
In particular, this means that for all $v_i \in P$, it holds:
\begin{align*}
     \pr{B(v_i,\gamma\mathcal{D}) \subset C_{x}} &\geq \pr{\Upsilon_v^{(x)} \leq -(2\gamma+10\epsilon)\cdot\mathcal{D}}
\end{align*}
\end{lemma}
\begin{proof}
Let $x_u$ be any center that can potentially cluster node $u$. 
By Claim \ref{claim:disttocenter}, each vertex joins the cluster of a center at a distance at most $(1+\epsilon)2\mathcal{D}$.
Therefore, the distances between $v$ and $v_i$ and $u$ and $x_u$ is at most $(1+\epsilon)2\mathcal{D}$.
As $u \in B(v,\gamma\mathcal{D})$, by the triangle inequality, the distance between $x_u$ and $v$ is at most
\begin{align*}
d_G(v,x_u) &\leq d_G(v,v_i) +  d_G(v_i,u) + d(u,x_u)\\
&\leq (1+\epsilon)2\mathcal{D} + \gamma\mathcal{D} + (1+\epsilon)2\mathcal{D} \\
&\leq (1+\epsilon)4\mathcal{D} + \gamma\mathcal{D} \leq 6\mathcal{D} 
\end{align*}
Therefore, $x_u \in N_v$. 

\medskip

We want to show that for $\Upsilon_v^{(x)} \leq -(2\gamma+10\epsilon)\mathcal{D}$, center $x$ is closer to $u$ than $x_u$ \textbf{even} when using approximate distances.
    The idea behind the proof is, loosely speaking, that for a big enough value of $\Upsilon_v^{(x)}$ the error introduced by the approximate SSSP is \emph{canceled out} in some sense. 
Suppose for contradiction that $x_u$ (and \emph{not} $x$) adds node $u$ to its cluster, i.e., it holds $u \in C_{x_u}$. 
This only happens if and only if $x_u$ is the penultimate node on the approximate shortest path to $s$.
Or, equivalently, is the penultimate node on the exact path in the tree $T_s$
Formally, it holds:
\begin{align*}
    d_{T_s}(s,u) = d(s,x_u) + d_{T_s}(x_u,u) \geq  d(s,x_u) + d_{G_s}(x_u,u) = d(s,x_u,u)
\end{align*}
On the other hand, using approximation guarantee and the triangle inequality, we see that the following holds:
    \begin{align*}
        d_{T_s}(s,u) 
        &\leq (1+\epsilon) \cdot d_{G_s}(s,u) \\
        &\leq d_{G_s}(s,x,u) + \epsilon\cdot d_{G_s}(s,x,u)\\
        &\leq d_{G_s}(s,x,v_i) + d_{G_s}(v_i,u) + \epsilon\cdot d_{G_s}(s,x,u)
    \end{align*}
By Claim \ref{claim:random_shift}, we furthermore get: 
     \begin{align*}     
        d_{T_s}(s,u) &\leq  \left(d_{G_s}(s,x_u,v_i) + \Upsilon_{v_i}^{(x)}\right) + d_{G_s}(u,v_i) + \epsilon\cdot d_{G_s}(s,x,u) \\
        &\leq  \left(d_{G_s}(s,x_u,u) + d_{G_s}(u,v_i) + \Upsilon_{v_i}^{(x)}\right) + d_{G_s}(u,v_i) + \epsilon\cdot d_{G_s}(s,x,u) \\
        &\leq  d_{G_s}(s,x_u,u) + 2d_{G_s}(u,v_i) + \epsilon\cdot d_{G_s}(s,x,u) + \Upsilon_{v_i}^{(x)}
    \end{align*}
As $v_i$ lies on the exact shortest path from $v$ to $x$, it holds by Lemma \ref{lemma:upsusmaller} that $\Upsilon_{v_i}^{(x)} \leq \Upsilon_{v}^{(x)}$ and therefore:
\begin{align*}
    d_{T_s}(s,u)  \leq  d_{G_s}(s,x_u,u) + 2d_{G_s}(u,v_i) + \epsilon\cdot d_{G_s}(s,x,u) + \Upsilon_{v}^{(x)}
\end{align*}
Now, we bound $d_{G_s}(s,x,u)$.
To this end, recall that $d(s,x) = \mathcal{D} - \delta_x \leq \mathcal{D}$.
Further, since $x \in N_v$, we have $d(x,v_i) \leq d(x,v) \leq 6\mathcal{D}$.
Finally, by definition, we have $d(v_i,u) = \gamma\mathcal{D}$.
Thus, combining all these facts, it holds:
\begin{align*}
    d_{G_s}(s,x,u) &:= d_{G_s}(s,x) + d_{G_s}(x,u)\\
    &\leq d_{G_s}(s,x) + d_{G_s}(x,v_i) + d_{G_s}(v_i,u)\\
    &\leq \mathcal{D} + 6\mathcal{D} + \gamma\mathcal{D} \leq 8\mathcal{D}
\end{align*}
Now use our assumptions and see: 
%\rhd \\
        \begin{align*}
          d_{T_s}(s,u)   &\leq  d_{G_s}(s,x_u,u) + 2d_{G_s}(u,v_i) + \epsilon\cdot d_{G_s}(s,x,u) + \Upsilon_{v}^{(x)}\\
          \text{ As }d_{G_s}(s,x,u) \leq 8\mathcal{D}:\\
          &\leq  d_{G_s}(s,x_u,u) + 2d_{G_s}(u,v) + \epsilon\cdot 8\mathcal{D} + \Upsilon_v^{(x)}\\
        \text{ As }d_{G_s}(v,u) \leq \gamma\mathcal{D}:\\
          &\leq  d_{G_s}(s,x_u,u) + 2\gamma\mathcal{D} + \epsilon\cdot 8\mathcal{D} + \Upsilon_v^{(x)}\\
          \text{ As }\Upsilon_v^{(x)} < -(2\gamma+10\epsilon)\mathcal{D}:\\
            &\leq d_{G_s}(s,x_u,u) + 2\gamma\mathcal{D} + \epsilon\cdot 8\mathcal{D} - 2\gamma\mathcal{D} - \epsilon\cdot 10\mathcal{D}\\
           &<  d_{G_s}(s,x_u,u) 
    \end{align*}
Thus, it follows that:
\begin{align*}
     d_{T_s}(s,u) \underset{\textbf{!}}{<}  d_{G_s}(s,x_u,u) \leq d_{T_s}(s,u)
\end{align*}
This is a contradiction.
As this holds for any possible choice $x_u \in N_v$, none of them add $u$ to their cluster.
Therefore $u$ must be part of $C_x$ as claimed and the lemma follows.
\end{proof}
This lemma tells us that --- if the random shift $\delta_x$ is big enough --- a center $x \in N_v$ will cluster all nodes that are \emph{close} to its shortest path to $v$.
Again, note that this statement is stronger than what we need for a pseudo-padded decomposition, but we show it in this generality to reuse it later in the next section.
In the remainder of this section, we will prove that with non-trivial probability, there is a center $x_i \in N_v$ with a large enough shift.
Concretely, we show the following lemma:
\begin{lemma}[Random Shift with Approximate Distances]
\label{lemma:random_shift}
Suppose that 
\begin{align*}
\gamma &\leq \nicefrac{1}{8}\\
\epsilon &\leq \min \left\{\frac{1}{40}, \frac{1}{20\cdot(2\log\tau+2) }\right\}
\end{align*}
Then, it holds
\begin{align*}
    \pr{\exists x_i \in N_v : \Upsilon_i \leq -(2\gamma+10\epsilon)\cdot\mathcal{D} } \geq e^{-16(\gamma+\epsilon)\log\tau}-\frac{160\log(\tau)}{n^c}
\end{align*}
\end{lemma}
The full proof is presented below.
It is based on the proof of Theorem 1 in \cite{Filtser19} that proves a similar statement but assumes exact shortest path computations.
The main difficulty is quantifying the effect of approximation parameter $\epsilon$.

\medskip

We conclude this section by proving Lemma \ref{lemma:padding_property} using Lemma \ref{lemma:x_cluster_B} and Lemma \ref{lemma:random_shift}.
Note that, by Lemma \ref{lemma:x_cluster_B}, it holds for all $v_j \in P_{x_v}$ that   
\begin{align*}
     \pr{B(v_j,\gamma\mathcal{D}) \subset C_{x_v}} &\geq \pr{\exists x_i \in N_v : \Upsilon_i \leq -(2\gamma+10\epsilon)\cdot\mathcal{D}} 
\end{align*}
Further, by Lemma \ref{lemma:random_shift}, we have
\begin{align*}
         \pr{\exists x_i \in N_v : \Upsilon_i \leq -(2\gamma+10\epsilon)\cdot\mathcal{D}} &\geq e^{-16(\gamma+5\epsilon)\log\tau}+\frac{160\log \tau}{n^c}
\end{align*}
This proves Lemma \ref{lemma:padding_property}.
% To conclude, we obtain a strongly $6\mathcal{D}$-bounded partition, such that for every $\gamma\le\frac{1}{32}$ and $v\in V$, the ball $B_G(v,\gamma\cdot4\mathcal{D})$ is fully contained in a single cluster with probability at least 
% \begin{align*}
%   \pr{B(v,\gamma\cdot6\mathcal{D})\subseteq P(v)}\ge e^{-4\cdot(4\gamma+\epsilon)\cdot\lambda}\geq e^{-32\cdot(\gamma+\epsilon)\cdot(\ln\tau+1)} - 4\lambda\epsilon~.  
% \end{align*}
% \henning{Overall, I don't quite get how this very last step implies theorem \ref{theorem:generalpartition}.}
Together with the initial observation that the algorithm can be implemented with a \emph{single} approximate shortest path computation on a set-source, this implies Theorem \ref{theorem:generalpartition}. 

\subsubsection{Proof of Lemma \ref{lemma:random_shift}}
\label{sec:padding_proof}

In this section, we will prove Lemma \ref{lemma:random_shift} using techniques and arguments from \cite{Filtser19} where a similar statement was shown for \emph{exact} distance computations.
In fact, we follow the analysis almost verbatim and only adapt the arguments that do not generally hold for approximate distances.
The main technical difficulties are in the facts that a) we need to carry the approximation error $\epsilon$ through the calculations and b) that the triangle inequality does not hold for approximate distances and, therefore, a node may not be added to the cluster of the center on the exact shortest path to virtual source $s$.

Condition on the event that there are $\tau$ centers in $N_v$ and let $N_v=\{x_1,x_2,\dots, x_\tau\}$.
Until now, we have always considered the rounded random variables $\delta_{x_1}, \ldots, \delta_{x_\tau}$ that can be encoded in $c\log n$ bits.
These are the variables that will be used by the approximate shortest path algorithm to compute the clusters.
In this section, we will work with actual continuous values, i.e., the actual values returned by sampling the truncated distribution.
For each center $x_i \in N_v$, we defined the random shift as
\begin{align*}
    \Upsilon'_i &:= (\mathcal{D}-\delta'_i\cdot\mathcal{D}) + d_{G}(x_i,v) - \min_{x_j \in N_v\setminus\{x_i\}} \left((\mathcal{D}-\delta'_j\cdot\mathcal{D}) + d_{G}(x_j,v)\right)
\end{align*}

As we round them up to be encodable in $c \log n$ bits for some $c \geq 1$, it holds for all $x_i \in N_v$ that
\begin{align}
    \delta_{x_i} - \delta'_{x_i}\cdot\mathcal{D} &\leq \frac{1}{n^c}\\
     (\mathcal{D}-\delta'_i\cdot\mathcal{D}) + d_{G}(x_i,v) - (\mathcal{D}-\delta_i) + d_{G}(x_i,v) &\leq \frac{1}{n^c}\\
    |\Upsilon'_i - \Upsilon_i| &\leq \frac{1}{n^c}
\end{align}
Thus, the error introduced by the rounding is negligible.
Nevertheless, we need to carry it through.

% Note that the closer the other centers are to $v$, the bigger the value of $\Upsilon_i$. 
% Therefore, for this proof, we will pessimistically assume that \textbf{all} centers $x_i \in N_v$ are close enough to cluster $v$. 
% In other words, as we know that $v$ will clustered by a center in distance at $(1+\epsilon)2\mathcal{D}$, we implicitly shorten all longer distances to $(1+\epsilon)2\mathcal{D}$.

Denote by $\mathcal{F}_i$ the event that $x_i$ is the truly closest center to $v$.
In other words, the value $\Upsilon'_i$ is non-positive and if there is another center with the same distance, it has a lower identifier, i.e., we have 
\begin{align}
\label{eq:fvalue}
 \mathcal{F}_i := \left\{ x_i = \argmin_{x_i \in N_v} \Upsilon_i \leq 0 \right\}   
\end{align}
Note that the tie-breaker only comes into effect when there two centers $x_,x_j$ with $\Upsilon_i = \Upsilon_j = 0$.
With exact distance computations, the minimal center will cluster $v$ as it is on the shortest path to supersource $s$.
This is not necessarily true when using approximate distances and rounding.
As we always round down, it holds:
\begin{align}
    \pr{\Upsilon'_i < -\frac{1}{n^c}} \leq \pr{\mathcal{F}_i} \leq \pr{\Upsilon'_i \leq \frac{1}{n^c}}
\end{align}
Further, let $\epsilon_l = \epsilon + \nicefrac{1}{n^c}$ and $\epsilon_u := \nicefrac{1}{n^c}$ denote by $\mathcal{C}_i$ the event that
\begin{align*}
   \mathcal{C}_i := \left\{ \Upsilon'_i \in \left(-2\cdot(\gamma+5\epsilon_l)\cdot\mathcal{D},\epsilon_u\right]\right\}. 
\end{align*}
In other words, a center $x_i$ is the closest center but not by a lot.
This is a \emph{bad} event as it implies that the ball $B(v,\gamma\mathcal{D})$ may not be added to $C_{x_i}$.
Perhaps not even $v$ is added.
On the flip side, if \emph{none} of these \emph{bad} events $\mathcal{C}_1, \ldots, \mathcal{C}_\tau$ happen, there \emph{must} be a center with a big enough shift.
Recall that we want to show that:
\begin{align*}
    \pr{\exists x_i \in N_v : \left(\Upsilon_i \leq -2\cdot(\gamma+5\epsilon)\cdot\mathcal{D}\right) } \geq e^{-16(\gamma+5\epsilon)\log\tau}-\frac{160\log\tau}{n^c}
\end{align*}
It holds:
\begin{align*}
     &\pr{\forall x_i \in N_v : \Upsilon_i \geq -2\cdot(\gamma+5\epsilon)\cdot\mathcal{D}}\\
     \text{As all $\mathcal{F}_i$ disjoint:}\\
     =& \sum_{i=1}^\tau \pr{\mathcal{F}_i} \cdot \pr{\forall x_i \in N_v : \Upsilon_i \geq -2\cdot(\gamma+5\epsilon)\cdot\mathcal{D} \mid \mathcal{F}_i}\\
     \text{As $\mathcal{F}_i$ implies $\Upsilon_i$ is minimal:}\\
     =& \sum_{i=1}^\tau \pr{\mathcal{F}_i} \cdot \pr{\Upsilon_i \geq -2\cdot(\gamma+5\epsilon)\cdot\mathcal{D} \mid \mathcal{F}_i}\\
        =& \sum_{i=1}^\tau \pr{\Upsilon_i \geq -2\cdot(\gamma+5\epsilon)\cdot\mathcal{D} \cap \mathcal{F}_i}\\
    \text{By Ineq. \eqref{eq:fvalue}:}\\
        =& \sum_{i=1}^\tau \pr{\Upsilon_i \geq -2\cdot(\gamma+5\epsilon)\cdot\mathcal{D} \cap \{\Upsilon'_i \leq \nicefrac{1}{n^c}\}}\\
        \text{As $|\Upsilon_i - \Upsilon_i'| \leq \nicefrac{1}{n^c}$:}\\
     \leq& \sum_{i=1}^\tau \pr{\Upsilon'_i \in \big(-2\cdot(\gamma+5\epsilon+\nicefrac{1}{n^c})\cdot\mathcal{D},\nicefrac{1}{n^c}\big]}\\
          \leq& \sum_{i=1}^\tau \pr{\Upsilon'_i \in \big(-2\cdot(\gamma+5\epsilon_l)\cdot\mathcal{D},\epsilon_u\big]}\\
     =& \sum_{i=1}^\tau \pr{\mathcal{C}_i}
\end{align*}
Using this insight, we note that the following equality holds:
\begin{align}
\label{eqn:c_union_bound}
      \pr{\exists x_i \in N_v : \Upsilon_i \leq -2\cdot(\gamma+5\epsilon)\cdot\mathcal{D} } \geq 1 - \sum_{i=1}^\tau \pr{\mathcal{C}_{i}}.
\end{align}
In the following, we will bound $\pr{\mathcal{C}_{i}}$ for a fixed $x_i$ and show the following.
\begin{lemma}
\label{lemma:lower_bound_cz}
Let $\lambda := 2\log \tau + 2$.
Then, for all $i \leq \tau$, it holds:
\begin{align*}
            \pr{\mathcal{C}_i} \leq \left(1-e^{-2(\gamma+5\epsilon_l+\epsilon_u)\cdot\lambda}\right)\cdot e^{2\lambda\epsilon_u}\cdot\left(\pr{\mathcal{F}_i }+\frac{1}{e^{\lambda}-1}\right)
        \end{align*}
\end{lemma}

\begin{proof}
For easier notation, we drop the index and denote $x=x_{i}$ to simplify notation and analogously fix $\mathcal{C}:=\mathcal{C}_{i}, \mathcal{F}:=\mathcal{F}_i$, and $\delta':=\delta'_{x_i}$.
    Consider a subset of $\tau$ centers $\mathcal{X'} := \{x_1, \dots, x_{\tau}\}$ and let $\mathcal{Z} := \{\delta_1, \dots, \delta_\tau\}$ be a realization of their random shifts.
Further, we define the value
	\begin{align*}
	    \rho_{\mathcal{Z}} := \frac{1}{\mathcal{D}}\cdot\left(d_{G}(x,v)+\mathbf{max}_{j<\tau}\left\{\delta'_{x_{j}}\cdot \mathcal{D}-d_{G}(x_{j},v)\right\} \right)    
	\end{align*}
Further, define $\Upsilon'_{\mathcal{Z}}=\mathcal{D}\cdot\rho_{\mathcal{Z}}-\mathcal{D}\cdot\delta'$ as the difference between $x$ and the closest center. 
Note that $\Upsilon'_{\mathcal{Z}}$ is a random variable that only depends on the shift $\delta$ (as everything else is fixed by conditioning on $\mathcal{Z}$).
To prove the Lemma, we will use that the following holds for any value $\alpha \geq 0$:
\begin{align*}
\pr{\Upsilon'_{\mathcal{Z}} \leq -\alpha\mathcal{D} } = \pr{\mathcal{D}(\rho_\mathcal{Z}-\delta') \leq -\alpha\mathcal{D} } = \pr{\delta' \geq \rho_\mathcal{Z} + \alpha }    
\end{align*}
We will now use the law of total probability to prove the lemma.
Denote by $f$ the density function of the distribution over all possible values of $\delta'$.
By the law of total probability, it holds:
    \begin{align*}
         \pr{\mathcal{C} \mid \mathcal{Z}} &:= \int_{y=0}^1 \pr{\mathcal{C}_{\mathcal{Z}} \mid \delta = y} f(y) dy\\
         &\leq \int_{y=\rho-\epsilon_u}^{\min\left\{ 1,\rho+2(\gamma+5\epsilon_l)\right\}} \pr{\mathcal{C}_{\mathcal{Z}} \mid \delta = y} f(y) dy\\ 
         &\leq \int_{y=\rho-\epsilon_u}^{\min\left\{ 1,\rho+2(\gamma+5\epsilon_l)\right\}} f(y) dy\\ 
         &=\int_{\rho-\epsilon_u}^{\min\left\{ 1,\rho+2(\gamma+5\epsilon_l)\right\} }\frac{\lambda\cdot e^{-\lambda y}}{1-e^{-\lambda}}dy
    \end{align*}
We can simplify this statement using some fundamental calculations and the definition of the truncated exponential function.
It holds:
\begin{claim}
\label{lemma:c_bound}
Consider a random variable $\delta' \sim \mathsf{Texp}(\lambda)$ drawn from a truncated exponential distribution with parameter $\lambda>0$.
For values $\rho,\gamma',\epsilon' > 0$, it holds:
\begin{align*}
    \int_{\rho-\epsilon'}^{\min\left\{ 1,\rho+\gamma'\right\} }\frac{\lambda\cdot e^{-\lambda y}}{1-e^{-\lambda}}dy \leq \left(1-e^{-\lambda(\gamma'+\epsilon')}\right)\cdot e^{2\epsilon'\lambda} \cdot\left(\pr{\delta \geq \rho+\epsilon'}+\frac{1}{e^{\lambda}-1}\right)~
\end{align*}
\end{claim}
\begin{proof}
The proof follows directly from the definition of $\delta$.
First, we note that it holds that
	\begin{align}
	     \pr{\delta>\rho+\epsilon'}
=\int_{\rho+\epsilon'}^{1}\frac{\lambda\cdot e^{-\lambda y}}{1-e^{-\lambda}}dy
	=\frac{e^{-(\rho+\gamma')\cdot\lambda}-e^{-\lambda}}{1-e^{-\lambda}}~.
	\end{align}
On the other hand, it holds:
	\begin{align*}
	\pr{\rho-\epsilon' \le\delta\le\rho+\gamma'}&=\int_{\rho-\epsilon'}^{\min\left\{ 1,\rho+\gamma'\right\} }\frac{\lambda\cdot e^{-\lambda y}}{1-e^{-\lambda}}dy\\
	& \le\frac{e^{-(\rho-\epsilon')\cdot\lambda}-e^{-\left(\rho+\gamma'\right)\cdot\lambda}}{1-e^{-\lambda}}\\
 & =\frac{e^{-(\rho-\epsilon')\cdot\lambda}-e^{-\left(\rho - \epsilon' + \epsilon'+\gamma'\right)\cdot\lambda}}{1-e^{-\lambda}}\\
  & =\frac{e^{-(\rho-\epsilon')\cdot\lambda}-e^{-\left(\rho - \epsilon')\cdot\lambda - (\epsilon'+\gamma'\right)\cdot\lambda}}{1-e^{-\lambda}}\\
	 &=\left(1-e^{-(\gamma'+\epsilon')\cdot\lambda}\right)\cdot\frac{e^{-(\rho-\epsilon')\cdot\lambda}}{1-e^{-\lambda}}\\
  &=\left(1-e^{-(\gamma'+\epsilon')\cdot\lambda}\right)\cdot\frac{e^{-(\rho+\epsilon'-2\epsilon')\cdot\lambda}}{1-e^{-\lambda}}\\
    &=\left(1-e^{-(\gamma'+\epsilon')\cdot\lambda}\right)\cdot e^{2\epsilon'\lambda}\cdot\frac{e^{-(\rho+\epsilon')\cdot\lambda}}{1-e^{-\lambda}}\\
      &=  \left(1-e^{-(\gamma'+\epsilon')\cdot\lambda}\right)\cdot e^{2\epsilon'\lambda}\cdot\frac{e^{-(\rho+\epsilon')\cdot\lambda}-e^{-\lambda}+e^{-\lambda}}{1-e^{-\lambda}}\\
    &=  \left(1-e^{-(\gamma'+\epsilon')\cdot\lambda}\right)\cdot e^{2\epsilon'\lambda}\cdot\left(\frac{e^{-(\rho+\epsilon')\cdot\lambda}-e^{-\lambda}}{1-e^{-\lambda}} + \frac{e^{-\lambda}}{1-e^{-\lambda}}\right) \\
    &= \left(1-e^{-(\gamma'+\epsilon')\cdot\lambda}\right)\cdot e^{2\epsilon'\lambda}\cdot \left(\pr{\delta>\rho+\epsilon'} + \frac{e^{\lambda}e^{-\lambda}}{e^{\lambda}(1-e^{-\lambda})}\right)\\
     &= \left(1-e^{-(\gamma'+\epsilon')\cdot\lambda}\right)\cdot e^{2\epsilon'\lambda}\cdot \left(\pr{\delta>\rho+\epsilon'} + \frac{1}{e^{\lambda}-1}\right)
	\end{align*}
Thus, the claim follows.
\end{proof}
Now, we set $\gamma' = 2\gamma+10\epsilon_l$ and $\epsilon' = \epsilon_u$ in our formula.
We get:
    \begin{align}
         \pr{\mathcal{C} \mid \mathcal{Z}} &\leq \left(1-e^{-2\lambda(\gamma+5\epsilon_l + \epsilon_u)}\right)\cdot e^{2\lambda\epsilon_u} \cdot\left(\pr{\delta \geq \rho+\epsilon_u}+\frac{1}{e^{\lambda}-1}\right)~
    \end{align}
% We now argue that $\pr{\mathcal{F}} \geq \pr{\delta \geq \rho+10\epsilon}$.
% Suppose it holds that $\Upsilon'_{\mathcal{Z}} \leq -6\cdot\epsilon\cdot\mathcal{D}$.
% As the distance of $v$ to itself is $0$, by choosing $\gamma=0$ in Lemma \ref{lemma:x_cluster_B}, this implies that $x$ clusters $v$.
% Note that $\Upsilon'_{\mathcal{Z}} \leq -10\cdot\epsilon\cdot\mathcal{D}$ is equivalent to assuming $\delta \geq \rho_{\mathcal{Z}} + 6\epsilon$.
% Therefore, we see:
% \begin{align}
%     \pr{\mathcal{F}} \geq \pr{\delta > \rho_{\mathcal{Z}} + 10\epsilon} = \frac{e^{-\lambda(\rho_\mathcal{Z}+10\epsilon)}-e^{-\lambda}}{1-e^{-\lambda}} 
% \end{align}
% Thus, we get:
% \begin{align}
%          \pr{\mathcal{C} \mid \mathcal{Z}} &\leq \left(1 - e^{-2(\gamma+5\epsilon)\lambda}\right)\cdot e^{20\epsilon\lambda}\cdot\left(\pr{\mathcal{F} \mid \mathcal{Z} }+\frac{1}{e^{\lambda}-1}\right)~
% \end{align}

To conclude the proof, denote by $f$ the density function of the distribution over all possible values of $\mathcal{Z}$.
 Using the law of total probability, we can bound the probability for event $\mathcal{C}$ as follows:
	\begin{align}
	\pr{\mathcal{C}} &=\int_{\mathcal{Z}}\pr{\mathcal{C}\mid \mathcal{Z}}\cdot f(\mathcal{Z})~d\mathcal{Z}\\
 &\leq \int_{\mathcal{Z}}\left(1-e^{-2\lambda(\gamma+5\epsilon_l + \epsilon_u)}\right)\cdot e^{2\lambda\epsilon_u} \cdot\left(\pr{\delta \geq \rho+\epsilon_u}+\frac{1}{e^{\lambda}-1}\right)\cdot f(\mathcal{Z})~d\mathcal{Z}\\
  &\leq \left(1-e^{-2\lambda(\gamma+5\epsilon_l + \epsilon_u)}\right)\cdot e^{2\lambda\epsilon_u} \cdot \int_{\mathcal{Z}}\left(\pr{\delta \geq \rho+\epsilon_u}+\frac{1}{e^{\lambda}-1}\right)\cdot f(\mathcal{Z})~d\mathcal{Z}\\
	 &\le\left(1 - e^{-2(\gamma+5\epsilon+\epsilon_u)\lambda}\right)\cdot e^{2\epsilon_u\lambda}\cdot\int_{\mathcal{Z}}\left(\pr{\mathcal{F}\mid \mathcal{Z}}+\frac{1}{e^{\lambda}-1}\right)\cdot f(\mathcal{Z})~d\mathcal{Z}\\
	 &=\left(1 - e^{-2(\gamma+5\epsilon+\epsilon_u)\lambda}\right)\cdot e^{2\epsilon_u\lambda}\cdot\left(\pr{\mathcal{F}}+\frac{1}{e^{\lambda}-1}\right)\label{eq:cisf}
	\end{align}
 This proves the lemma.
\end{proof}
Using this lemma, we see that it holds:
\begin{align*}
\sum_{i=1}^{|N_{v}|}\pr{\mathcal{C}_{i}}& \le\left(1 - e^{-2(\gamma+5\epsilon_l + \epsilon_u)\lambda}\right)\cdot e^{2\epsilon_u\lambda}\cdot\sum_{i=1}^{|N_{v}|}\left(\pr{\mathcal{F}_{i}}+\frac{1}{e^{\lambda}-1}\right)\\
& \le\left(1 - e^{-2(\gamma+5\epsilon_l + \epsilon_u)\lambda}\right)\cdot e^{2\epsilon_u\lambda}\cdot\left(1+\frac{\tau}{e^{\lambda}-1}\right)&\text{As $\sum \pr{\mathcal{F}_i}=1$}
\end{align*}
We can finalize the proof by carefully considering the possible values of $\gamma, \epsilon_l$ and $\epsilon_u$.
As $\epsilon_u = \nicefrac{1}{n^c}$ with $c>1$ and $\lambda = 2\log \tau + 2 \leq 2\log n + 2$, we can assume $2\epsilon_u\lambda < \nicefrac{1}{40\lambda} \leq 1.79$ for a large enough $n$.
This allows us to apply the well-known inequality $e^x \leq 1 + x + x^2$ for $x \leq 1.79$ and get
\begin{align}
\label{eqn:e12_bound}
    e^{2\epsilon_u\lambda} \leq 1 + 2\epsilon\lambda + (2\epsilon\lambda)^2 \leq 1 + 4\epsilon\lambda
\end{align}
Further, we picked $\gamma \leq \nicefrac{1}{8}$ and $\epsilon \leq \nicefrac{1}{40}$ small enough such that $\gamma+5\epsilon_l + \epsilon_u \leq \nicefrac{1}{4}$.
Given that observation, we see that it holds:
\begin{align*}
e^{-2(\gamma+5\epsilon_l + \epsilon_u)\lambda}&=\frac{e^{-2(\gamma+5\epsilon_l + \epsilon_u)\lambda}\left(e^{\lambda}-1\right)}{e^{\lambda}-1} \ge\frac{e^{-2(\gamma+5\epsilon_l+\epsilon_u)\lambda}\cdot e^{\lambda-1}}{e^{\lambda}-1}\\
&\ge
\frac{e^{-2(\nicefrac{1}{4})\lambda}\cdot e^{\lambda-1}}{e^{\lambda}-1}
\geq 
\frac{e^{\frac{\lambda}{2}-1}}{e^{\lambda}-1} 
\end{align*}
Here, we used the fact that $e^{x-1} \leq e^{x}-1$ for $x > 2$, which can be easily verified.
By our choice of $\lambda := 2\log(\tau) + 2$, it furthermore holds:
\begin{align}
      e^{-2(\gamma+5\epsilon_l + \epsilon_u)\lambda} \ge\frac{e^{\frac{\lambda}{2}-1}}{e^{\lambda}-1}  \geq \frac{e^{\frac{2\log(\tau)+2}{2}-1}}{e^{\lambda}-1}  = \frac{e^{\log\tau}}{e^{\lambda}-1}   \geq  \frac{\tau}{e^{\lambda}-1}\label{eq:taulambda}
\end{align}   
Therefore, we can simplify our formula as follows:
\begin{align*}
\sum_{i=1}^{|N_{v}|}\pr{\mathcal{C}_{i}} & \le\left(1 - e^{-2(\gamma+5\epsilon_l + \epsilon_u)\lambda}\right)\cdot e^{2\epsilon_u\lambda}\cdot\left(1+\frac{\tau}{e^{\lambda}-1}\right)\\
 \text{ By Ineq. \eqref{eqn:e12_bound}}: \\
& \le \left(1 + 4\lambda\epsilon_u\right)\cdot \left(1 - e^{-2(\gamma+5\epsilon_l + \epsilon_u)\lambda}\right)\left(1+\frac{\tau}{e^{\lambda}-1}\right)\\
\text{ By Ineq.\eqref{eq:taulambda}}:\\
&\le \left(1 + 4\lambda\epsilon_u\right)\cdot \left(1 - 
e^{-2(\gamma+5\epsilon_l + \epsilon_u)\lambda}\right)\left(1+e^{-2(\gamma+5\epsilon_l+\epsilon_u)\cdot\lambda}\right)\\
\text{ As } (1+x)(1-x) &= 1^2-x^2:\\
&=\left(1+4\lambda\epsilon_u\right) \left(1-e^{-4(\gamma+5\epsilon_l + \epsilon_u)\cdot\lambda}\right)\\
&\le 1-e^{-4(\gamma+5\epsilon_l + \epsilon_u)\cdot\lambda} + 4\lambda\epsilon_u \\
&\le 1-e^{-4(\gamma+5\epsilon + 5\epsilon_u + \epsilon_u)\cdot\lambda} + 4\lambda\epsilon_u \\
&= 1-e^{-4(\gamma+5\epsilon)\cdot\lambda}\cdot e^{-24\epsilon_u\cdot\lambda} + 4\lambda\epsilon_u \\
\text{As $e^{-x} \leq 1-x/2$:}\\
&\leq 1-e^{-4(\gamma+5\epsilon)\cdot\lambda}\cdot (1-12\epsilon_u\cdot\lambda) + 4\lambda\epsilon_u \\
&\leq 1-e^{-4(\gamma+5\epsilon)\cdot\lambda} + 12 \epsilon_u\cdot\lambda + 4\lambda\epsilon_u \\
&\leq 1-e^{-4(\gamma+5\epsilon)\cdot\lambda} + 16 \epsilon_u\cdot\lambda
\end{align*}
Putting this bound on $\sum_{i=1}^{|N_{v}|}\pr{\mathcal{C}_{i}}$ back in our initial inequality gives us the following final approximation of our desired probability:
\begin{align*}
    \pr{\exists x_i \in N_v : \Upsilon_i \leq -(2\gamma+10\epsilon)\cdot\mathcal{D} } & \geq 1 - \sum_{i=1}^\tau \pr{\mathcal{C}_{i}} \\
    & \geq 1 - \left(1-e^{-4(\gamma+5\epsilon)\cdot\lambda} + 16\lambda\epsilon\right)\\
    & = e^{-4(\gamma+5\epsilon)\cdot\lambda}  + 16\lambda\epsilon_u\\
    \text{Recalling that }\lambda = 2\log\tau + 2 \leq 4\log\tau:\\
& \geq e^{-4(\gamma+5\epsilon)\cdot(4\log\tau)}  - 16(4\log\tau )\epsilon_u\\
& \geq e^{-16(\gamma+5\epsilon)\cdot(\log\tau)}  - 64\epsilon_u\log\tau\\
& = e^{-16(\gamma+5\epsilon)\cdot(\log\tau)}  - \frac{64\log\tau}{n^c}
\end{align*}
This proves Lemma \ref{lemma:random_shift}.

\newpage

\section{Full Analysis of \autoref{thm:genericldd}}
\label{sec:appendix_ldc}

We now present the algorithm behind Theorem \ref{thm:genericldd} in more detail.
As input, we are given a weighted graph $G=(V,E,\ell)$ with polynomially bounded weights, a diameter bound $\mathcal{D}$, and the centers $\mathcal{X} \subseteq V$. 
On a high level, the algorithm consists of four stages: In the first stage, we compute an initial decomposition by the algorithm from Theorem \ref{theorem:generalpartition}.
We want to use the blurry ball growing from Lemma \ref{lemma:bbg} on these clusters to get the desired cut probabilities.  
However, we cannot directly use thems as input sets because the resulting balls would overlap.
Instead, we have to ensure somehow that the input sets are separated from each other. 
To this end, we shrink the clusters to avoid overlaps. 
This works in two steps:
First, by identifying all nodes in sufficient distance to other clusters, and then shrinking the clusters in a specific way (which differs for the two the diameter guarantees).
Finally, we get the desired decomposition by using a blurry ball growing on each shrunk cluster.
As a parameter for the blurry ball growing, we choose 
\begin{align}
    \rho := \left(1-\alpha\right)\cdot\left(\frac{\mathcal{D}}{\cblur\log(\tau)}\right)
\end{align}
Here, $\alpha \in O\left(\nicefrac{\log \log n}{\log n}\right)$ is the value from Lemma \ref{lemma:bbg} and $\cblur$ is a (large) constant that will be determined in the analysis.
Further, we fix
\begin{align}
    \epsilon := \frac{1}{1024\cdot\log \tau}
\end{align}
as the parameter for approximate \SetSSP computations.

\newpage
In more detail, the four stages are as follows:

\begin{mdframed}
    \paragraph*{\textbf{(Step 1) Compute a Pseudo-Padded Decomposition:}}
First, we create a pseudo-padded decomposition using the algorithm from Theorem \ref{theorem:generalpartition}.
%     We quickly reiterate this algorithm here as we require some of its specifics.
%     Just as in \cite{BeckerEL20}, our algorithm creates an augmented graph $G_s$ by adding a virtual node $s$ and virtual edges $\{s,v\}$ for every possible center $x \in \mathcal{X}$. 
% To attach weights to the newly created virtual edges, we sample a value $\delta_x \in (0,1]$ by a truncated exponential distribution for every center $X$.
% Every virtual edge $\{s,X\}$ now gets assigned a weight of $\mathcal{D}-\delta_x\cdot\mathcal{D}$. 
% We now perform an approximate SSSP on $G_s$ to obtain a tree $T$ rooted at as $s$, such that $d_T(s,v) \leq (1+\epsilon)d_G(s,v)$ for every $v \in V$. For every virtual edge $\{s,x\}$ which is part of $T$, every node in the subtree of $x$ is assigned the same cluster. 
% In particular, during this process, each node learns its approximate distance to the center of its cluster and each cluster has a spanning tree that contains all its nodes.
% These two facts will be important for the next steps.
\paragraph*{\textbf{(Step 2) Compute (Approximate) Distance to Boundary:}}
Each node computes the approximate distance to the closest node in a different cluster.
This works in two substeps
\begin{itemize}
    \item {\textbf{(Substep 2a) Identify Nodes on Boundary:}} All the nodes of a cluster that are on the boundary of that cluster, i.e., nodes that are adjacent to another cluster, compute their (approximate) distance to the neighboring cluster.
    A node can determine whether it is on the boundary by checking its neighbors.
    Let $v \in V$ be a node on the boundary and let $N'_v$ be a subset of its neighbors, s.t., each $w \in N'_v$ is a neighbor in a different cluster.
    Then, node $v$ computes:
    \begin{align}
        \omega(v) = \min_{w \in N_v'} d(v,w) 
    \end{align}
    \item {\textbf{(Substep 2b) Compute (Approximate) Distance to Boundary:}} Finally, we create another virtual source $s'$ and all nodes $v$ on the boundary create an edge of length $\omega(v)$.
    Then perform an $(1+\epsilon)$-approximate shortest path from $s'$ and let each node compute $d_T(s',v)$. 
\end{itemize}
\paragraph*{\textbf{(Step 3) Identify Well-Separated Nodes:}} We identify \emph{connected} sets of nodes with large distance to all other connected sets. This works in two substeps:
\begin{itemize}
    \item \textbf{(Substep 3a) Find Nodes With Sufficient Distance to Boundary:} Each nodes examines the distance $d_T(s',v)$ computed in the previous step. Suppose it holds
\begin{align*}
    d_T(s',v) \geq \left(\frac{\rho}{1-\alpha} + \epsilon\mathcal{D} \right)
\end{align*}
we mark $v$ as \emph{active} and set
\begin{align*}
    V_a := \left\{v \in V \mid d_T(v,s') \geq \left(\frac{\rho}{1-\alpha} + \epsilon\mathcal{D} \right) \right\}
\end{align*}
\item \textbf{(Substep 3b) Find Connected Set of Active Nodes:}
Let $G[V_a]$ be the graph induced by the active nodes.
A node can locally decide which edges are part of this graph by checking which neighbors are active.
Let $\mathcal{X}_a = \mathcal{X} \cap V_a$ be the set of active centers.
Perform $(1+\epsilon)$-approximate \SetSSP from $\mathcal{X}_a$ in $G[V_a]$.
Let $T_a$ be the resulting approximate SSSP tree.
Then, if and only if an active node has a path of length $3\mathcal{D}$ to its cluster center in $T_a$, it remains active.
We denote the remaining active nodes as $V'_a$.

\end{itemize}
\paragraph*{\textbf{(Step 4) Execute Blurry Ball Growing:}}
We execute a blurry ball growing on the active set $V'_a$ with parameter $\rho = \left(1-\alpha\right)\cdot\left(\frac{\mathcal{D}}{\cblur\log(\tau)}\right)$.
As per Lemma \ref{lemma:bbg}, we obtain a superset $S(V'_a)$ of $V'_a$.

\paragraph*{\textbf{(Step 5) Create a Spanning Tree:}}
We perform one last approximate set-source SSSP from all centers $x \in \mathcal{X} \cap V'_a$.
Each node $v \in V'_a$ adds itself to the cluster of the closest center.
We add all these the refined clusters to our decomposition.
\end{mdframed}

\medskip

Note that the algorithm of Becker, Emek, and Lenzen only differs in Substep (3b), which is simply absent in their algorithm.
Instead they perform BBG on the active nodes computed in in Substep (3a).
However, this is exactly the additional step which ensures that the resulting clusters are connected and therefore have a strong diameter.

\subsection{Analysis} 
\label{sec:ldc_analysis}

We now prove Theorem \ref{thm:genericldd}. 
The proof is divided into three parts.
First, prove that the cluster's diameter is bounded by $8\mathcal{D}$ in Lemma \ref{lemma:clustering_diameter}.
After that, we analyze the cutting probabilities in Lemma \ref{lemma:clustering_cut} and show that each edge is cut with probability at most $O\left(\frac{\alpha\ell}{\mathcal{D}}\right)$.
Finally, we show the algorithm's complexity in Lemma \ref{lemma:clustering_complexity}.
Together, Lemmas \ref{lemma:clustering_diameter}, \ref{lemma:clustering_cut}, and \ref{lemma:clustering_complexity} prove Theorem \ref{thm:genericldd}.

\paragraph*{\textbf{Cluster Diameter:}}
To start off the analysis, we show that the algorithm creates a clustering of diameter $8\mathcal{D}$.
This follows more of less directly from the construction.

\begin{lemma}[Diameter Guarantee]
\label{lemma:clustering_diameter}
    Each cluster created in the last step has strong diameter at most $8\mathcal{D}$.
\end{lemma}

First, we note the diameter guarantee follows \textbf{if} the blurry ball growing process never adds two nodes from different clusters to the same ball.
To prove this, let $A_i$ be the connected set of active nodes in the cluster of $x_i$. 
Note that we enforce that all active nodes have a path of length $3\mathcal{D}$ to a center that is completely contained in the cluster.
Thus, the claim already follows for the nodes in $A_i$.
Further, recall that for all nodes $w$ that are added to the cluster through the blurry ball growing, it holds:
\begin{align}
 d(w,A_i) = d(w,V'_a) \leq \frac{\rho}{1-\alpha} = \frac{\left(1-\alpha\right)\cdot\left(\frac{\mathcal{D}}{\cblur\log(\tau)}\right)}{1-\alpha}= \left(\frac{\mathcal{D}}{\cblur\log(\tau)}\right) \leq \mathcal{D}
\end{align}
Now consider $v,w \in V$ to be two nodes added to cluster $K_i$ of center $x_i \in \mathcal{X}$.
Then, it holds:
\begin{align}
d_{K_i}(v,w) &\leq d_{G}(v,A_i) + \mathbf{max}_{v',w' \in A_i} d_{G}(v,V'_a) + d_{G}(v,A_i)\\
&\leq \mathcal{D} + d_{G}(v',x_i) + d_{G}(w',x_i) + \mathcal{D}\\
&\leq 2(3\mathcal{D}+\mathcal{D}) = 8\mathcal{D}.
\end{align}
Thus, to prove the diameter guarantee, we must argue why we can safely execute blurry balls growing from the active nodes without adding active nodes of a different cluster.
\begin{claim}
    For each active node $v \in V_a$, all nodes in distance $\frac{\rho}{1-\alpha}$ are part of the same cluster. 
\end{claim}
\begin{proof}
    For contradiction, assume there is node $u \in V$ distance smaller than $\frac{\rho}{1-\alpha}$ to $v$ that is in a different cluster.
    In particular, let $u$ be the closest such node.
    Clearly, if we prove a contradiction for $u$, it holds for all other nodes, too.
    
    We begin with the following observation:
    As there can be no other node in a different cluster on the exact shortest path from $u$ to $v$ in $G$ (as this node would closer to $v$ than $u$),
    there must be boundary node $w \in V$, s.t., the following two conditions hold:
    \begin{enumerate}
        \item $w$ is on the path from $u$ to $v$. Therefore, it holds $d(u,v) = d(u,w) + d(w,v)$.
        \item $w$ adds a virtual edge of length $\omega(w) \leq d(u,w)$ to the virtual source $s$.
    \end{enumerate}
    Together, these two facts imply:
    \begin{align}
        d_{G_s}(s,v) \leq d_{G}(u,v) \label{eqn:s_closer_u}
    \end{align}
    Here, $G_s$ is the graph we obtain by adding the (virtual) supernode to $G$.
    Therefore, it holds:
    \begin{align*}
        \frac{\rho}{1-\alpha} &> d_G(u,v) & \rhd \textit{ Per Assumption }\\
        &\geq d_{G_s}(s,v) & \rhd \textit{ By Ineq. }\eqref{eqn:s_closer_u}
\end{align*}
Now recall that $v$ is an active node.
Otherwise, it would not be part of the set that is blurred.
Thus, it holds
\begin{align}
    d_T(s,v) \geq \frac{\rho}{1-\alpha} + \epsilon\mathcal{D}
\end{align}
by the definition of active nodes.
Therefore, it holds
\begin{align*}
 d_{G_s}(s,v) &\geq \frac{d_T(s,v)}{1+\epsilon} & \rhd \textit{ As } d_T(s,v) \leq (1+\epsilon)d_{G_s}(s,v)\\
        &\geq \frac{\left(\frac{\rho}{1-\alpha} + \epsilon\mathcal{D} \right)}{1+\epsilon}  & \rhd \textit{ As $v$ is active} \\
        &\geq \frac{\left(1 + \epsilon \right)\frac{\rho}{1-\alpha}}{1+\epsilon}& \rhd \textit{ As } \frac{\rho}{1-\alpha} \leq \mathcal{D} \\
        &\geq \frac{\rho}{1-\alpha}
    \end{align*}
This is a contradiction as this implies:
\begin{align*}
    \frac{\rho}{1-\alpha} >  d_{G_s}(s,v) \geq \frac{\rho}{1-\alpha}
\end{align*}
So, node $u$ cannot exist, which proves the lemma.
\end{proof}
Therefore, it is safe to execute the ball growing on a subset of active nodes without risking overlapping with other clusters.
Thus, all balls must have a diameter of $8\mathcal{D}$

\paragraph{\textbf{Edge-Cutting Probability}}
Next, we consider the probability of cutting edges.
Altogther, we want to show that the following holds:
\begin{lemma}
\label{lemma:clustering_cut}
An edge of length $\ell$ is cut with probability at most $O\left(\frac{\ell\cdot\log \tau}{\mathcal{D}}\right)$.
\end{lemma}
Recall that the only operation that cuts edges is the BBG procedure, and so at first glance, Lemma \ref{lemma:bbg} should immediately yield the result. 
However, this only holds true for a single iteration.
The main difficulty of this proof is that the algorithm runs multiple iterations until all nodes are clustered, and therefore, the BBG has many \emph{chances} to cut a given edge.
To handle this, we will exploit that each edge is cut or added to a cluster after at most $2$ iterations on expectation.

More precisely, the proof is structured as follows: 
First, we consider only a single iteration and show that each edge is cut with probability at most $O\left(\frac{\ell\cdot\log \tau}{\mathcal{D}}\right)$.
Further, we show that a node is added to a cluster with constant probability, and so all its edges must be gone after $O(\log n)$ iterations, w.h.p.
Then, we can use a (more or less) standard argument from the study of \LDD's. It roughly goes as follows: Any clustering algorithm $\mathcal{A}$ that cuts edges between cluster with some probability $\alpha$ and places a constant fraction of nodes into clusters, can be turned into an \LDD with quality $O(a)$.
The idea is simple and straightforward. Recursively apply the algorithm  until all nodes are clustered.
For the sake of completeness, we prove this general statement and use it finish our analysis.

We begin with the cutting probability of a single iteration.
Here, it holds:
\begin{lemma}
\label{lemma:clustering_cut_fixed_iteration}
\textbf{In a fixed iteration}, an edge of length $\ell$ is cut with probability at most $O\left(\frac{\ell\cdot\log \tau}{\mathcal{D}}\right)$.
\end{lemma}
\begin{proof}
Note that the blurry ball growing is the only operation that removes edges.
Following Lemma \ref{lemma:bbg}, the probability that an edge is cut by this process is $\ell\cdot\rho$. 
By our choice of $\rho$, the short edges are cut with probability
\begin{align*}
    \frac{\ell}{\rho} &= \ell\cdot\left(\frac{\mathcal{D}}{\cblur\log \tau}\left(1-\alpha\right) \right)^{-1} := \ell\cdot\left(\frac{\mathcal{D}}{\cblur\log \tau}\left(1-O\left(\frac{\log\log n}{\log n}\right)\right) \right)^{-1} \\
    &\leq \ell\cdot\left(\frac{\mathcal{D}}{\cblur\log \tau}\left(1-\frac{1}{2}\right) \right)^{-1} \leq \ell\cdot\left(\frac{\mathcal{D}}{2\cblur\log \tau}\right)^{-1} \in O\left(\frac{\ell\cdot\log \tau}{\mathcal{D}}\right)
\end{align*}
This proves the desired bound on the probability.
\end{proof}

\paragraph{\textbf{Clustering Probability}} Further, we lower bound how many nodes are actually added to a cluster.
Here, the analysis differs from Becker, Lenzen, and Emek and we use our tighter probability bounds on the behaviour of the Pseudo-Padded Decomposition.
Whereas they show that an \emph{individual} node is active with constant probability in Substep (3a), we show that this also implies that (with constant probability) all nodes on the path its cluster center are active as well.
For this, we use a technical observation from the analysis of Theorem \ref{theorem:generalpartition}.
We show that the following holds:
\begin{lemma}[Clustering Probability]
\label{lemma:clustering_prob}
    In a fixed iteration, each node $v \in V$ is clustered with probability at least $\beta \geq \nicefrac{1}{2}$.
\end{lemma}
\begin{proof}
Recall that each \emph{active} node in the set $V'_a$ is added to a cluster no matter what.
To this end, we will analyse how many nodes are active.
For the analysis, define $\gamma = (\frac{1}{\cblur\log\tau} + \epsilon)$.
Further, let $x_v$ now be the center clusters $v$ in the pseudo-padded decomposition. 
Note that a node $v \in V$ is definitely active if the ball $B(v,\gamma\mathcal{D})$ is fully contained in the same cluster as $v$.
In this case, the boundary is in distance at least $\gamma\mathcal{D}$. 
Recall our approximate shortest path algorithm only \emph{overestimates} and therefore $v$ passes the check and is added to $V_a$.

This alone is not sufficient to be also added to $V_a'$. 
We must additionally show that a path of length at most $\frac{3\mathcal{D}}{(1+\epsilon)}$ to $x_v$ also remains active.
To this end, let $P_{x_v} = (v_1, \ldots, v_\ell)$ be the \textbf{exact} shortest path from $x_v$ to $v$ in $G$.
We will show that with constant probability, \emph{all} nodes on this path (including $x_v$) are active.
Note that this path is at most of length $(1+\epsilon)2\mathcal{D}$ as per \ref{claim:disttocenter} from the analysis in the previous section.
This implies that there is path of active nodes of length $(1+\epsilon)2\mathcal{D}$.
For small enough $\epsilon$, it holds $(1+\epsilon)2\mathcal{D} \leq \frac{3\mathcal{D}}{(1+\epsilon)}$.
Thus, our approximate shortest path algorithm must find a path of length at most $(1+\epsilon)\frac{3\mathcal{D}}{(1+\epsilon)} \leq 3\mathcal{D}$ in the induced active graph $G[V_a]$.
Therefore, $v$ will remain active, is added to $V'_a$, and therefore will be in a cluster.

Now, recall the analysis of Theorem \ref{theorem:generalpartition}. 
There, we proved that:

\padding*

This allows us to prove the lemma.
Recall $\cblur\geq {128}$ and we execute a pseudo padded decomposition with padding parameter $\lambda = 2\log\tau + 2$, diameter parameter $\mathcal{D}$, and $\epsilon \leq \frac{1}{1024\log\tau}$. 
Let $\gamma = (\frac{1}{\cblur\log\tau} + \epsilon)$.
Note that $\gamma \leq \frac{1}{8}$ for a our choice of $\nicefrac{1}{\cblur} \leq \frac{1}{128}$ and $\epsilon \leq \frac{1}{1024\log\tau}$.
Finally, define:
\begin{align*}
    \mathcal{F}_v := \bigcup_{w \in P_{x_v}} \left\{ B_G\left(w, \gamma\mathcal{D}\right) \subset C_{x_v} \right\}
\end{align*}
Then, by Lemma \ref{lemma:padding_property}, it holds:
\begin{align*}
           \pr{\mathcal{F}_v} &\geq e^{\left(-16 \cdot \log\tau \cdot \left(\frac{1}{\cblur'\log\tau}+6\epsilon\right) \right)} - \frac{160{\log\tau}}{n^c}\\
\text{Using that } e^x \geq 1+x:\\
           &\geq 1 - 16 \cdot \log\tau \cdot \left(\frac{1}{\cblur'\log\tau}+6\epsilon\right) - \frac{160{\log\tau}}{n^c} \\
           \text{Using that } \nicefrac{1}{n^c} \leq \epsilon:\\
           &= 1 - \frac{16\log\tau}{\cblur\log\tau} - 256\cdot\epsilon{\log\tau}\\
           &= 1 - \frac{16}{\cblur} - 256\cdot\epsilon{\log\tau}\\
\text{Using that } \cblur \geq 128 = 4\cdot16:\\
           &\geq 1 - \frac{1}{4} - 256\cdot\epsilon{\log\tau}\\
\text{Using that } \epsilon \leq \nicefrac{1}{1024\log\tau} = \nicefrac{1}{4\cdot 256\cdot\log\tau}:\\
&\geq 1 - \frac{1}{4} - \frac{1}{4} = \frac{1}{2}
\end{align*}
Therefore, with constant probability, a node is active all nodes on the shortest path to its cluster center are also active.
This means, there is a path of length $3\mathcal{D}$ to the cluster center that only consists of active nodes.
Thus $v \in V'_a$ with probability at least $\nicefrac{1}{2}$ as needed.
\end{proof}

\paragraph{Complexity}
Finally, we consider the algorithm's runtime.
The follows directly from the algorithm description as it only uses minor aggregations and approximate shortest-path computations.

\begin{lemma}[Complexity]
\label{lemma:clustering_complexity}
    The algorithm can be implemented with $\Tilde{O}(1)$ approximate \SetSSP computations and minor aggregations, w.h.p.
\end{lemma}
\begin{proof}
As there can be at mots $O(\log n)$ iterations, w.h.p, we focus on a single iteration.
Each iteration begins with one execution of the pseudo-padded decomposition; this requires only \SetSSP computations with approximation parameter $O(\nicefrac{1}{\log n})$. This follows directly from Theorem \ref{theorem:generalpartition}.

Next, we require one aggregation to let the boundary nodes compute the distances to nodes in neighboring clusters. 
To be precise, each node $v \in V$ is its own minor.
We perform no aggregation within the minor and jump directly to the last step where we perform a aggregation on the edges.
In other words, we skip the consensus step and directly proceed to the aggragation step.
For each incident edge $\{v,w\} \in E$, node $v \in V$ chooses the input $(x_v, \ell_{(v,w)})$.
Here, $x_v \in \mathcal{X}$ is the (identifier of the) center that clustered $v$.
Then, we aggregate the minimum of all values with $x_w \neq x_v$ (by considering inputs containing $x_v$ as $\infty$).
Thus, each node learns the distance closest \emph{adjacent} node in a different cluster with \emph{one} minor aggregation.
To conclude the second step, we perform one approximate SSSP on graph $G_{s'}$ that contains one virtual node $s'$.
The approximation parameter is $\epsilon \in O(\nicefrac{1}{\log n})$.

In the third step, we need another local aggregation to determine which nodes are active. 
Similar to the step before, each node performs an aggregation on all its edges.
Then, we perform an approximate SSSP on the graph induced by the active nodes.
The approximation parameter is again $\epsilon \in O(\nicefrac{1}{\log n})$

Finally, we perform one execution of the blurry ball growing, which only uses $\Tilde{O}(1)$ approximate \SetSSP computations with parameter $\alpha \in O(\nicefrac{\log\log n}{\log n})$ as per Lemma \ref{lemma:bbg}.
A final SSSP to determine the connected components conludes the step.

Thus, each step only performs $\Tilde{O}(1)$ approximate \SetSSP computations or aggregations.
In all cases, the approximation factor is within $\Omega(\nicefrac{1}{\log^2})$.
Altogether, this proves the runtime bounds from Theorem \ref{thm:genericldd}.
\end{proof}

\newpage

\section{Proof of \autoref{lemma:folkloreldd}}
\label{sec:appendix_ldd}

Let $G_0 := G$ and define $G_i:= (V_i, E_i)$ with $i \geq 1$ the remaining graph after we apply the algorithm $\mathcal{A}$ on $G_{i-1}$.
On expectation, in each application of $\mathcal{A}$, a $\beta$-fraction of the nodes is added to some cluster.
This follows from the definition of $\beta$ and the linearity of expectation:
For each node $v \in V$, let $X^i_v$ be the binary random variable that $v$ gets clustered in the $i^{th}$ application of the algorithm, i.e., whether $v$ will be in $G_{i+1}$ or not. 
It holds $\pr{X_i \mid v \in V_i} \geq \beta$ by definition.
By the linearity of expectation, it therefore holds:
\begin{align*}
    \E{|V_{i+1}|} &= \sum_{v \in V} \pr{v \in V_i} \left(1 - \pr{X_i \mid v \in V_i}\right)\\
    & \leq (1-\beta)\sum_{v \in V} \pr{v \in V_i}\\
    &= (1-\beta)\cdot \E{|V_{i}|}
\end{align*}
Thus, a simple induction reveals that after $t$ applications, it holds:
\begin{align*}
    \E{|V_{t}|} \leq \E{|V_{t}|}\cdot(1-\beta)^t = n\cdot(1-\beta)^t 
\end{align*}
After $t = (c+1) \cdot\beta\cdot \log n$ recursive steps, the expected number of nodes is less than 
\begin{align*}
    \E{|V_{t}|} &\leq  n\cdot(1-\beta)^t\\
    &= n\cdot(1-\beta)^{(c+1) \cdot\beta\cdot \log n}\\
    &\leq n \cdot e^{-(c+1)\log n} = n^{-c}
\end{align*}
Here, we used that $(1-\frac{1}{x})^x \approx \frac{1}{e}$.
Therefore, using Markov's inequality, we have
\begin{align*}
    \pr{|V_t| \geq 1} &= \pr{|V_t| \geq n^{-c}\cdot n^c}\\
    &\leq \pr{|V_t| \geq n^{c}\cdot \E{|V_t|}}\\
    &\leq \nicefrac{1}{n^c} 
\end{align*}
Thus, there are no nodes left, w.h.p, proves the runtime of $O(\beta^{-1}\log n)$ application of $\mathcal{A}$.

\medskip

Next, we consider the cutting probability.
Fix an edge $z \in E$ and denote by $\Cut_z$ the event that this edge is cut in any of the $O(\beta^{-1}\log n)$ iterations.
We need to show that
\begin{align*}
    \pr{\Cut_z} \leq O\left(\frac{\alpha \cdot \ell_z}{\beta^{2} \cdot \mathcal{D}}\right).
\end{align*}
Recall that \emph{being cut} means that one endpoint of $z$ is added to a cluster, but the other is added to a different cluster (or not at all). 
In the following, we say that an edge gets \emph{decided} if either endpoint is added to some cluster.
If this happens, the other endpoint is added to the same cluster or not.
In the former case, the edge is saved as it cannot be cut in any future iteration. In the latter case, the edge is irreversibly cut.
For a fixed edge $e \in E$, let $Y_i$ be the event that the edge is decided in recursive step $i$.
Recall that either endpoint of an edge $z \in E_i$ that is still fully contained in the graph is clustered with probability at least $\beta$. 
Thus, the probability that \emph{any} of the two endpoints is added to some cluster is also at least $\beta$.
Therefore, the probability that no endpoint is added to a cluster in the first $i-1$ iterations is at most:
\begin{align}
\label{eq:decide}
    \pr{Y_i} \leq (1-\beta)^{i-1}
\end{align}
Conditioned on the event $Y_i$, the edge can \emph{only} be cut in the $i^{th}$ iteration.
In particular, the probability of a cut is $0$ in the first $i-1$ iterations, as the event $Y_i$ implies that both endpoints survive these iterations.
Further, it cannot be cut later as it will not exist in subsequent graphs.
Seeking formalization of these two observations, let now $\Cut_z^i$ be the indicator that $z$ is cut in the $i^{th}$ iteration.
Then, the probability that an edge is cut \textbf{under the condition that one endpoint is added to a cluster in $G_i$} is:
\begin{align*}
    \pr{\Cut \mid Y_i} &= \pr{\Cut \mid  \{z \in E_{i}\} \wedge \{z \not\in E_{i+1}\}}\\
    &= \lim_{T \to \infty} \sum_{i=1}^T \pr{\Cut_z^i \mid  \{z \in E_{i}\} \wedge \{z \not\in E_{i+1}\}}\\
    & = \pr{\Cut_z^i \mid \{z \in E_{i}\} \wedge \{z \not\in E_{i+1}\}} 
    \end{align*}
Recall that, in a fixed iteration, the edge $z$ is cut with probability at most $O\left(\frac{\ell_z\cdot \alpha}{\mathcal{D}}\right)$. 
However, we need to condition this probability on the fact that the edge gets decided in the $i^{th}$ iteration.
Using the law of total probability, we see:
\begin{align*}
    \pr{\Cut_z^i \mid z \in E_{i}} =& \pr{ \{z \not\in E_{i+1} \mid z \in E_{i} } \cdot \pr{\Cut_z^i \mid \{z \in E_{i}\} \wedge \{z \not\in E_{i+1}\}}\\
    &+ \pr{ \{z \in E_{i+1} \mid z \in E_{i} } \cdot \pr{\Cut_z^i \mid \{z \in E_{i}\} \wedge \{z \in E_{i+1}\}}\\
    =& \pr{ \{z \not\in E_{i+1} \mid z \in E_{i} } \cdot \pr{\Cut_z^i \mid \{z \in E_{i}\} \wedge \{z \not\in E_{i+1}\}}
\end{align*}
Therefore, we conclude:
\begin{align}
\label{eq:cutcond}
    \pr{\Cut \mid Y_i} &\leq \frac{\pr{\Cut_z^i \mid z \in E_i}}{\pr{z \not\in E_{i+1} \mid z \in E_i}} \leq O\left(\frac{\alpha\cdot\ell_z}{\beta \cdot \mathcal{D}}\right)
\end{align}
Plugging together our two bounds, the law of total probability now gives us:
\begin{align*}
    \pr{\Cut} &\leq \lim_{T \to \infty} \sum_{i=1}^T \pr{Y_i} \cdot \pr{\Cut \mid Y_i} & \rhd \textit{Law of Total Prob.}\\
    &\leq \lim_{T \to \infty} \sum_{i=1}^T \pr{Y_i} \cdot O\left(\frac{\alpha\cdot\ell_z}{\beta \cdot \mathcal{D}}\right) & \rhd \textit{By Ineq. \eqref{eq:cutcond}} \\
    &\leq O\left(\frac{\alpha\cdot\ell_z}{\beta \cdot \mathcal{D}}\right) \cdot \lim_{T \to \infty} \sum_{i=1}^T \pr{Y_i}  \\
     &\leq O\left(\frac{\alpha\cdot\ell_z}{\beta \cdot \mathcal{D}}\right) \cdot \lim_{T \to \infty} \sum_{i=1}^T (1-\beta)^i & \rhd \textit{By Ineq. \eqref{eq:decide}} \\
     &\leq O\left(\frac{\alpha\cdot\ell_z}{\beta \cdot \mathcal{D}}\right) \cdot \frac{1}{\beta} \leq O\left(\frac{\alpha\cdot\ell_z}{\beta^2 \cdot \mathcal{D}}\right)
\end{align*}
In the last line, we used that for each $x$ with $\frac{|x-1|}{|x|} \leq 1$, it holds:
\begin{align*}
    \lim_{T \to \infty} \sum_{i=1}^T (1-\nicefrac{1}{x})^{i-1} = \lim_{T \to \infty} \sum_{i=0}^T (1-\nicefrac{1}{x})^{i} = x
\end{align*}
This is clearly the case for $\nicefrac{1}{x} := \beta \leq 1$.
Thus, the proclaimed cutting probability of $O\left(\frac{\alpha\cdot\ell_z}{\beta^2 \cdot \mathcal{D}}\right)$ follows, which concludes the proof.

\section{Full Analysis for \autoref{thm:distributed_weak_separator}}
\label{sec:appendix_separator}

In this section, we present our technical contribution that enables the distributed and parallel construction of routing schemes and clustering in $k$-path separable graphs.
More precisely, we present the algorithm behind \autoref{thm:distributed_weak_separator}.
Namely, for a given $k$-path separable graph $G := (V,E,\ell)$, we give an algorithm constructs a weak $O(\epsilon^{-1}\cdot k\cdot \log n)$-path $(\mathcal{D},\epsilon)$-separator for parameters $\mathcal{D}$ and $\epsilon$.
The algorithm is essentially a ball carving process, that draws a path of length at most $4\cdot \mathcal{D}$ and removes balls of diameter at most $\epsilon\mathcal{D}$ around the paths.
As we will see, it can be implemented by using approximate shortest-path computations and minor aggregations.

\subsection{Algorithm Description}
The algorithm consists of $T \in O(\epsilon^{-1}\cdot k\log n)$ sequential steps.
In each step $t \leq T$, we construct a path $P_t$ and ball $K_t \supset P_t$ around that path.
Further, we define $G_{t} := (V_t,E_t,w)$ to be the graph induced by all nodes that have not yet been added to any set $B_1, \ldots, B_{t}$.
Finally, we fix the parameter  $\epsilon' := \nicefrac{\epsilon}{2}$.
Given these definitions, a single step of the algorithm consists of the following five operations:

    \paragraph*{\textbf{(Step 1) Choose a random node $v_t \in V_{t-1}$:}}
First, we pick a node $v \in V_{t-1}$ uniformly at random.
This can be done by letting each node $v \in V_{t-1}$ draw a random number $r_v \in [1,n^c]$ for a constant $c > 2$ uniformly and independently at random.
All nodes draw unique numbers for a large enough $c$, w.h.p.
Given that there are no duplicate numbers, we aggregate the maximum of all these numbers and pick the node with the highest number as $v_t:= \argmax_{v \in V} r_v$.
\paragraph*{  \textbf{(Step 2) Construct a $2$-approximate SSSP tree $T_{v_t}$ for $v_t$:}}
%This step is perhaps self-explanatory.
We perform a $2$-approximate SSSP from node $v_t$.
% This can be done, for example, using the algorithm of Rozhon \textcircled{r} al.\cite{DBLP:conf/stoc/RozhonGHZL22}.
% However, any other approximate or exact single source shortest path algorithm can be used as well.
Afterwards, all nodes $w \in V$ know a value $d_{T_{v_t}}(w,v_t)$, which is its $2$-approximate distance to $v_t$ in graph $G_{t-1}$.
In particular, it holds for all $w \in V_{t-1}$:
\begin{align*}
    d_{G_{t-1}}(v_t,w) \leq d_{T_{v_t}}(v_t,w) \leq 2d_{G_{t-1}}(v_t,w)
\end{align*}
\paragraph*{ \textbf{ (Step 3) Choose a random path $P_t$:}}
Given the distances computed in the previous step, each node can check whether it is in the distance at most $4\mathcal{D}$ to $v_t$.
In particular, we can define the set 
\begin{align*}
   N_{v_t} := \left\{w \in V_{t-1} \mid d_{T_{v_t}}(v_t,w) \leq 4\mathcal{D} \right\}
\end{align*}
We now want to sample one of these nodes uniformly at random.
This can again be done by letting each node $w \in N_{v_t}$ draw a random number uniformly and independently at random and then aggregating the maximum. 

Let now $P_t := (v_t, \ldots, w_t)$ be the path from $v_t$ to $w_t$ in $T_{v_t}$.
All nodes on this path need to learn that they are on this path.
For this, we use tree operations from Lemma \ref{lemma:tree_operations}, in particular, that all nodes can compute aggregate functions on all their descendants with $\Tilde{O}(1)$ aggregations.
Node $w_t$ chooses $1$ as its private input; all others choose $0$.
Now, each node computes the sum of all its descendants' inputs.
It is easy to verify that this sum is $1$ for all nodes on $P_t$ and $0$ for all others.
For the latter, we exploit that $v_t$ is the tree's root (otherwise, it would not be true, and we would require a second aggregation).

\paragraph*{ \textbf{ (Step 4) Compute an Approximation of $B_G(P_t, \epsilon \cdot \mathcal{D})$:}}
Next, we compute a $2$-approximate shortest path tree for $\{P_t\}$ in $G$.
    Thereby, we obtain a $2$-approximate SSSP tree $T_{P_t}$ and all nodes $w \in V_{t-1}$ know a value $d_{T_{P_t}}(w,P_t)$, which is its $2$-approximate distance to $P_t$ in graph $G$.
    After that, every node knows whether it is in the set
    \begin{align*}
        K_t := \left\{w \in V_t \mid d_{T_{P_t}}(w,P_t) \leq 2\epsilon'\mathcal{D} \right\}
    \end{align*}
As we chose $\epsilon' = \nicefrac{\epsilon}{2}$, the set $K_t$ is a strict subset of $B_G(P_t,\epsilon\cdot\mathcal{D})$ and a superset of $B_G(P_t,\epsilon/2\cdot\mathcal{D})$.
%\henning{I think this last statement is false. If we have a node $w$ of distance exactly $\epsilon\mathcal{D}$ to $P_t$, it could compute $d_{T_{P_t}}(w,P_t) = 2\epsilon\mathcal{D}$ and thus $w \not\in K_t$. And thus, $K_t$ is not a superset of $B_G(P_t,\epsilon\cdot\mathcal{D})$. It would be true if we define $K_t := \left\{w \in V_t \mid d_{T_{P_t}}(w,P_t) \leq 2\epsilon\mathcal{D} \right\}$}

\paragraph*{ \textbf{(Step 5) Remove $K_t$ from $G_{t-1}$:}}
All nodes use their previously computed distances to determine whether they are in $K_t$.
If so, they remove themselves from the graph, i.e., they will not be considered for picking the path to the next iteration. However, they will be considered for computing the balls around the paths.

A visualization of our algorithm's main ideas can be found in Figure \ref{fig:sep_algo}.
Clearly, the algorithm produces paths and sets that fit the description of Definition \ref{def:weak_k_path_separator}.
Therefore, to prove Lemma \ref{thm:distributed_weak_separator}, it remains to show that (a) removing the sets $K_1, \ldots, K_T$ weakly separates the graph and (b) the algorithm can be implemented in $\Tilde{O}(T)$ minor aggregations and shortest path computations.

% \begin{figure}
%      \begin{subfigure}[t]{0.48\textwidth}
%          \centering
%          \includestandalone[width=\textwidth]{../../Figures/MeshSeparator}
%          \caption{Marked in red is a separator path unknown to the algorithm.}
%      \end{subfigure}
%      \hfill
%      \begin{subfigure}[t]{0.48\textwidth}
%          \centering
%          \includestandalone[width=\textwidth]{../../Figures/MeshPath}
%          \caption{When sampling a random path of bounded length, the probability of crossing the separator is constant.}
%      \end{subfigure}

% \bigskip
%         \begin{subfigure}[t]{0.48\textwidth}
%          \centering
%          \includestandalone[width=\textwidth]{../../Figures/MeshBall}
%          \caption{We remove nodes in distance $\epsilon\mathbf{D}$ around the path. 
%          This also removes an $\epsilon\mathcal{D}$-fraction of the separator.}
%      \end{subfigure}
%     \hfill
%      \begin{subfigure}[t]{0.48\textwidth}
%          \centering
%          \includestandalone[width=\textwidth]{../../Figures/MeshPath2}
%          \caption{If residual graph is not well-separated, the probability to cross (what remains of) the separator is still constant.}
%      \end{subfigure}
%         \caption{We choose a two-dimensional mesh for illustration. While the operations in Figures \textbf{(b)} and \textbf{(c)} are straightforward to prove, the core of our analysis we will be the claim we make in Figure \textbf{(d)}.}
%         \label{fig:sep_algo}
% \end{figure}

\subsection{Analysis}

We divide the analysis into two steps.
First, we will prove the algorithm's correctness and show that it indeed produces a separator with the desired properties within $T \in O(\epsilon^{-1}\cdot k \cdot \log n)$ iterations.
To be precise, we prove the following lemma:
\begin{lemma}
\label{thm:weak_divandconquer}
Let $G_\tau$ be the graph after the $\tau^{\mathsf{th}}$ step of our algorithm.
Then, for $\tau_c \in O\left(c \cdot \epsilon^{-1} \cdot k\right)$ with probability at least $1-O(e^{-c})$, it holds for all $v \in V_{\tau_c}$:
\begin{align*}
    B_{G_{\tau_c}}\left(v, \mathcal{D}\right) \leq (\nicefrac{7}{8})\cdot n
\end{align*}
\end{lemma}
Then, we show that each iteration of the algorithm can be implemented with few minor aggregations and approximate path computations.
Formally, we show that
\begin{lemma}
\label{lemma:separator_runtime}
    An iteration of the algorithm can be implemented with $\Tilde{O}(1)$ minor aggregations and two $(1+\epsilon)$-approximate shortest path computations.
\end{lemma}
Together, these two results imply Theorem \ref{thm:distributed_weak_separator}.
For $T \in O(\epsilon^{-1}\cdot k \cdot \log n)$, Lemma \ref{thm:weak_divandconquer} implies that we sample a separator, w.h.p., and Lemma \ref{lemma:separator_runtime} implies that the $T$ can be implemented efficiently in \HYBRID and \CONGEST.

\paragraph{\textbf{Correctness (Proof of Lemma \ref{thm:weak_divandconquer}):}}
Before we get to the formal proof, let us first clarify the intuition.
On a high level, we want to show that the sets $K_1, \ldots, K_T$ we sample quickly cover the graph's $k$-path separator $S = (\mathcal{P}_1, \mathcal{P}_2, \ldots)$.  
As many paths must cross the separator, we will likely pick such a path through the random sampling.
Recall that we sample paths of length at most $4\mathcal{D}$ with an area of at least $\nicefrac{\epsilon}{2}\mathcal{D}$ around them.
Thus, we cover a subpath of length $\nicefrac{\epsilon}{2}\mathcal{D}$ whenever we pick a path that intersects with the separator.
Say that a path in the separator $S$ is of (weighted) length $\mathcal{D}'$.
Then, we require at least $\Omega(\nicefrac{\mathcal{D}'}{\epsilon\mathcal{D}})$ crossings to cover it. 
Thus, we need to bound the length of the paths in the separator.
A major factor that complicates this is the fact that the paths in a $k$-path separator are divided into subsets $(\mathcal{P}_1,\mathcal{P}_2,\ldots)$.
The paths in $\mathcal{P}_i$ with $i> 1$ are only the shortest paths if the paths from previous subsets are removed.
Thus, if the (weighted) diameter of $G$ is $\mathcal{D}$, this only gives us a bound for the paths in $\mathcal{P}_1$.
This does not help us when considering the lengths of the paths in the induced graph $G\setminus \mathcal{P}_1$.
The diameter of $G\setminus \mathcal{P}_1$ and, therefore, the length of the paths in $\mathcal{P}_2$ are unbounded.

Now, recall that we do \textbf{not} need to cover the whole separator as it is sufficient that each node only has a constant fraction of nodes in distance $\mathcal{D}$.
Our main trick is only considering a carefully picked subset of the separator paths that intersect with paths until length $O(\mathcal{D})$.
This subset only consists of paths of bounded length that separate not all, but still a \emph{large enough} fraction of the nodes.
We begin our search for this set of paths with the following lemma:
\begin{lemma}
\label{lemma:helper_separator}
    Let $G := (V, E, w)$ be a weighted {$k$-path separable} graph of $n$ nodes with (weighted) diameter $\mathcal{D}$.
    Then, there exists a set $\mathcal{Z} \subseteq V$ with $|\mathcal{Z}| \geq (\nicefrac{3}{4})\cdot n$ and a set $\mathcal{B} = \{P_1, \ldots, P_\kappa\}$ of at most $k$ simple path of length $32\mathcal{D}$ such that for all $v \in \mathcal{Z}$, it holds $B_{G \setminus \mathcal{B}}(v,8\mathcal{D}) \leq \left(\nicefrac{3}{4}\right)\cdot n$.
\end{lemma}

\begin{proof}
Define $\mathcal{D}' = 8\mathcal{D}$.
To construct set $\mathcal{Z}$, we will first establish some helpful properties of $k$-path separable graphs.
Our first important concept is the \emph{critical index} of a $k$-path separator. 
We define the \textbf{critical index} $i_{\mathcal{D}'}$ to be the highest index, $1 \leq i_{\mathcal{D}'} \leq k'$, for which it holds:
    \begin{align*}
        \exists v \in V: |B_{G^{(i_{\mathcal{D}'})}}(v,\mathcal{D}')| > (\nicefrac{3}{4}) \cdot n
    \end{align*}
    and 
     \begin{align*}
        \not\exists v \in V: |B_{G^{(i_{\mathcal{D}'}+1)}}(v,\mathcal{D}')| > (\nicefrac{3}{4}) \cdot n 
    \end{align*}
%\end{definition} 
One can easily verify that there is a well-defined $i_{\mathcal{D}'}$ for each $\mathcal{D}' \geq \mathcal{D}$.
Note that $\mathcal{D}'$ is greater than the graph's diameter and therefore:
\begin{align*}
    \forall v \in V: |B_{G}(v,\mathcal{D}')| =  n
\end{align*}
Recall that by definition of $k$-path separators by Abraham and Gavoille in \cite{DBLP:conf/podc/AbrahamG06}, each connected component in $G \setminus S$ is of size at most $\nicefrac{n}{2}$.
Thus, \emph{any} ball of \emph{any} size around \emph{any} node (that has not been removed as part of the separator) can contain at most $\nicefrac{n}{2}$ nodes once we removed $S$.
Thus, if there were no critical index for a distance $\mathcal{D}'$, it would hold: 
   \begin{align*}
        \exists v \in V: |B_{G\setminus S}(v,\mathcal{D}')| > (\nicefrac{3}{4}) \cdot n
    \end{align*}
This is a contradiction as all connected components in $G\setminus S$ are of size at most $\nicefrac{n}{2}$.
From now on, we are interested in these first $i_{\mathcal{D}'}$ (sets of) paths of the separator $S$, which we denote as
\begin{align*}
    \mathcal{S}' := (\mathcal{P}_1, \ldots, \mathcal{P}_{i_{\mathcal{D}'}}). 
\end{align*}
Note that it holds for all $v \in V\setminus \mathcal{S}'$:
\begin{align*}
    B_{G \setminus \mathcal{S}'}(v,\mathcal{D}') \leq (\nicefrac{3}{4}) \cdot n
\end{align*}
Thus, the set $\mathcal{S}'$ is a weak $\mathcal{D}'$-separator.
We will refer to it as the critical separator.
Further, we call the last \emph{unseparated} subgraph $G^{(i_{\mathcal{D}'})}$ the critical graph.
Next, we introduce another central element of our proof, the \emph{terminal node}, which is a node that still has more than $(\nicefrac{3}{4})\cdot n$ nodes in the distance $\mathcal{D}'$ in $G^{(i_{\mathcal{D}'}-1)}$.
As there can be more than one node with this property, we pick the node with the highest identifier.
However, any other method to break ties would also work.
Formally, we say a node $z_{\mathcal{D}'} \in V$ is a terminal node if and only if it holds: 
    \begin{align*}
         z_{\mathcal{D}'} := \argmax_{z \in V} \left\{z \in V \mid \forall j \leq i_{\mathcal{D}'}: |B_{G^{(j)}}(z,\mathcal{D}')| \ge (\nicefrac{3}{4})\cdot n \right\}
    \end{align*}
Given the concept of terminal nodes and critical indices, we can now define our set $\mathcal{Z}$ as
\begin{align*}
\mathcal{Z} &:= B_{G^{(i_{\mathcal{D}'})}}( z_{\mathcal{D}'} ,\mathcal{D}') 
\end{align*}  
%\end{definition}
Clearly, $\mathcal{Z}$ contains at least $(\nicefrac{3}{4}) \cdot n$ nodes by the definition of the terminal node. 
Thus, we already have the correct size required by the lemma.
Therefore, it remains to prove the existence of the set $\mathcal{B}$.
We define it as the intersections of all balls $B_{G^{(i)}}(z,\mathcal{D}')$ of all nodes $z \in \mathcal{Z}$ with all paths of $\mathcal{S}'$ from the respective graphs.
Formally, the boundary is defined as follows:
\begin{definition}[Boundary]
\label{def:boundary}
Let $\mathcal{S}(\mathcal{D}') := (\mathcal{P}_1, \ldots, \mathcal{P}_{i_{\mathcal{D}'}})$ with $\mathcal{P}_i := \{P_{i_1}, P_{i_2}, \ldots \}$ be the critical separator and let $\mathcal{Z}$ be the core set.
Then, we define the boundary $\mathcal{B} :=  \{ \mathcal{P}'_{1}, \mathcal{P}'_{2}, \ldots \} $ as follows:
    \begin{align*}
        P'_{i_j} &:= P_{i_j} \cap B_{G^{(i-1)}}(\mathcal{Z},\mathcal{D}')\\
        \mathcal{P}'_{i} &:= \{P'_{i_1}, P'_{i_2}, \ldots\}\\
        \mathcal{B} &:= \{ \mathcal{P}'_{1}, \mathcal{P}'_{2}, \ldots \} 
    \end{align*}
\end{definition}
\begin{figure}
    \centering
    \begin{tikzpicture}
\tikzset{opacity=1/.style={decorate, decoration=snake}}
    \coordinate (A) at (0,0);   
    \draw[name path=c1, fill=gray, opacity=0.3] (A) circle [radius=5cm];
    \draw[name path=c2, fill=gray, opacity=0.1] (A) circle [radius=3cm];
    \node[] at (0,-2) {\Large $\mathcal{Z}(\mathcal{D}')$};
    \node[] at (0,-4) {\Large $B_{G^{(i-1)}}(\mathcal{Z}(\mathcal{D}'), \mathcal{D}')$};

    \coordinate (B) at (-2,1);
    \coordinate (C) at (-4.5,-4.5);
    \path[name path=BC] (B) -- (C);
    \path [name intersections={of=BC and c1,by=BxC}];
    \draw [color=red, very thick, opacity=1] (B) -- (BxC);
    \draw [very thick, opacity=1](BxC) -- (C);
    
    \coordinate (D) at (5.5, -2.5);
    \path[name path=BD] (B) -- (D);
    \path [name intersections={of=BD and c1,by=BxD}];
    \draw [color=red, very thick, opacity=1] (B) -- (BxD);
    \draw [very thick, opacity=1](BxD) -- (D);

    \coordinate (E) at (-4, 4);
    \path[name path=EB] (E) -- (B);
    \path [name intersections={of=EB and c1,by=ExB}];
    \draw [color=red, very thick, opacity=1] (B) -- (ExB);
    \draw [very thick, opacity=1](ExB) -- (E);
    \node[] at (-4.3, 4.3) {\Large $\mathcal{S}(\mathcal{D}')$};

    \coordinate (G) at (5.5, -0.5);
    \path[name path=EG] (C) -- (G);
    \path [name intersections={of=EG and BD, by=H}];
    \path[name path=GH] (G) -- (H);
    \path [name intersections={of=GH and c1,by=GxH}];
    \draw [color=red, very thick, opacity=1] (H) -- (GxH);
    \draw [very thick, opacity=1](GxH) -- (G);

    \coordinate (I) at (-2.5, 5);
    \coordinate (J) at (5.3, 2);
    \coordinate (K) at (1, 2);
    \path [name path=IK] (I) -- (K);
    \path [name path=KJ] (K) -- (J);
    \path [name intersections={of=IK and c1, by=IxK}];
    \draw [color=red, very thick, opacity=1] (IxK) -- (K);
    \draw [very thick, opacity=1](I) -- (IxK);
    
    \path [name intersections={of=KJ and c1, by=KxJ}];
    \draw [color=red, very thick, opacity=1] (K) -- (KxJ);
    \draw [very thick, opacity=1](KxJ) -- (J);
    
\end{tikzpicture}
    \caption{An illustration of the boundary $\mathcal{B}(\mathcal{Z}(\mathcal{D}'))$. The inner circle denotes the core set $\mathcal{Z}(\mathcal{D}')$ and the outer circle denotes the ball $B_{G^{(i-1)}}(\mathcal{Z}(\mathcal{D}'), \mathcal{D}')$. $\mathcal{S}(\mathcal{D}')$ is the critical separator, consisting of several shortest paths. The red parts of $\mathcal{S}(\mathcal{D}')$ make up the boundary $\mathcal{B}(\mathcal{Z}(\mathcal{D}'))$. It suffices to remove the boundary s.t. a constant fraction of nodes have a distance greater than $\mathcal{D}'$ from the core set.}
    \label{fig:boundary}
\end{figure}
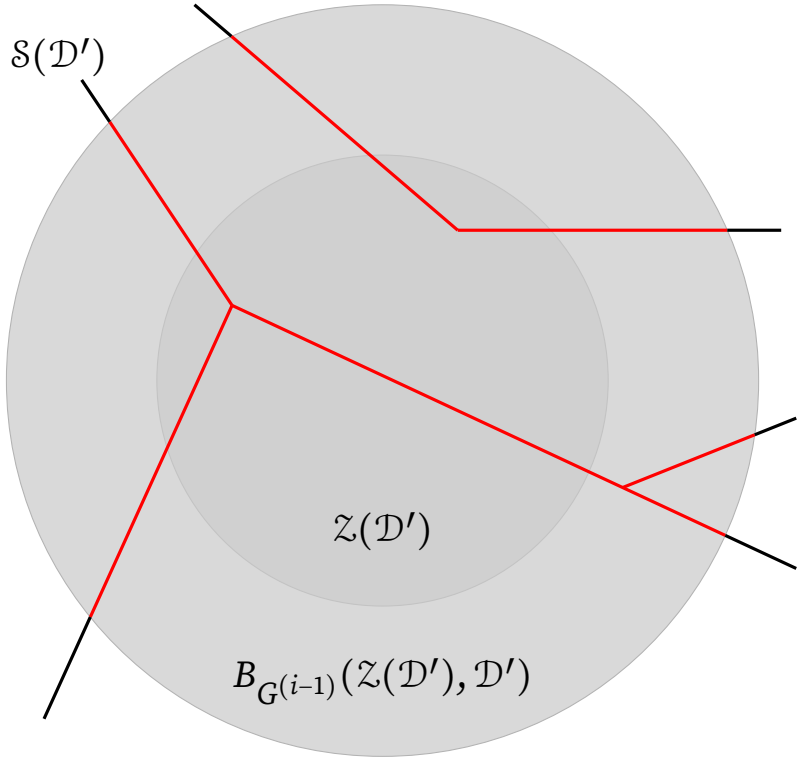
Figure \ref{fig:boundary} depicts the intuition behind this definition. 
Note that the boundary is a strict subset of a $k$-path separator.
Therefore, there can be at most $k$ sets $P'_{i_j}$, one for each path $P_{i_j}$ in the separator. 
However, the sets $P'_{i_j}$ themselves are not necessarily connected paths.
They are merely subsets of paths that do not necessarily consist of consecutive nodes of that path.
We will show that for each $P'_{i_j}$ there is a path of length at most $4\mathcal{D}'$ that covers all its nodes.
More formally, it holds:
\begin{claim}
\label{lemma:boundary_weight}
Let $G := (V,E,\ell)$ be a weighted $k$-path separable graph let $\mathcal{B}(\mathcal{Z}(\mathcal{D}')) := \{\mathcal{P}'_1, \mathcal{P}'_2, \ldots\}$ with $\mathcal{P}'_i = \{P'_{i_1}, P'_{i_2}, \ldots\}$ be as defined in Definition \ref{def:boundary}.
Then, for all $P'_{i_j}$, there is a simple path $\Tilde{P}_{i_j}$, s.t., it holds: 
    \begin{align*}
         \ell(\Tilde{P}_{i_j}) \leq 4\mathcal{D}' 
    \end{align*}
\end{claim}
\begin{proof}
Let $P_{i_j}$ be a path from the critical separator $\mathcal{S}$.
   In the following, we denote $v_{(1)}$ as the first node of $P_{i_j}$ that intersects with $B_{G^{(i-1)}}(\mathcal{Z},\mathcal{D}')$.
    Likewise, let $v_{(\ell')}$ be the last node of $P_{i_j}$ that intersects with $B_{G^{(i-1)}}(\mathcal{Z},\mathcal{D}')$.
    Thus, all nodes of $P_{i_j}$ that intersect with $B_{G^{(i-1)}}(\mathcal{Z},\mathcal{D}')$ lie on the subpath $\Tilde{P}_{i_j} := (v_{(1)},\ldots,v_{(\ell')})$.
    As it covers all nodes, it remains bound to its length.
    Recall that $P_{i_j}$ is a shortest path in $G^{(i-1)}$ and therefore, it holds that: 
    \begin{align*}
        d_{P_{i_j}}(v_{(1)},v_{(\ell')}) = d_{G^{(i-1)}}(v_{(1)},v_{(\ell')}).
    \end{align*}
    For both $v_{(1)}$ and $v_{(\ell')}$ there must be (not necessarily distinct) nodes $v',w' \in \mathcal{Z}$, s.t., it holds $v_{(1)} \in B_{G^{(i-1)}}(v',\mathcal{D}')$ and $v_{(\ell')} \in B_{G^{(i-1)}}(w',\mathcal{D}')$.
    By definition, both $v'$ and $w'$ are in distance $\mathcal{D}'$ to terminal node $z_{\mathcal{D}'}$ in the critical graph ${G^{(i_{\mathcal{D}'})}}$.
    Since we only remove nodes from $G$, it holds:
    \begin{align*}
        B_{G^{(1)}}(z_{\mathcal{D}'},\mathcal{D}') \supseteq B_{G^{(2)}}(z_{\mathcal{D}'},\mathcal{D}') \supseteq \ldots \supseteq B_{G^{(i_{\mathcal{D}'})}}(z_{\mathcal{D}'},\mathcal{D}')
    \end{align*}
    Thus, we can upper bound the distance between $v'$ and $w'$ in $G^{(i-1)}$ by their distance in $G^{(i_{\mathcal{D}'})}$.
    By the triangle inequality, it therefore holds for all $v',w' \in \mathcal{Z}(\mathcal{D}')$:
\begin{align*}
         d_{G^{(i-1)}}(v',w') \leq d_{G^{(i_{\mathcal{D}'-1})}}(v',w') \leq d_{G^{(i_{\mathcal{D}'-1})}}(v',z_{\mathcal{D}'}) + d_{G^{(i_{\mathcal{D}'}-1)}}(z_{\mathcal{D}'},w') \leq 2\mathcal{D}'
\end{align*}
Thus, both $v'$ and $w'$ are in the distance at most $2\mathcal{D}'$ in graph $G^{(i-1)}$ where $P_{i_j}$ is a shortest path. 
Combining these facts with the triangle inequality gives us:
    \begin{align*}
        d_{G^{(i-1)}}(v_{(1)},v_{(\ell')}) &\leq d_{G^{(i-1)}}(v_{(1)},v') + d_{G^{(i-1)}}(v',w')+ d_{G^{(i-1)}}(w',v_{(\ell')})\\
        &\leq \mathcal{D}' + d_{G^{(i-1)}}(v',w')+ \mathcal{D}'\\
        &\leq 2\mathcal{D}' + 2\mathcal{D}'
        \leq 4\mathcal{D}'
    \end{align*}
    Thus, the claim follows.
\end{proof}
Next, we claim that this boundary must intersect with many paths of length $\mathcal{D}'$ as removing the boundary causes all nodes in $\mathcal{Z}$ to lose the connection to a constant fraction of nodes.
\begin{claim}
    For all $z \in \mathcal{Z} \setminus \mathcal{B}$, it holds $B_{G\setminus\mathcal{B}}(z,\mathcal{D}') \leq (\nicefrac{3}{4})\cdot n$.
\end{claim}
\begin{proof}
Recall that $\mathcal{S}'$ is weak $\mathcal{D}'$-separator and therefore, it holds for all nodes $v \in V\setminus \mathcal{S}'$ by definition:
\begin{align*}
     |B_{G \setminus \mathcal{S}'}(v,\mathcal{D}')| \leq (\nicefrac{3}{4})\cdot n
\end{align*}
The boundary definition contains all nodes of $\mathcal{S}'$ that intersects with a path of length $\mathcal{D}'$ from any node in $\mathcal{Z}$.
Therefore, for all nodes $z \in \mathcal{Z}$, it holds:
\begin{align*}
     |B_{G \setminus \mathcal{S}'}(v,\mathcal{D}')| = |B_{G \setminus \mathcal{B}}(v,\mathcal{D}')|
\end{align*}
This can be asserted by a simple induction over the sets $\mathcal{P}'_1, \dots, \mathcal{P}'_{i_{\mathcal{D}'}}$ and $\mathcal{P}_1, \dots, \mathcal{P}_{i_{\mathcal{D}'}}$ that make up $\mathcal{B}$ and $\mathcal{S}$ respectively.  For easier notation, define:
    \begin{align*}
        G_{i} := G\setminus \bigcup_{j<i}\mathcal{P}_j\\
        G'_{i} := G\setminus \bigcup_{j<i}\mathcal{P}'_j
    \end{align*}
\begin{itemize}
    \item For the induction beginning, we suppose that no paths are removed. In this case, the statement holds trivially as $G'_1 = G_1$.
    \item For the induction steps, we assume that for all $z \in \mathcal{Z}$, it holds:
\begin{align*}
    B_{G_i}(z,\mathcal{D}') = B_{G'_{i}}(z,\mathcal{D}')
\end{align*}
    We now show that $B_{G_{i+1}}(z,\mathcal{D}') = B_{G'_{i+1}}(z,\mathcal{D}')$  for all $z \in \mathcal{Z}$.
    As each $P'_{i_j}$ is a subset of $P_{i_j}$, it holds that $B_{G_{i+1}}(v,\mathcal{D}') \subseteq B_{G'_{i+1}}(v,\mathcal{D}')$ as we can only reach more nodes from $z$ we remove fewer nodes.
    For contradiction, assume there is a node $u \in B_{G'_{i+1}}(z,\mathcal{D}')$ that is not in $B_{G_{i+1}}(z,\mathcal{D}')$.    
    As $u \not\in B_{G_{i+1}}(z,\mathcal{D}')$, a node on \textbf{every} path from $z$ to $u$ of length at most $\mathcal{D}'$ is removed.
    In other words, every such path intersects with some path $P_{i_j} \in \mathcal{P}_i$.
    Let $p$ be such an intersecting node.
    Clearly, $p \in B_{G_i}(z,\mathcal{D}')$ as $p$ is closer to $z$ than $u$ because it is a predecessor of $u$ on a path of length $\mathcal{D}'$.
    Therefore, any such $p$ also part of $P'_{i_j} := B_{G'_i}(z,\mathcal{D}') \cap P_{i_j}$ as, per the induction's hypothesis, it holds $B_{G'_i}(z,\mathcal{D}') = B_{G_i}(z,\mathcal{D}')$.
    Thus, all possible nodes $p$ are removed from $G'_i$ and $u$ cannot not reachable in $G'_{i+1}$ from $z$ within radius $\mathcal{D}'$.
    This is a contradiction as we assumed $u \in B_{G'_{i+1}}(z,\mathcal{D}')$.
    Therefore, it holds $B_{G_{i+1}}(v,\mathcal{D}') \supseteq B_{G'_{i+1}}(v,\mathcal{D}')$ and together with our initial observation that $B_{G_{i+1}}(v,\mathcal{D}') \subseteq B_{G'_{i+1}}(v,\mathcal{D}')$, we have $B_{G_{i+1}}(v,\mathcal{D}') = B_{G'_{i+1}}(v,\mathcal{D}')$ as claimed.
\end{itemize}
This concludes the proof of the claim as
\begin{align*}
     B_{G\setminus\mathcal{S}'}(z,\mathcal{D}') = B_{G_{i_{\mathcal{D}'}}}(z,\mathcal{D}') = B_{G'_{i_{\mathcal{D}'}}}(z,\mathcal{D}') =  B_{G \setminus \mathcal{B}}(z,\mathcal{D}')
\end{align*}
\end{proof}
Thus, the sets $\mathcal{Z}$ and $\mathcal{B}$ fulfill the properties required by the lemma. As $\mathcal{D}' = 8\mathcal{D}$, $\mathcal{B}$ consists of (at most) $k$ paths of length $32\mathcal{D}$. Further, as for each $z \in Z$ it holds
\begin{align*}
     B_{G \setminus \mathcal{B}}(z,8\mathcal{D}) =  B_{G\setminus\mathcal{S}'}(z,8\mathcal{D}) \leq (\nicefrac{3}{4})\cdot n, 
\end{align*}
the lemma follows.
\end{proof}

Next, we require the following helpful observation that is true for any distance metric on graphs
\begin{lemma}
\label{lemma:2delta}
    Let $G := (V,E,\ell)$ be a weighted graph and let $\mathcal{D}'$ be some distance parameter.
    Suppose there is a set of nodes $\mathcal{Z} \subseteq V$ of size $|\mathcal{Z}| \geq (\nicefrac{1}{c})\cdot n$ for some $c > 1$, s.t., it holds for all $v \in \mathcal{Z}$ that $|B_G(v,2\mathcal{D}')| \leq (1-\nicefrac{1}{c})\cdot n$.
    Then, \textbf{for all nodes} $v \in V$, it holds: 
    \begin{align*}
        |B_G(v,\mathcal{D}')| \leq (1-\nicefrac{1}{c})\cdot n
    \end{align*}
\end{lemma}
\begin{proof}
The lemma clearly holds for all nodes in $\mathcal{Z}$, so it suffices to analyze the remaining nodes.
Assume for contradiction that there is a node $w \in V\setminus \mathcal{Z}$ with $|B(w,\mathcal{D}')| > (1-\nicefrac{1}{c}) \cdot n$.
Then, $B(w,\mathcal{D}')$ must contain at least one node $v \in \mathcal{Z}$, which --- by the very definition of $B(w,\mathcal{D}')$ --- implies $d(v,w) \leq \mathcal{D}'$.
By the triangle inequality, it then holds for every other node $u \in B(w,\mathcal{D}')$ that
    $d(v,u) \leq d(v,w) + d(w,u) \leq 2\mathcal{D}'$.
Therefore, it follows $B(w,\mathcal{D}') \subseteq B(v,2\mathcal{D}')$
As we have a clear bound on the size of $B(v,2\mathcal{D}')$, this implies that: 
\begin{align*}
     (1-\nicefrac{1}{c}) \cdot n < |B(w,\mathcal{D}')| \leq |B(v,2\mathcal{D}')|  \leq (1-\nicefrac{1}{c}) \cdot n
\end{align*}
This is a contradiction.
\end{proof} 

% \begin{figure}
%     \centering
%     \includestandalone[width=0.6\textwidth]{Figures/Triangle_Inequality.tex}
%     \caption{An illustration showing our use of the triangle inequality in the proof of lemma \ref{lemma:helper_separator}. Note that $s,t$ are in distance of at most $32\mathcal{D}$.}
%     \label{fig:triangle_inequality}
% \end{figure}

With this insight, we can now show Lemma \ref{thm:weak_divandconquer}.
As our main analytical tool, we consider the following potential:
\begin{align*}
    \Phi_t := \left|\left\{ v \in V_t \bigm\vert |B_{G_t}(v,2\mathcal{D})| \geq (\nicefrac{7}{8})\cdot n \right\}\right|
\end{align*}
As soon as this potential drops below $(\nicefrac{7}{8})\cdot n$, we have more than $(\nicefrac{1}{8})\cdot n$ nodes with less than $(\nicefrac{7}{8})\cdot n$ nodes in distance $2\mathcal{D}$.
Therefore, by Lemma \ref{lemma:2delta}, each node has at most $(\nicefrac{1}{8})\cdot n$ nodes in distance $\mathcal{D}$ as desired.
Thus, we want to bound the time until the potential drops.

Let sets $\mathcal{Z} \subseteq V$ and $\mathcal{B}:= \{P_1, P_2, \ldots\}$ be as defined in Lemma \ref{lemma:helper_separator}.
Note that for all $z \in \mathcal{Z}$, it holds:
\begin{align}
    B_{G\setminus \mathcal{B}}(z,8\mathcal{D}) \leq (\nicefrac{3}{4})\cdot n
\end{align}
Lemma \ref{lemma:2delta} now implies that:
\begin{lemma}
    Let $G := (V, E, w)$ be a weighted {$k$-path separable} graph of $n$ nodes with (weighted) diameter $\mathcal{D}$.
    Then, there exists a set $\mathcal{B} = \{P_1, \ldots, P_\kappa\}$ of at most $k$ simple paths of length $32\mathcal{D}$ such that for all $v \in V$, it holds $B_{G \setminus \mathcal{B}}(v,4\mathcal{D}) \leq \left(\nicefrac{3}{4}\right)\cdot n$.
\end{lemma}
Therefore, for each node $v \in V$ there must be $(\nicefrac{1}{4})\cdot n$ nodes $w \in V$ such that all paths of length at most $4\mathcal{D}$ from $v$ to $w$ cross $\mathcal{B}$.
For each node $v \in V$, we denote these nodes as the cutoff nodes $\mathcal{C}(v) \subseteq B(v,4\mathcal{D})$.
It holds $\mathcal{C}(v) := V \setminus B_{G\setminus \mathcal{B}}(v,4\mathcal{D})$.
We want to show that --- as long as the potential is high --- there is a good chance to sample a path between a node $v$ and one of its cutoff nodes.
To this end, not that every node considered in the potential must have $(\nicefrac{1}{8})\cdot n$ cutoff nodes in the distance $2\mathcal{D}$.
This follows from the definition of the potential:
There are at most $(\nicefrac{1}{8})\cdot\mathcal{D}$ in distance greater than $2\mathcal{D}$. 
Even if all these nodes are cutoff nodes, there must still be $(\nicefrac{1}{8})\cdot n$ cutoff nodes in the distance $2\mathcal{D}$.
As we compute $2$-approximate shortest paths to all nodes, the approximate distance to these nodes is at most $4\cdot\mathcal{D}$.
Therefore, all cutoff nodes of $v$ are the potential endpoints of the approximate shortest path we sample.
We may not choose the shortest path (of length less than $2\mathcal{D}$), but the length of the path we pick is at most $4\mathcal{D}$ and (by definition) crosses $\mathcal{B}$.
Thus, as long as $\Phi_t \geq p\cdot n$, the probability that we sample a path that crosses $\mathcal{B}$ is at least $\nicefrac{p}{8}$.
Intuitively, this means that after carving $\tau \in \Omega(\log n)$ paths, \textbf{either} the potential has dropped, \textbf{or} there are $\Theta(\tau)$ paths that intersect with $\mathcal{B}$, w.h.p. 
Note that every time the latter happens, we remove a chunk of length $(\nicefrac{\epsilon}{2})\mathcal{D}$ from some path in $\mathcal{B}$.
However, this cannot happen too often due to these paths' bounded length of $32\mathcal{D}$; therefore, the potential must drop quickly.

\medskip

We now want to formalize this intuition.
To this end, fix a path $P := (v_1, v_2, \ldots)$ in graph $G$ and denote $\mathcal{K}_t({P})$ as the number of sets from $K_1, \ldots, K_t$ that intersect with ${P}$.
 Formally, it holds:
\begin{align*}
\mathcal{K}_t(P) := \sum_{i \leq t} \mathds{1}_{K_i \cap P \not= \emptyset}
\end{align*}
Now define the number of clusters that intersect with $\mathcal{B}$ as
\begin{align*}
    \mathcal{K}_t := \sum_{P_i \in \mathcal{B}} \mathcal{K}_t(P_i)
\end{align*}   
More formally, it holds:
\begin{lemma}
\label{lemma:sekt_oder_selters}
After carving $\tau \geq c \cdot 768$ sets $K_1, \ldots, K_\tau$, it holds:
\begin{align*}
    \pr{\left\{\Phi_\tau \geq (\nicefrac{7}{8})\cdot n \right\} \cap \left\{\mathcal{K}_\tau \leq \nicefrac{\tau}{256}\right\}} \leq e^{-c}
\end{align*}
\end{lemma}
\begin{proof}
First, we denote $\delta\mathcal{K}_t := \mathcal{K}_t- \mathcal{K}_{t-1}$
as the difference in the number of intersections from step $t-1$ to step $t$.
Note that this difference is always non-negative as the number of balls intersecting with the separator can not decrease.
Equipped with this definition, we want to count the number of rounds in which we either increase $\mathcal{K}_t$, or $\Phi_t$ falls below $\nicefrac{n}{8}$.
We refer to these rounds as \emph{good} steps.
%To this end, we introduce the notion of a \emph{good} step where we either increase $\mathcal{K}_t$ or $\Phi_t$ falls below $\nicefrac{n}{8}$.
Clearly, after at most $\rho\tau$ goods steps (with $\rho \leq 1$) either $\rho\tau$ balls intersect with the boundary, or it holds $\Phi_t \leq \nicefrac{n}{8}$.
To bound the number of good steps, we define the following indicator variable for good steps:
\begin{align*}
    Y_t := \mathds{1}_{\{\Phi_t \leq \nicefrac{n}{8}\} \vee \{\delta \mathcal{K}_t \geq 1\}}
\end{align*}
Further, let $\mathcal{Y}_\tau := \sum_{i=1}^\tau Y_i$ the total number of good steps until some step $\tau>0$.
%Based our arguments above, we can conclude that that $\pr{\Phi_t \geq \nicefrac{n}{2}} \leq \pr{\mathcal{Y}_t \leq k\cdot\nicefrac{4}{\epsilon}}$.
Next, we will bound the probability that a given step is good to a small constant, namely:
\begin{claim}
\label{claim:y_prob}
Let $H_{t-1}$ denote all choices made by the algorithm in steps $1, \dots, t-1$.
Then, it holds: 
\begin{align*}
\E{Y_t \mid H_{t-1}} \geq \nicefrac{1}{64}    
\end{align*}
\end{claim}
\begin{proof}
To simplify notation, define $\Phi'_{t-1}$ to be the number of critical nodes at the beginning of step $t$ given $H_{t-1}$, i.e., it holds:
\begin{align*}
    \pr{\Phi'_{t-1} = x } = \pr{\Phi_{t-1} = x \mid H_{t-1} }
\end{align*}
Analogously, $\Phi'_{t}$ will be the number of critical nodes in step $t$ conditioned on $H_{t-1}$.
By the law of total expectation, it holds:
\begin{align*}
    \E{Y_t \mid H_{t-1}} =& \pr{\Phi'_{t-1} < \nicefrac{n}{8}}\cdot\E{Y_t \mid \Phi'_{t-1} < \nicefrac{n}{8}} + \\
    &\pr{\Phi'_{t-1} \geq \nicefrac{n}{8}}\cdot\E{Y_t \mid \Phi'_{t-1} \geq \nicefrac{n}{8}}
\end{align*}
This implies:
\begin{align*}
\E{Y_t \mid \Phi'_{t-1} < \nicefrac{n}{8}} = 1 \geq \E{Y_t \mid \Phi'_{t-1} < \nicefrac{n}{8}}    
\end{align*}
The inequality follows from the fact that $Y_t$ is a binary random variable.
Therefore, it holds:
\begin{align}
   \E{Y_t\mid H_{t-1}} &= \pr{\Phi'_{t-1} < \frac{n}{8}} \cdot 1 + \pr{\Phi'_{t-1}\geq \frac{n}{8}}\E{Y_t \mid \Phi'_{t} \geq \frac{n}{8}}\\
   &\geq \pr{\Phi'_{t-1} < \frac{n}{8}}  \E{Y_t \mid \Phi_{t} \geq \frac{n}{8}} + \pr{\Phi'_{t-1}\geq \frac{n}{8}}\E{Y_t \mid \Phi_{t} \geq \frac{n}{8}}\\
   &\geq \left(\pr{\Phi'_{t-1} < \frac{n}{8}}+\pr{\Phi'_{t-1}\geq \frac{n}{8}}\right) \cdot \E{Y_t \mid \Phi_{t} \geq \frac{n}{8}}\\
   &= \E{Y_t \mid \Phi'_{t-1} \geq \frac{n}{8}}\\
   &\geq  \E{\delta\mathcal{K}_t \geq 1 \mid \Phi'_{t-1} \geq \frac{n}{8} }\label{eq:ytok}
\end{align}
Let $P_t := (v_t, \ldots, w_t)$ be the path sampled in step $t$ and let $v_t$ and $w_t$ be its respective start and endpoint.
As the algorithm only removes nodes and edges, the path $P$ must have also existed in the original graph $G$.
Therefore, the path must cross $\mathcal{B}$ if $v_t$ is a critical node and $w_t$ is one of its cutoff nodes.
Hence $\mathcal{K}_t$ increases by one in this case.
Denote this event as $\{v_t \underset{\mathcal{B}}{\leadsto} w_t\}$.
Further, denote 
\begin{align*}
    C_t = \left\{ v \in V_t \bigm\vert |\Cutoff{v} \cap B_{G_t}(v,2\mathcal{D})| \geq \nicefrac{n}{8} \right\}
\end{align*}
as the critical nodes and note that it holds:
\begin{align*}
    \E{\delta\mathcal{K}_t \geq 1 \mid \Phi'_{t-1}} &= \E{v_t \underset{\mathcal{B}}{\leadsto} w_t \mid \Phi'_{t-1}} \\
    &\geq \sum_{v \in V} \pr{v \in C_t \mid \Phi'_{t-1} } \cdot \pr{v_t = v} \cdot \pr{w_t \in \Cutoff{v} \mid v_t = v}\\
    &\geq \frac{1}{n} \cdot\sum_{v \in V} \pr{v \in C_t \mid \Phi'_{t-1}}  \cdot \pr{w_t \in \Cutoff{v} \mid v_t = v}
\end{align*}
Recall that by definition, each critical node has at least $\nicefrac{n}{8}$ cutoff nodes in distance $2\mathcal{D}$.
By the properties of the approximate shortest path algorithm, the approximate distance to all these nodes is $4\mathcal{D}$ and, thus, all these nodes are considered in the sampling.
Further, as we use approximate shortest paths that can only overestimate the actual distance, we sample from a subset of all nodes in the distance $4\mathcal{D}$.
Note that there can obviously be at most $n$ nodes in distance $(1+\epsilon)2\mathcal{D}$. 
Therefore, the probability that one $v_t$'s cutoff nodes are picked as the path's endpoint $w_t$ is:
    \begin{align*}
        \pr{w_t \in \Cutoff{v} \mid v_t = v} &:= \frac{\left|\Cutoff{v} \cap B_{G_{t-1}}(v,2\mathcal{D})\right|}{\left|B_{G_{t-1}}(v,4\mathcal{D})\right|} \geq \frac{n}{8n} = \frac{1}{8}  
    \end{align*}
    Putting this back in the previous formula, we get:
     \begin{align}
         \E{v_t \underset{\mathcal{B}}{\leadsto} w_t \mid \Phi'_{t-1}}  &\geq \frac{1}{n} \cdot\sum_{v \in V} \pr{v \in C_t \mid \Phi'_{t-1}}  \cdot \pr{w_t \in \Cutoff{v} \mid v_t = v}\\
        &\geq  \frac{1}{n} \cdot \sum_{v \in V} \pr{v \in C_t \mid \Phi'_{t-1}} \cdot \frac{1}{8}
        = \frac{\Phi'_{t-1}}{8n}\label{eq:exp_change}
    \end{align}
This proves the lemma as it implies that 
\begin{align*}
     \E{\delta\mathcal{K}_t  \mid \Phi'_{t-1}  \geq \frac{n}{16} }&\geq \sum_{\eta=\nicefrac{n}{8}}^n \pr{\Phi'_{t-1} = \eta \mid \Phi'_{t-1} \geq \nicefrac{n}{8}}   \E{v_t \underset{\mathcal{B}}{\leadsto} w_t \mid \Phi'_{t-1}=\eta}\\ 
     &\geq \sum_{\eta=\nicefrac{n}{8}}^n \frac{\pr{\{\Phi'_{t-1} = \eta\} \cap \{\Phi'_{t-1} \geq \nicefrac{n}{8}\}}}{\pr{ \Phi'_{t-1} \geq \nicefrac{n}{8}}}   \E{v_t \underset{\mathcal{B}}{\leadsto} w_t \mid \Phi'_{t-1}=\eta}\\ 
               &\geq \sum_{\eta=\nicefrac{n}{8}}^n \frac{\pr{\Phi'_{t-1} = \eta}}{\pr{ \Phi'_{t-1} \geq \nicefrac{n}{8}}} \E{v_t \underset{\mathcal{B}}{\leadsto} w_t \mid \Phi'_{t-1}=\eta}\\ 
    \text{By Ineq. }\eqref{eq:exp_change}:\\
          &\geq \sum_{\eta=\nicefrac{n}{8}}^n \frac{\pr{\Phi'_{t-1} = \eta}}{\pr{ \Phi'_{t-1} \geq \nicefrac{n}{8}}} \frac{\eta}{8n}\\ 
     \text{Using }\eta \geq \nicefrac{n}{8}:\\
     &\geq \sum_{\eta=\nicefrac{n}{8}}^n \frac{\pr{\Phi'_{t-1} = \eta}}{\pr{ \Phi'_{t-1} \geq \nicefrac{n}{8}}} \frac{n}{8\cdot 8 \cdot n}\\
     &\geq\frac{1}{64} \cdot \sum_{\eta=\nicefrac{n}{8}}^n \frac{\pr{\Phi'_{t-1} = \eta}}{\pr{ \Phi'_{t-1} \geq \nicefrac{n}{8}}}\\ 
     \text{As }\pr{ \Phi'_{t-1} \geq \nicefrac{n}{8}} &= \sum_{\eta=\nicefrac{n}{8}}^n  \pr{\Phi'_{t-1} = \eta}:\\
     &\geq \frac{1}{64}
\end{align*}
Thus, by Inequality \eqref{eq:ytok}, it holds:
\begin{align*}
     \E{Y_t\mid H_{t-1}} \geq \E{\delta\mathcal{K}_t \geq 1 \mid \Phi'_{t-1} \geq \frac{n}{8} } \geq \frac{1}{64}
\end{align*}
This proves the claim.
\end{proof}
Note that this probability bound holds regardless of all events before step $t$.
We need the following slightly generalized version of the Chernoff inequality to finalize the proof.
\begin{lemma}[Generalized Chernoff Bound, cf. Theorem 3.52 in \cite{ScheidelerHabil}]   
\label{lemma:general_chernoff}
    Let $X_1, \ldots, X_n \in \{0,1\}$ be a series of (not necessarily independent) binary random variables and define $X := \sum_{i=1}^n X_i$.
    Suppose, there is a value $\rho \in [0,1]$, s.t., for all subsets $X_{i_1}, \ldots, X_{i_j}$ with $1 \leq j \leq n$, it holds:
    \begin{align*}
        \E{\prod_{k = i_1, \ldots, i_j} X_k} \geq \rho^j.
    \end{align*}
    Then, for any $\delta \leq 1$, it holds:
    \begin{align*}
        \pr{X < (1-\delta)\rho n} \leq e^{-{\frac{\delta\rho n}{3}}}
    \end{align*}
\end{lemma}
Let now $i_1, \ldots, i_j$ be some subset of the $\tau$ steps that we make.
Note that by the fact that all $Y_t$ are binary random variables and the chain rule of conditional expectations, it holds:
\begin{align*}
    \E{\prod_{k = i_1}^{i_j} Y_k} &:= \pr{\bigcap_{k = i_1}^{ i_j} Y_k = 1} = \prod_{k = i_1}^{i_j}\pr{Y_k = 1 \mid \bigcap_{k' = i_1}^{k-1} Y_{k'} = 1}\\
    &= \prod_{k = i_1}^{i_j}\E{Y_k \mid \bigcap_{k' = i_1}^{k-1} Y_{k'} = 1} \geq \prod_{k = i_1}^{i_j}\frac{1}{64} = \left(\frac{1}{64}\right)^j
\end{align*}
Here, we exploited that $\left\{\cap_{k' = i_1}^{k-1} Y_{k'} = 1\right\}$ only depends on events prior to step $k$ and used the bound of Claim \ref{claim:y_prob}.
Thus, the conditions of Lemma \ref{lemma:general_chernoff} apply to $\mathcal{Y}_\tau$ by choosing $\rho = \nicefrac{1}{64}$.
Therefore, it holds for all $\tau \geq 1$ and $\delta \leq 1$ that:
\begin{align*}
    \pr{\mathcal{Y}_\tau \leq (1-\delta)\rho\tau}
    \leq e^{-{\frac{\delta \rho \tau}{3}}}
\end{align*}
For $\tau \geq c\cdot 768$, we get the following bound:
\begin{align*}
    \pr{\mathcal{Y}_\tau \leq \frac{\tau}{256}}&= \pr{\mathcal{Y}_\tau \leq \left(1-\frac{1}{2}\right)\frac{\tau}{64}}
    = \pr{\mathcal{Y}_\tau \leq \left(1-\frac{1}{2}\right)\rho\tau}\\
    &\leq e^{-{\frac{(\nicefrac{1}{2}) \rho \tau}{3}}} = e^{-{\frac{\rho \tau}{6}}}
    = e^{-{\frac{\tau}{768}}} \leq e^{-c}
\end{align*}
This proves the lemma.
\end{proof}
This verifies our intuition that we sample many paths that cross $\mathcal{B}$ unless the potential drops down by a constant factor.
However, to prove Lemma, we need more than that.
Further, we need an upper bound on the number of intersections with $\mathcal{B}$. For this, it holds:
\begin{lemma}
\label{lemma:path_limit}
Consider any path $P$ of length $32\mathcal{D}$ in $G$. Then for any $t \geq 0$, it holds $\mathcal{K}_t(P) \leq O(\epsilon^{-1})$.
\end{lemma}
\begin{proof}
    To this end, we introduce the notion of \emph{chunks}.
A chunk is a non-empty subpath of the path, s.t., the distance between the first node of two consecutive chunks is at least $\epsilon\mathcal{D}$.
Formally, we can recursively define them as follows.
\begin{definition}[Chunk]
    Let $P := (v_1, \ldots, v_m)$ be a path in graph $G$ and let $\epsilon > 0$ be a distance parameter.
    Then, the chunks $C^{(\epsilon)}(P) := \left(C^{(\epsilon)}_1, \ldots, C^{(\epsilon)}_{m´}\right)$ of path $P$ are partition of $P$ into connected subpaths.
    The chunks are recursively defined as follows:
    \begin{enumerate}
        \item The first chunk $C_1^{(\epsilon)}$ begins with node $v_1$.
        \item The first node of chunk $C_{i+1}$ is the first node in distance (greater than) $\epsilon\mathcal{D}$ to the first node of chunk $C_i$. 
        \item A chunk $C_i$ contains all nodes between its first node and the first node of $C_{i+1}$. The last chunk contains its first node until the end of the path.
    \end{enumerate}
\end{definition}
\begin{figure}[ht]
    \centering
    \begin{tikzpicture}
        \graph[grow right=2cm] {
        "$v_1$"[circle, draw, thick] --[thick] "$v_2$"[circle, draw, thick] --[dashed] "$v_3$"[circle, draw, thick] --[thick] "$v_4$"[circle, draw, thick]  --[ thick] "$v_5$"[circle, draw, thick]  --[dashed] "$v_6$"[circle, draw, thick] --[ thick] "$v_7$"[circle, draw, thick] --[dashed] "$v_8$"[circle, draw, thick];
        "$v_1$" ->[bend left = 45, "$\geq \epsilon\mathcal{D}$"] "$v_3$";
        "$v_3$" ->[bend left, "$\geq \epsilon\mathcal{D}$"] "$v_6$";
        "$v_6$" ->[bend left = 45, "$\geq \epsilon\mathcal{D}$"] "$v_8$";
        };
    \draw (0,-1) -- node[below] {$C_1^{(\epsilon)}$} (2,-1);
    \draw (0,-1.1) -- (0,-0.9);
    \draw (2,-1.1) -- (2,-0.9);
    
    \draw (4,-1) -- node[below] {$C_2^{(\epsilon)}$} (8,-1);
    \draw (4,-1.1) -- (4,-0.9);
    \draw (8,-1.1) -- (8,-0.9);
    
    \draw (10,-1) -- node[below] {$C_3^{(\epsilon)}$} (12,-1);
    \draw (10,-1.1) -- (10,-0.9);
    \draw (12,-1.1) -- (12,-0.9);
    
    \draw (13.5,-1) -- node[below] {$C_4^{(\epsilon)}$} (14.5,-1);
    \draw (13.5,-1.1) -- (13.5,-0.9);
    \draw (14.5,-1.1) -- (14.5,-0.9);
    
\end{tikzpicture}
    \caption{An example for the chunks $C^{(\epsilon)}(P)$ of distance $\epsilon$ of a path $P := v_1, \ldots, v_8$.
    The solid lines denote edges between two nodes of the same chunk, the dashed lines denote edges between chunks.}
    \label{fig:enter-label}
\end{figure}
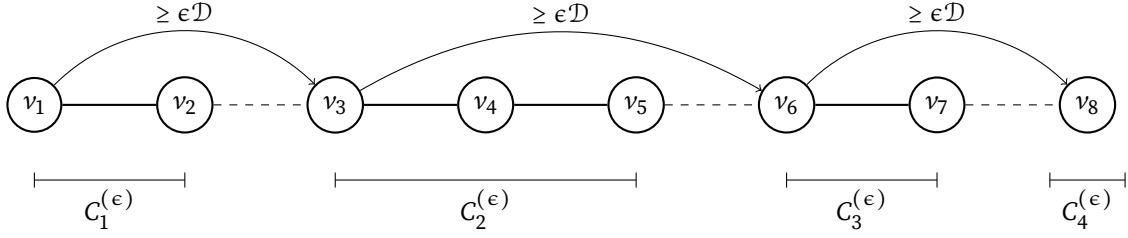
Further, note that every time a random path $P_t$ crosses $P$, we remove all nodes in the chunk it crosses.
In particular, it does not matter if some nodes \emph{between} two nodes of a chunk have already been removed as we always consider the distance with regard to $G$.
Thus, the number of chunks is an upper bound on the number of clusters that intersect with the separator in \textbf{any} weak ball carving process.

Clearly, a path $P$ of length $32\mathcal{D}$ consists of $128\epsilon^{-1}$ chunks of length $\nicefrac{\epsilon}{4}\mathcal{D}$.
This can verified as follows: 
Denote by $C(P)$ the chunks of length $\nicefrac{\epsilon}{4}\mathcal{D}$ in $P$.
For the sake of creating a contradiction, assume that there are more than $\nicefrac{128}{\epsilon}$ chunks.
     Denote this number as 
     \begin{align*}
        c := \left| C(P)\right|
     \end{align*}  
The chunks $C(P)$ can be ordered $C_{(1)}, \ldots, C_{(c)}$ based on their position in path $P$, s.t., chunks with higher order are closer to the endpoint of the path.
We denote the first node of chunk $C_{(i)}$ that is part of the boundary as $v_{(i)}$.
Recall that the distance between the first nodes of two neighboring chunks is at least $\nicefrac{\epsilon}{4}\mathcal{D}$.
If the chunks are not neighboring, the distance can only be greater.
Thus, it holds:
     \begin{align*}
         d_{P}(v_{(i)},v_{(i+1)}) \geq \nicefrac{\epsilon}{4}\mathcal{D}
    \end{align*}
Let now $v$ be the last node in the first chunk $C_{(1)}$ and $w := v_{(c)}$ be the first node in the last chunk $C_{(c)}$ that intersect with the boundary.
By the definition of chunks, the distance between $v$ and $w$ on path $P_{i_j}$ is at least:
     \begin{align*}
        d_{P}(v,w) &= d_{P_{i_j}}(v,v_{(2)}) + \sum_{i=2}^{c} d_{P_{i_j}}(v_{(i)},v_{(i+1)})\\
        &\geq d_{P}(v,v_{(2)}) + \sum_{i=2}^c \nicefrac{\epsilon}{4}\mathcal{D}\\
        & = d_{P}(v,v_{(2)}) + (c-1) \nicefrac{\epsilon}{4}\mathcal{D}\\
        & \geq  d_{P}(v,v_{(2)}) + \left(\nicefrac{128}{\epsilon}+1-1\right) \epsilon\mathcal{D}\\
        & = d_{P_{i}}(v,v_{(2)}) + 32\mathcal{D}\\
        & >32\mathcal{D}
    \end{align*}
This is a contradiction as $P$ is of length $32\mathcal{D}$.
\end{proof}

Finally, suppose for contradiction that after $O(\epsilon^{-1}\cdot k \cdot \log n)$ iterations, the potential is still above $(\nicefrac{7}{8})\cdot n$.
Then, $\mathcal{K}_\tau \in \Omega(\epsilon^{-1}\cdot k \cdot \log n)$ paths must have crossed $\mathcal{B}$, w.h.p.
However, by Lemma \ref{lemma:path_limit}, $\mathcal{K}_\tau$ is bounded by $O(\epsilon^{-1} \cdot k)$.
A contradiction.
Thus, after $O(\epsilon^{-1}\cdot k \cdot \log n)$ iterations, the potential must be low, w.h.p., which proves the lemma.
\medskip

\paragraph{\textbf{Complexity (Proof of Lemma \ref{lemma:separator_runtime}):}} All five steps of the algorithm clearly contain either $\Tilde{O}(1)$ minor aggregations or a $2$-approximate shortest path computation.
We go through them one by one.
Step $1$ is a simple aggregation of the maximum of a set of random values. 
Step $2$ is a $2$-approximate shortest path computation.
Step $3$ consists of another maximum aggregation and $\Tilde{O}(1)$ aggregations to implement the descendant aggregation from Lemma \ref{lemma:tree_operations}.
Step $4$ is another $2$-approximate shortest path computation from a set.
Step $5$ is purely local operation.
Therefore, a step of the algorithm can be executed within $\Tilde{O}(1)$ rounds of minor aggregation and approximate shortest path computations as claimed.
This proves Lemma \ref{lemma:separator_runtime}.

\section{Full Analysis of \autoref{lemma:backbone_clustering_k_path}}
\label{sec:appendix_backbone}

In this section, we prove the following technical lemma which that that can can efficiently constuct backbone clusters using approximate shortest paths and minor aggregations.

\backbone*

The algorithm behind this lemma is a divide-and-conquer algorithm executed in parallel on all connected components of $G$.
For the \emph{divide} step, we use the \LDD promised by Theorem \ref{thm:clustering_general} to create disjoint subgraphs of diameter ${\mathcal{D}}' \in O(\mathcal{D}\log^2 n)$. 
The \emph{conquer} step works in five synchronized phases that compute and remove a weak $({\mathcal{D}'},\epsilon)$-separator from each of these subgraphs.
Further, we create a backbone cluster from each separator.
This works in two steps, by first adding all nodes at a random distance and then applying the blurry ball growing procedure from Lemma \ref{lemma:bbg}.
We repeat this process until all subgraphs are empty.
As we remove a $\mathcal{D}'$-separator in each step and create clusters of diameter $\mathcal{D}'$, this requires $O(\log n)$ iterations.

After this intuition, we move to the detailed description of the algorithm.
For this, we define some useful notations/constants, namely
\begin{align*}
    \mathcal{D}' &= 100 \cdot \mathcal{D} \cdot \log^2 n \\
    \epsilon &= \frac{1}{10000\log^2 n}\\
    \mathcal{D}_{BC} &= \frac{\mathcal{D}}{4}\\
    \rblur &:= \frac{\mathcal{D}}{\cblur\cdot\log\log n} \leq \frac{\mathcal{D}}{8}
\end{align*}
Here, $\cblur > 8$ is a large constant that will be determined in the analysis.
In the following, we will consider a single recursive step $t \in [1, O(\log n)]$ of the algorithm.
We will slightly abuse notation and denote the set of all nodes that have already been added to some cluster as $\mathcal{K}_t$
Further, we call a node $v \in V \setminus \mathcal{K}_t$, which is not yet part of a cluster, an \emph{uncharted} node.
A single recursive \emph{conquer} step works on the graph $G_t = G \setminus \mathcal{K}_t$ of uncharted nodes works as follows:

    \paragraph*{\textbf{(Step 1) Create Partitions:}} Let $C_1, \ldots, C_N$ be the connected components of $G_t$. 
    Compute an \LDD with diameter $\mathcal{D}'$ in each subgraph $C_i$ using the algorithm from Theorem \ref{thm:clustering_general}.
    The resulting partitions are connected subgraphs $P_1, \dots, P_{N'}$ with diameter $\mathcal{D}'$.
    \paragraph*{\textbf{(Step 2) Create Weak Separators in All Partitions:}}  In each partition $P_i$, we compute a weak  $(\mathcal{D}',\epsilon)$-separator $S_i$ using the algorithm from Theorem \ref{thm:distributed_weak_separator}.
    As the distance parameter for the separator, we choose ${\mathcal{D}'}$, and for the approximation parameter, we pick $\epsilon$.
    % This results in a weak $O(k \log^2 n)$-path $(\mathcal{D}',\epsilon)$-separator that consists of $O(k \log^2 n)$ paths of length at most $\mathcal{D}'$ and nodes in distance at most $\epsilon\mathcal{D}' = \frac{\mathcal{D}}{100}$ to these paths. 
\paragraph*{\textbf{(Step 3) Create Random Ball Around Separator:}}  We carve a ball with random diameter $O(X_{i}\cdot\mathcal{D}_{BC})$ where $X_{i} \sim \mathsf{Texp}(4\cdot\log\log n)$.
To this end, we perform a $(1+\epsilon)$-approximate \SetSSP for $S_i$.
    Thereby, we obtain an a $(1+\epsilon)$-approximate \SetSSP tree $T_{i}$ and all nodes $w \in P_{i}$ know a value $d_{T_i}(w,S_i)$, which is its $(1+\epsilon)$-approximate distance to $S_i$ in $P_{i}$.
    After that, we mark every node in the set:
    \begin{align*}
        \kexp(S_i) := \left\{w \in P_{i} \mid d_{T_i}(w,S_i) \leq (1+\epsilon) \cdot X_{i} \cdot \mathcal{D}_{BC}\right\}
    \end{align*}
\paragraph*{\textbf{(Step 4) Blur the Ball:}} Finally, we apply the BBG procedure of Lemma \ref{lemma:bbg} to the ball $\kexp(S_i)$ in each partition.
As the distance parameter for this procedure, we choose $\rblur$.
Thus, for each set $\kexp(S_i)$ we obtain the superset $\kblur(S_i) := \textsf{blur}\left(\kexp(S_i),\rblur\right)$. Here, $\textsf{blur}(S,\rho)$ is an application of BBG from Lemma of \ref{lemma:bbg} to a set $S$ with parameter $\rho$.
    \paragraph*{ \textbf{(Step 5) Prepare Next Recursion:}} We choose each $K_i = \kblur(S_i)$ as a backbone cluster and add We $K_i$ to $\mathcal{K}$. Each node in $\kblur(S_i)$ marks itself as inactive and removes itself from future iterations.
All edges between active nodes and nodes in $\kblur(S_i)$ are marked as \emph{cut}. 

\bigskip

We will now present the proof of Lemma \ref{lemma:backbone_clustering_k_path}.
Recall we want to construct an $\left(O(\log{(k\log n)}), O(\log^2), O(k\log^2 n)\right)$-backbone clustering with pseudo-diameter $\mathcal{D}_{BC} = (\nicefrac{4}{100})\cdot\mathcal{D}$ for $G$.
The proof is divided into four parts.
We begin by showing that the algorithm produces a clustering that fulfills all conditions of a backbone clustering, namley the backbone property (Lemma \ref{lemma:k_path_paths}), the pseudo-diameter (Lemma \ref{lemma:k_path_diameter}), and the cutting probability (Lemma \ref{lemma:k_path_cutting}). 
Finally, we prove the time complexity in Lemma \ref{claim:k_path_complexity}.
Thus, altogther Lemma \ref{claim:k_path_complexity} and Lemmas \ref{lemma:k_path_paths}, \ref{lemma:k_path_diameter}, and \ref{lemma:k_path_cutting} imply Lemma \ref{lemma:backbone_clustering_k_path}.

\paragraph{Backbone Property} First, we note that, indeed, each cluster has a backbone that consists of a few short paths.
However, it follows directly from the construction.
For the sake of completeness, we note that it holds:
\begin{lemma}[Backbone Property]
\label{lemma:k_path_paths}
   Each cluster contains a backbone $\mathcal{B}_i$ consists of at most $O(k\log^3 n)$ paths of length $O(\mathcal{D}\log^2 n)$, w.h.p.
\end{lemma}
This follows directly from the guarantees of our weak separator.
Note that the \LDD between each step ensures that each subgraph has a diameter $\mathcal{D}'$.
Recall that by Theorem, we can construct a weak $O(\epsilon^{-1}\cdot k \cdot \log n)$-path $(\mathcal{D},\epsilon)$-separator, for a weighted $k$-path separable graph $G := (V,E,\ell)$ with weighted diameter smaller than $\mathcal{D}'$. 
For $\epsilon \in \Theta(\nicefrac{1}{\log^2 n})$, this separator has $O(k \log^3 n)$ paths of length $O(\mathcal{D}\log^2 n)$, w.h.p.
We choose these paths as the backbone, which proves the lemma.
\paragraph{Pseudo Diameter} Second, we observe the \emph{diameter} of each cluster, i.e., the distance of each node to the closest backbone is $\mathcal{D}$. 
It holds:
\begin{lemma}[Pseudo-Diameter]
\label{lemma:k_path_diameter}
    Each node $v \in K(\mathcal{B}_i)$ is in distance at most $\mathcal{D}$ to $\mathcal{B}_i$.
\end{lemma}
\begin{proof}
For the second statement, consider the $O(k\log^3 n)$ paths in the separator and denote this set as $\mathcal{B}_i$.
Each node in separator $S_i$ is in the distance at most $\epsilon\mathcal{D}'$ to one of these paths.
This follows from the definition of a weak $(\mathcal{D}',\epsilon)$-separator. 
Our choice of $\epsilon$ and $\mathcal{D}'$ yields:
\begin{align}
    \epsilon\mathcal{D}' \leq \frac{100 \mathcal{D} \log^2 n}{10000 \cdot \log^2 n} = \frac{\mathcal{D}}{100} \leq \frac{\mathcal{D}}{4} = \mathcal{D}_{BC}
\end{align}
Further, each node in $\kexp(S_i)$ is in the distance at most $(1+\epsilon) \cdot X_t \cdot \mathcal{D}_{BC} \leq 2\mathcal{D}_{BC}$ as to $S_i$.
This follows because $X_t \in [0,1]$ and $\epsilon \leq 1$.
Finally, each node $v \in \kblur(S_i)$ is in the distance at most $\mathcal{D}_{BC}$ to $\kexp(S_i)$ by our choice of $\rho$ and guarantees of Lemma \ref{lemma:bbg}. 
Following Lemma \ref{lemma:bbg}, the process adds nodes in distance $\frac{\rho}{1-\alpha}$ where $\alpha \in O(\nicefrac{\log\log n}{\log n})$.
For a large enough $n$ and $\cblur \geq 4$, we have
\begin{align}
    \frac{\rho}{1-\alpha} \leq \frac{\mathcal{D}}{\cblur\cdot\log\log n(1-O(\nicefrac{\log\log n}{\log n}))} \leq \frac{\mathcal{D}}{4} = \mathcal{D}_{BC}
\end{align}
Summing up, let $\mathcal{B}_i$ the set of paths in $S_i$, then it holds:
\begin{align}
    d(u, \mathcal{B}_i) &\leq d(u, \kexp(S_i)) + \mathbf{max}_{v \in \kexp(S_i)} d(v, \mathcal{B}_i)\\ 
     &\leq d(v, \kexp(S_i)) + \mathbf{max}_{v \in \kexp(S_i)} d(v, S_i) + \mathbf{max}_{w \in S_i} d(w,\mathcal{B}_i)\\
     &\leq \mathcal{D}_{BC}+ 2\mathcal{D}_{BC} + \mathcal{D}_{BC}\\
     &= \frac{4\mathcal{D}}{4} = \mathcal{D}
\end{align}
Therefore, the total distance from a node in $\kblur(S_i)$ to a path in $S_i$ is at most $\mathcal{D}$.
\end{proof}

\paragraph{Cut Probability} Next, we show the cut probability, which perhaps has the most challenging proof of this chapter.
We prove that the following holds:
\begin{lemma}
\label{lemma:k_path_cutting}
    An edge $z \in E$ is cut with probability at most $O\left(\frac{\alpha\ell_z}{\mathcal{D}}\right)$.
\end{lemma}

For each edge $z \in E$, two operations can cut an edge in each iteration.
The decomposition from Theorem \ref{thm:clustering_general} and the ball $\kblur(S_i)$, which is created through \emph{blurry ball growing}.
The other possibility is that $z$ is cut by $\kblur(S_i)$.
Again, there are at most $O(\log n)$ sets $S_i$ whose ball $\kblur(S_i)$ can potentially cut it. 
By our choice of $\rho$, the probability for the ball cut to $z$ is $O\left(\nicefrac{\ell_z\log\log n}{\mathcal{D}}\right)$.
Therefore, we cannot simply use the union bound to sum up the probabilities over all iterations.
However, we can exploit that $S_i$ must be \emph{close}, i.e., within distance $O(\rho)$, to either endpoint of $z$ if $\kblur(S_i)$ can possibly cut it.
In such a case, both endpoints $z$ are already added $\kexp(S_i)$ with constant probability, i.e., they are safe before the blur procedure is even executed.
This follows from the properties of the truncated exponential distribution.
Thus, on expectation, there will only be a constant number of tries before $z$ is either cut or safe.
This is sufficient for our probability bound.

Having established the rough idea, let us now prove the lemma in detail.
We define the so-called $(\lambda, \cblur, \cldd)$-bounded Blurry Ball Carving Processes that resembles our algorithm:
\begin{definition}[$(\lambda, \cblur, \cldd)$-bounded Blurry Ball Carving Process]
\label{def:ball_carving}
Let $G := (V,E, w)$ a weighted graph and let $\mathcal{D}_{BC}$ be a distance parameter.
Then, a $(\lambda, \cblur, \cldd)$-bounded blurry ball carving process with parameters $\lambda \geq 1$ and $\cblur, \cldd \geq 0$ creates a series of subgraphs $G := G_0, G_1, \ldots$, vertex sets $A_1, A_2, \ldots$, and clusters $K_1, K_2, \ldots$, s.t., it holds: 
\begin{enumerate}
    \item $K_t \supseteq A_t$ is a superset of the ball  $B_{G_{t-1}}(A_t, X_t\cdot\mathcal{D}_{BC})$ with $X_t \sim \mathsf{Texp}(\lambda)$, i.e., it holds:
    \begin{align}
        K_t \supseteq B_{G_{t-1}}(A_t, X_t\cdot\mathcal{D}_{BC}) \label{bbc_property_1a}\tag{\textbf{Property 1.a}}
    \end{align}
    Further, it holds that:
    \begin{align}
        K_t &\subseteq B_{G_{t-1}}\left(A_t, \left(X_t+\frac{1}{2\cdot\cblur\cdot\lambda}\right) \cdot \mathcal{D}_{BC} \right) \label{bbc_property_1b}\tag{\textbf{Property 1.b}}
    \end{align}
    For all edges $(v,w) := z \in E_t$ of length $\ell$, it holds:
    \begin{align}
        \pr{v \in K_t, w \not\in K_t} \in O\left(\frac{\cblur\cdot\lambda\cdot \ell}{\mathcal{D}_{BC}}\right) \label{bbc_property_1c}\tag{\textbf{Property 1.c}}
    \end{align}
    \item $G_t \subseteq G_{t-1} \setminus K_t$ is a (random) decomposition of $G_{t-1} \setminus K_t$ that only removes edges.
    Further, for each step $t$ there is value $\alpha(G_{t-1})$ that depends only on $G_{t-1}$ and $t$, s.t., it holds:
    \begin{align}
        \pr{e \not\in E(G_t) \mid e \in E(G_{t-1} \setminus K_t)} &\leq \alpha(G_{t-1}) \cdot \frac{\ell_e}{\mathcal{D}_{BC}} \label{bbc_property_2a}\tag{\textbf{Property 2.a}}
    \end{align}
    Finally, let $\Omega_t$ be the set of all sequences of graphs that can be created by the process until step $t$, then
    \begin{align}
        \mathbf{max}_{T < n} \mathbf{max}_{G_1, \dots, G_T \in \Omega_T} \sum_{t=1}^T {\alpha(G_t)} &\leq c_{\mathsf{ldd}}\cdot \lambda \label{bbc_property_2b}\tag{\textbf{Property 2.b}}
    \end{align}
\end{enumerate}
If we choose $K_t$ based on the (approximate) distances in $G$ and not in $G_{t-1}$
%, i.e., 
%\begin{align}
%        K_t := \left\{ w \in V_{t-1} \mid \Tilde{d}_{G}(A_t,w) \leq (1+\epsilon) \cdot X_t \cdot \mathcal{D} \right\},
%\end{align}
we call it a \textbf{weak} ball carving process.
Note that the set $A_t \subseteq V_{t-1}$ can be freely chosen from the unclustered nodes. 
\end{definition}

\begin{lemma}
The algorithm is a $(\lambda, \cblur, \cldd)$-bounded blurry ball carving process with $\lambda = 4\log\log n$, $\cblur = 250$, and $\cldd \in O(1)$.
\end{lemma}
\begin{proof}
We prove that it fulfills all $5$ properties.
The sets $A_1, A_2, \ldots$ are the separators that we sample in each recursive step.
    \begin{itemize}
        \item \ref{bbc_property_1a}: This property is fulfilled because $\kexp$ adds all nodes in approximate distance $ X_{i} \cdot \frac{\mathcal{D}}{100}$.
        Let $w$ a node in distance at most $ X_{i} \cdot \frac{\mathcal{D}}{100}$ to $S_i$.
        By the guarantees of our approximation algorithm, it holds:
        \begin{align*}
               d_{P_i}(w,A_t) \leq (1+\epsilon)\cdot d_{T_i}(w,S_i) \leq  (1+\epsilon) \cdot X_{i} \cdot \mathcal{D}_{BC}
        \end{align*}
        \item \ref{bbc_property_1b}: This property is fulfilled due to the use of blurry ball growing with parameter $\rho \in O\left(\frac{\cblur\lambda}{\mathcal{D}}\right)$.
        It holds:
    \begin{align*}
         d(w',A_t) &\leq d(w',\kexp(A_t)) + \mathbf{max}_{w \in \kexp(A_t)} d(w,A_t)\\
         \text{As we overestimate:}\\
         &\leq d(w',\kexp(A_t)) + (1+\epsilon) \cdot X_{i} \cdot \mathcal{D}_{BC}\\
         \text{By Lemma \ref{lemma:bbg}:}\\
         &\leq \left(1-O\left(\frac{\log\log n}{\log n}\right)^{-1}\right) \cdot \frac{\mathcal{D}_{BC}}{250\cblur\lambda}\\  &+   (1+\epsilon) \cdot X_{i} \cdot \mathcal{D}_{BC}\\
         \text{For $O\left(\frac{\log\log n}{\log n}\right) \leq \frac{1}{2}$: }\\
         &\leq \frac{\mathcal{D}_{BC}}{125\cblur\lambda}  +   (1+\epsilon) \cdot X_{i} \cdot \mathcal{D}_{BC}\\
         \text{As ${X_t \leq 1}$:}\\
          &\leq \frac{\mathcal{D}_{BC}}{125\cblur\lambda}  +  (X_t+\epsilon) \cdot \mathcal{D}_{BC}\\
          &\leq \left(X_{i} + \frac{1}{125\cblur\lambda}  + \epsilon\right) \cdot \mathcal{D}_{BC}\\
            \text{For $\epsilon \leq \frac{\mathcal{D}_{BC}}{125\cblur\lambda}$: }\\
             &\leq \left(X_{i} + \frac{1}{60\cblur\lambda}\right) \cdot\mathcal{D}_{BC}\\
    \end{align*}
        \item \ref{bbc_property_1c}: This property is also fulfilled due to the use of blurry ball growing with parameter $\rho \in O\left(\frac{\cblur\lambda}{\mathcal{D}}\right)$ (cf. Lemma \ref{lemma:bbg}).
        \item \ref{bbc_property_2a}: As we pick $\mathcal{D}' \in O(\mathcal{D}\log^2 n)$, a single application of Theorem \ref{thm:clustering_general} cuts an edge with probability $O\left(\nicefrac{\ell_z\log n}{\mathcal{D}'}\right) = O\left(\nicefrac{\ell_z}{\mathcal{D}\log n}\right)$.
        Therefore, we gave $\alpha(G_{t-1}) \in O(\nicefrac{1}{\log n})$ in each recursive step
        \item \ref{bbc_property_2b}: Recall that we perform at most $O(\log n)$ recursive step until all nodes are clustered.
        Further, we have $\alpha(G_{t-1}) \in O(\nicefrac{1}{\log n})$ for each graph, no matter the topology or size.
        Thus, as we apply the decomposition $O(\log n)$ times, the parameters sum up to $O(1)$.
    \end{itemize}
    This proves the lemma.
\end{proof}

Now fix an edge $z \in E$.
During this proof, we will only consider the connected component that contains $z$.
Thus, when we talk about a step $A_t$ or a cluster $K_t$, we mean the set or cluster in $z$'s component.
For convenience, we further define the following terms:
\begin{align}
    \lambda &:= 4\log\log n\\
    \gamma_z &:= \frac{\ell_z}{\mathcal{D}_{BC}}\\
    \gblurdash &= \frac{1}{2\cblur\lambda}
\end{align}
First, we note that only because a cluster around a center $A_t$ could \textbf{potentially} cut $z$, it does not automatically mean that it will happen.
If an edge is not cut, there are two possibilities.
Either both endpoints remain active (and will be clustered in a future iteration), or
both endpoints of the edge could be added to the cluster $K_t$. 
In the latter case, the edge would be \emph{safe} as it can never be cut in future steps. 
Our goal is now to bound the probability that an edge is safe.
We begin this part of our analysis with a useful observation that quantifies under which circumstances a node is added to a cluster and when it is \emph{safe} from being clustered.
Recall that by \eqref{bbc_property_1a}, the probability that both endpoints of $z$ is added to $K_t$ if 
\begin{align}
    d_{{G_t}}(z,A_t) < X_t \cdot \mathcal{D}_{BC} - \gamma_z\cdot\mathcal{D}_{BC}.
\end{align}
Further, by \eqref{bbc_property_1b}, for any node $v \in V$ added added to $K_t$, it must hold:
\begin{align}
     d_{G_t}(v,A_t) \leq  X_t \cdot \mathcal{D}_{BC} + \gblurdash\cdot\mathcal{D}_{BC}
\end{align}
Going on, we say that an edge is \emph{threatened} by some set $A_t$ if 
\begin{align*}
    \mathcal{C}_t := \underbrace{\left\{d_{{G_t}}(z,A_t) \geq X_t \cdot \mathcal{D}_{BC} - \gamma_z\cdot\mathcal{D}_{BC} \right\}}_{z \textbf{ is not fully added to the cluster.}} \cap \underbrace{\left\{ d_{G_t}(z,A_t) \leq  X_t \cdot \mathcal{D}_{BC} + \gblurdash\cdot\mathcal{D}_{BC} \right\}}_{z \textbf{ may be added.}}
\end{align*}
In this case, the edge is \textbf{potentially} affected by the blurry growing procedure and may be cut. 
Note that it holds for every edge $z \in E$ by \ref{bbc_property_1c} and \ref{bbc_property_2a} that:
\begin{align}
   \pr{\Cut_z^t \mid \mathcal{C}_t} \leq \underbrace{O\left(\frac{\cblur\cdot\lambda\cdot\ell_z}{\mathcal{D}_{BC}}\right)}_{z \textbf{ cut by } K_t.} + \underbrace{\alpha(G_{t-1}) \cdot \frac{\ell_z}{\mathcal{D}_{BC}}}_{z \textbf{ cut by decomposition.}} \label{eqn:to_c}
\end{align}
Otherwise, it can only be cut by the decomposition, so by \ref{bbc_property_2a} we only have
\begin{align}
   \pr{\Cut_z^t \mid \overline{\mathcal{C}}_t} \leq \underbrace{\alpha(G_{t-1}) \cdot \frac{\ell_z}{\mathcal{D}_{BC}}}_{z \textbf{ cut by decomposition.}} \label{eqn:not_to_c}
\end{align}
We define the set of \emph{threateners} as:
  \begin{align}
        \mathcal{I}_z := \{A_t \mid d_{G_t}(z,A_t) \leq (1+\gblurdash)\cdot\mathcal{D}_{BC} \} 
\end{align}

We will show that we can bound the probability of an edge being cut based on the expected number of threateners.
Note that this is slightly more general than we need it to be.
It holds:
\begin{lemma}
\label{lemma:blurry_cut}
Consider a $(\lambda,\cblur,\cldd)$-bounded blurry ball growing process and assume that for some universal constants $c_1,c_2$ that are independent of $\lambda$, it holds:
\begin{align}
    \E{|\mathcal{I}_z|} \leq c_1 \cdot e^{\left(\lambda + c_2\right)\cdot \left(1 + \gblurdash\right) }
\end{align}
Then, for $\lambda_z \in O\left(\nicefrac{\mathcal{D}_{BC}}{\lambda}\right)$, it holds:
\begin{align}
   \pr{\Cut} \leq O\left(\frac{(\cblur+\cldd)\lambda \cdot \ell_z}{\mathcal{D}_{BC}}\right)
\end{align}
\end{lemma}
Note that in the algorithm at most $O(\log n)$ sets can threaten an edge $z$ on expectation. 
This follows, as we have at most $O(\log n)$ iterations, w.h.p., and there is at most one threatener per iteration.
Further, recall that we choose $\lambda = 4\log\log n$.
Therefore, we have:
\begin{align*}
     c_1 \cdot e^{\left(\lambda + c_2\right)\cdot \left(1 + \gblurdash\right) } \geq e^\lambda &= e^{4\log\log n} = \log^4(n) \geq O(\log n)  
\end{align*}
Thus, Lemma \ref{lemma:blurry_cut} proves Claim \ref{lemma:k_path_cutting}.

We will show that if an edge is endangered, the probability that both endpoints of the edge itself are added to the cluster is constant.
Thus, with a constant probability, the edge is safe from being cut in any future iteration.
Intuitively, this follows due to the exponential distribution.
If the random diameter is big enough to cluster a node close to $z$, it will likely also cluster $z$.
The proof is of course more difficult as we need to account for the fact that we truncate the distribution and the imprecision introduced by the approximate shortest paths.
We show that probability for this event depends on the \textbf{expected} number of centers that get close to $z$ and - rather surprisingly - not on the exact number. 
It holds:
\begin{lemma}
\label{lem:prob_fixed_step_clustered_comp}
  The probability that an edge is endangered in any of the $\tau$ steps we consider is:
  \begin{align}
      \sum_{t=0}^{\tau} \pr{\mathcal{C}_t} \leq  \left(1-e^{-\lambda(\gamma_z+\gblurdash)}\right)\cdot e^{\lambda\gblurdash}\cdot\left(1+\frac{\E{|\mathcal{I}_z|}}{e^{\lambda} - 1}\right)
  \end{align}
\end{lemma}
\begin{proof}
    Consider a step $t$.
    Let $A_t$ be the center around which a ball is carved and let $X_t$ be the exponentially distributed random distance picked by $A_t$.
    Further denote
    \begin{align}
        \rho_t := \frac{d_{G_t}(A_t,z)}{\mathcal{D}_{BC}}
    \end{align}
     as the normalized distance between $A_t$ and $z$ in $G_t$.
    Note that $G_t$ and therefore also this distance depends on the random decisions of all previous rounds, so $\rho_t$ is a random variable.
    In the case that $z$ is already clustered by step $t$, we define the distance to be infinite.
       % \begin{align}
    %     \pr{\mathcal{C}_t \mid \rho_i} \leq \pr{\rho'  \leq \delta  \leq \rho'+2\gamma+2\epsilon} 
    % \end{align}
    % And therefore, it holds:
    % \begin{align}
    %     \pr{\mathcal{C}_t \mid \rho_i} \leq \left(1-e^{-2\lambda(\gamma+\epsilon)}\right)\cdot\left(1+4\epsilon\lambda\right)\cdot\left(\pr{\delta \geq \rho+\epsilon}+\frac{1}{e^{\lambda}-1}\right)~
    % \end{align}
    % Now consider the event $\mathcal{C}_t$ that $A_t$ does not add an endpoint of $z$ to the cluster right away.
    % Choose $\rho' := \rho-\gamma$ and we are in the same setup as before.
    We now show that:
\begin{claim}
\label{lemma:c_bound_2}
It holds:
\begin{align*}
    &\pr{\rho_t- \gblurdash \leq X_t \leq \rho_t+\gamma_z \mid \rho_t} \\
    \leq& \left(1-e^{-\lambda(\gblurdash+\gamma_z)}\right)\cdot e^{\lambda\gblurdash}\cdot\left(\pr{X_t \geq \rho_t\mid \rho_t}+\frac{1}{e^{\lambda}-1}\right)~
\end{align*}
\end{claim}
\begin{proof}
The proof follows directly from the definition of $X_t$.
First, we note that it holds that
	\begin{align}
	     \pr{X_t>\rho_t-\gblurdash}
&=\int_{\rho_t-\gblurdash}^{1}\frac{\lambda\cdot e^{-\lambda y}}{1-e^{-\lambda}}dy
	=\frac{e^{-(\rho_t-\gblurdash)\cdot\lambda}-e^{-\lambda}}{1-e^{-\lambda}}\\
 &= \frac{e^{-(\rho_t-\gblurdash)\cdot\lambda}}{1-e^{-\lambda}} - \frac{e^{-\lambda}}{1-e^{-\lambda}}\\
 &= \frac{e^{-(\rho_t-\gblurdash)\cdot\lambda}}{1-e^{-\lambda}} - \frac{1}{e^{\lambda}-1}
	\end{align}
On the other hand, it holds:
	\begin{align*}
	&\pr{\rho_t-\gblurdash \le X_t\le\rho_t+\gamma_z \mid \rho_t}\\
 &=\int_{\rho_t-\gblurdash}^{\min\left\{ 1,\rho_t+\gamma_z\right\}}\frac{\lambda\cdot e^{-\lambda y}}{1-e^{-\lambda}}dy\\
	& \leq\frac{e^{-(\rho_t-\gblurdash)\cdot\lambda}-e^{-\left(\rho_t+\gamma_z\right)\cdot\lambda}}{1-e^{-\lambda}}\\
\textit{Substitute } \rho' = \rho_t-\gblurdash\\
 & \leq\frac{e^{-\rho'\cdot\lambda}-e^{-\left(\rho'+\gblurdash+\gamma_z\right)\cdot\lambda}}{1-e^{-\lambda}}\\
 & \leq \left(1 - e^{-\lambda(\gblurdash+\gamma_z)}\right) \frac{e^{-\lambda\rho'}}{1-e^{-\lambda}}\\
\textit{Substitute back}\\
  & \leq \left(1 - e^{-\lambda(\gblurdash+\gamma_z)}\right) \frac{e^{-\lambda(\rho_t-\gblurdash)}}{1-e^{-\lambda}}\\
    & \leq \left(1- e^{-\lambda(\gblurdash+\gamma_z)}\right) \frac{e^{-\lambda(\rho_t-\gblurdash)}}{1-e^{-\lambda}} -  \frac{1}{e^{\lambda}-1} + \frac{1}{e^{\lambda}-1}\\
    & \leq \left(1- e^{-\lambda(\gblurdash+\gamma_z)}\right) \cdot e^{\lambda\gblurdash}\cdot\frac{e^{-\lambda\rho_t}}{1-e^{-\lambda}} -  \frac{1}{e^{\lambda}-1} + \frac{1}{e^{\lambda}-1}\\
    & \leq \left(1 - e^{-\lambda(\gblurdash+\gamma_z)}\right) \cdot e^{\lambda\gblurdash}\cdot \left(\pr{X_t>\rho_t \mid \rho_t }  + \frac{1}{e^{\lambda}-1}\right)
	\end{align*}
Thus, the claim follows.
\end{proof}
     Next, Let $\mathcal{F}_t$ be the event that $K_t$ contains (at least) one endpoint of edge $z$.
    This is clearly the case if the random variable $X_t$ is at least $\rho_t$ and we have
    \begin{align*}
        \pr{\mathcal{F}_t} \geq \pr{X_t \geq \rho_t\mid \rho_t} 
    \end{align*}
    Further, we define the following helper variables:
    \begin{align*}
        \alpha &=e^{-\lambda(\gblurdash+\gamma_z)}\\
        \beta &=e^{\lambda\gblurdash}\\
        \zeta &= (1-\alpha)\beta
    \end{align*}
    Note that by Lemma \ref{lemma:c_bound_2}, we have the following relationship between $\mathcal{F}_t$ and $\mathcal{C}_t$:
    \begin{align*}
         \pr{\mathcal{C}_t \mid \rho_t} &\leq \pr{\rho_t- \gblurdash \leq X_t \leq \rho_t+\gamma_z \mid \rho_t} & \textit{By Definition if $\mathcal{C}_t.$}\\
         &\leq \left(1-e^{-\lambda(\gblurdash+\gamma_z)}\right)\cdot e^{\lambda\gblurdash}\cdot\left(\pr{X_t \geq \rho_t\mid \rho_t}+\frac{1}{e^{\lambda}-1}\right) & \textit{By Claim \ref{lemma:c_bound_2}.}\\
         &:= \zeta\cdot\left(\pr{X_t \geq \rho_t\mid \rho_t}+\frac{1}{e^{\lambda}-1}\right) & \textit{By Definition of $\zeta.$}
    \end{align*}
    Thus, by definition of $\mathcal{F}_t$, it holds:
    \begin{align}
          \pr{\mathcal{C}_t \mid \rho_t} &\leq  \zeta\cdot\left( \pr{\mathcal{F}_t\mid \rho_t}  +\frac{1}{e^{\lambda}-1}\right)   \label{eq:FleqC}
    \end{align}
    In the following, let $f(\cdot)$ be the probability density function of $\rho_t$.
    By the law of total probability, we have:
    \begin{align*}
        \pr{\mathcal{C}_t \mid A_t \in \mathcal{I}_z} &:= \pr{\mathcal{C}_t \mid \rho_t \leq 1+\gblurdash}&\textit{By Definition of $A_t$.}\\
        &:= \int_0^{1+\gblurdash} \pr{\mathcal{C}_t  \mid \rho_t = x}f(x)dx&\textit{By Law of Tot. Prob.}\\
        &\leq \zeta \cdot \int_0^{1+\gblurdash}\left(\pr{\mathcal{F}_t\mid  \rho_t = x} + \frac{1}{e^{\lambda}-1}\right)f(x)dx& \textit{By Inequality \eqref{eq:FleqC}.}\\
        &= \zeta\cdot\left(\pr{\mathcal{F}_t \mid \rho_t \leq 1+\gblurdash } + \frac{1}{e^{\lambda}-1}\right)&\textit{By Law of Tot. Prob.}\\
        &= \zeta\cdot\left(\pr{\mathcal{F}_t \mid A_t \in \mathcal{I}_z } + \frac{1}{e^{\lambda}-1}\right)&\textit{By Definition of $A_t$.}
    \end{align*}
    Recall that we observe our ball carving algorithm for $\tau$ steps.
    Note that it holds:
    \begin{align*}
        \E{|\mathcal{I}_z|} = \E{\sum_{t=1}^\tau \mathds{1}_{\{A_t \in \mathcal{I}_z\}}} =  \sum_{t=1}^\tau \E{\mathds{1}_{\{A_t \in \mathcal{I}_z\}}} = \sum_{t=1}^\tau \pr{A_t \in \mathcal{I}_z}
    \end{align*}
    Further, as each endpoint of $z$ will added to some cluster at most once, it holds:
     \begin{align*}
        \sum_{t=1}^\tau \pr{A_t\in \mathcal{I}_z}\cdot  \pr{\mathcal{F}_{i} \mid A_t\in \mathcal{I}_z} = \sum_{t=1}^\tau \pr{\mathcal{F}_{i}} \leq 1
    \end{align*}
    Now, we compute the probability that the edge is cut in any of the $\tau$ steps.
    It holds:
\begin{align*}
\sum_{t=1}^{\tau}\pr{\mathcal{C}_{t}} &= \sum_{t=1}^{\tau}\pr{A_t \in \mathcal{I}_z}\cdot\pr{\mathcal{C}_{t} \mid A_t \in \mathcal{I}_z}\\
& \le\zeta \cdot \sum_{t=1}^{\tau} \pr{A_t\in \mathcal{I}_z} \cdot \left( \pr{\mathcal{F}_{t} \mid A_t \in \mathcal{I}_z}+\frac{1}{e^{\lambda}-1}\right)\\
& \le\zeta \cdot \left(\sum_{t=1}^{\tau} \pr{A_t\in \mathcal{I}_z} \cdot  \pr{\mathcal{F}_{t} \mid A_t \in \mathcal{I}_z}+\sum_{t=1}^{\tau} \frac{\pr{A_t\in \mathcal{I}_z}}{e^{\lambda}-1}\right)\\
& \le\zeta \cdot \left(1+\frac{\E{|\mathcal{I}_z|}}{e^{\lambda}-1}\right)
\end{align*}
Now we can wrap up the proof by considering the value of $\zeta$.
\end{proof}

Now, we can plug our technical results together and see that each edge (that is sufficiently short) has only a constant probability to be endangered until step $\tau$.
Or, to put it differently, with constant probability, an edge is added to $K_t$ if the center is close enough.
It holds:
\begin{lemma}
\label{lemma:sum_c_bound}
Let $\lambda$ be the parameter of the truncated exponential distribution and assume that for some universal constants $c_1,c_2$ that are independent of $\lambda$, it holds:
\begin{align}
    \E{|\mathcal{I}_z|} \leq c_1 \cdot e^{\left(\lambda + c_2\right)\cdot \left(1 + \gblurdash\right) }
\end{align}
Then, for $\gamma_z \leq \gblurdash$, it holds:
\begin{align}
   \sum_{t=0}^{\tau}\pr{ \mathcal{C}_t}  \in O\left(1\right)
\end{align}
\end{lemma}
\begin{proof}
By Lemma \ref{lem:prob_fixed_step_clustered_comp}, it holds
\begin{align}
    \sum_{t=0}^\tau\pr{\mathcal{C}_t} \leq  \left(1-e^{-\lambda(\gblurdash+\gamma_z)}\right)e^{\lambda\gblurdash}\left(1+\frac{\E{|\mathcal{I}_z|}}{e^{\lambda} - 1}\right)
\end{align}
Focusing on the last term, we see that it holds:
\begin{align}
    \left(1+\frac{\E{|\mathcal{I}_z|}}{e^{\lambda}-1}\right) &\leq 1+\frac{c_1 \cdot e^{(\lambda+c_2)(1+\gblurdash)}}{e^{\lambda}-1} =  1+\frac{c_1 \cdot e^{\lambda(1+\gblurdash) + c_2 + c_2\gblurdash}}{e^{\lambda}-1}
\end{align}
Recall that it holds $\gblurdash \leq 1$ and we can further simplify the formula as follows:
\begin{align}   
    \left(1+\frac{\E{|\mathcal{I}_z|}}{e^{\lambda}-1}\right)&\leq 1 +\frac{e^{\lambda} c_1 e^{3c_2}}{e^{\lambda}-1}e^{\lambda\gblurdash}
\end{align}
One can easily verify that that it holds $\frac{e^\lambda}{e^\lambda-1} \leq 2$ for all $\lambda \geq 1$.
To wrap things up, we define $c_3 := 2\cdot c_1\cdot e^{3c_2}$ and see that it holds:
\begin{align}
\label{eqn:helper_e1}
    \left(1+\frac{\E{|\mathcal{I}_z|}}{e^{\lambda}-1}\right)&\leq 1 +2 c_1 e^{3c_2}e^{\lambda(\gblurdash)} \leq 1+ c_3 e^{\lambda(\gblurdash)}
\end{align}
Note that $c_3$ only depends on $c_1$ and $c_2$ and is independent of $\gblurdash \leq 1$ and $\lambda \geq 1$ as long as they are in their given bounds.
Moving on, we need some more algebra to properly bound our expression.
To this end, we prove the following claim:
\begin{claim}
\label{claim:helper_e2}
For each constant $c\geq 0$ and $x \in [0,\nicefrac{1}{2}]$, it holds:
    \begin{align}
        e^x(1-e^{-2x})(1+c\cdot e^x) \leq c\cdot 10 \cdot x
    \end{align}
\end{claim}
\begin{proof}
We first reorder the l.h.s of the formula and see:
\begin{align}
    e^x(1-e^{-2x})(1+c\cdot e^x) &= e^x(1 + ce^{x} - e^{-2x} - ce^{-x})\\
    &= e^x(\underbrace{1 - e^{-2x}}_{(*)} + c\underbrace{(e^{x}  - e^{-x})}_{(**)})
\end{align}     
One can easily verify that for $x \leq 1$, it holds $e^{-x} \geq 1-x $ as this is a well-known inequality.
Thus, for the first part of the formula, it holds
\begin{align}
   (*) &= 1 - e^{-2x} \leq 1 - (1-2x) = 2x
\end{align}
Further, it holds $e^{x} \leq 1+x+x^2$ for $x \leq 1$ and thus, we have:
\begin{align}
    (**) &= e^x - e^{-x} \leq (1+x+x^2) - (1-x) = 2x + x^2 \leq 3x
\end{align}
Plugging both inequalities back into our formula gives us the following:
\begin{align}
     e^x(1-e^{-2x})(1+c\cdot e^x) &\leq (1+x+x^2)(2x + 3cx) = (1+x+x^2)(x(3c+2))\\
     &\leq (1+x+x^2)5cx = 5cx + 5cx^2 + 5cx^3\\
     &\leq 10cx
\end{align}
We used that $x \leq \frac{1}{2}$ in the last inequality.
This proves the claim.
\end{proof}
Now, we can finalize the proof.
By choosing both $\gamma,\epsilon \in O\left(\frac{1}{\lambda}\right)$, for a small enough hidden constant, we can ensure that the exponent $2\lambda(\gblurdash)$ is always smaller than $\frac{1}{2}$.
Using our claim, we get that 
\begin{align*}
    \sum_{t=0}^\tau\pr{\mathcal{C}_t} &\leq  \left(1-e^{-\lambda(\gblurdash+\gamma_z)}\right)e^{\lambda\gblurdash}\left(1+\frac{\E{|\mathcal{I}_z|}}{e^{\lambda} - 1}\right)\\
\text{As $\gamma_z \leq \gblurdash$:}\\
&\leq  \left(1-e^{-\lambda(2\gamma_z)}\right)e^{\lambda\gblurdash}\left(1+\frac{\E{|\mathcal{I}_z|}}{e^{\lambda} - 1}\right)\\
\text{By Inequality \ref{eqn:helper_e1}:}\\
     &\leq  e^{\lambda\gblurdash}\left(1-e^{-2\lambda\gblurdash} \right)(1 + c_3 e^{\lambda\gblurdash})\\
     \text{By Claim \ref{claim:helper_e2}:}\\
&\leq 10c_3\lambda\gblurdash \\
\text{As $\gblurdash \leq \frac{1}{2\cdot\cblur\cdot\lambda}$:}\\
&\leq \frac{5c_3}{\cblur}
\end{align*}
Thus, the lemma follows as $\cblur$ is a constant larger than one.
\end{proof}

Finally, we can prove Lemma \ref{lemma:blurry_cut}.
\begin{proof}[Proof of Lemma \ref{lemma:blurry_cut}]
Let $\Cut^{(\tau)}_z$ be the event that edge $z$ by some cluster $K_t$ is cut \textbf{until} step $\tau$.
Further, let $\Cut^{t}_z$ be the event that edge $z$ is cut \textbf{in} step $t$.
Using all of our arguments and insights, we conclude:
\begin{align*}
    \pr{\Cut^{(\tau)}_z} &= \pr{\bigcup_{t=0}^\tau \Cut^{t}_z}
    \leq \sum_{t=0}^\tau \pr{\Cut^{t}_z} &\textit{By Union Bound.}\\
    &= \sum_{t=0}^\tau \pr{\mathcal{C}_t} \cdot \pr{\Cut^{t}_z \mid \mathcal{C}_t} + \pr{\overline{\mathcal{C}_t}} \cdot \pr{\Cut^{t}_z \mid \overline{\mathcal{C}_t}}  &\textit{By Law of Tot. Prob.}\\
    &\leq \sum_{t=0}^\tau \pr{\mathcal{C}_t} \cdot O\left(\frac{\cblur\cdot\lambda\cdot\ell_z}{\mathcal{D}}\right) + \sum_{t=0}^\tau \alpha(G_{t-1}) \cdot \frac{\ell_z}{\mathcal{D}_{BC}} &\textit{By Ineq. \eqref{eqn:to_c} and \eqref{eqn:not_to_c}.} \\
    &\leq   O\left(1\right) \cdot  O\left(\frac{\cblur\cdot\lambda\cdot\ell_z}{\mathcal{D}}\right)  + \sum_{t=0}^\tau \alpha(G_{t-1}) \cdot \frac{\ell_z}{\mathcal{D}_{BC}} &\textit{By Lemma \ref{lemma:sum_c_bound}.}\\
      &\leq O\left(\frac{\cblur\cdot\lambda\cdot\ell_z}{\mathcal{D}}\right)  + \frac{\cldd\cdot\lambda\cdot\ell_z}{\mathcal{D}_{BC}} &\textit{By \eqref{bbc_property_2b}.}\\
    &= O\left(\frac{(\cblur+\cldd)\cdot\lambda\cdot\ell_z}{\mathcal{D}}\right)
\end{align*}
This was to be shown.
\end{proof}

\paragraph{Complexity} 
Finally, we analyse the time complexity, which is straightforward.
\begin{lemma}[Complexity]
\label{claim:k_path_complexity}
        Each iteration of the algorithm can be implemented with $\Tilde{O}(k)$ approximate \SetSSP computations with parameter $\epsilon \in \Omega(\nicefrac{1}{\log^2 n})$ and minor aggregations, w.h.p.
\end{lemma}
\begin{proof}
For the complexity, consider a fixed iteration. 
In the following, we summarize minor aggregations and approximate SSSP with parameter $\epsilon \in O(\nicefrac{1}{\log^2 n})$ simply as \emph{operations}.
In Step $1$, we compute an \LDD using the algorithm from Theorem. It requires $\Tilde{O}(1)$ operations.
 In Step $2$, we compute a weak separator using the algorithm from Theorem. It requires $\Tilde{O}(k)$ operations.
 In Step $3$, we execute a \SetSSP from the separator. This clearly requires one operation.
 In Step $4$, we execute blurry ball growing from Lemma \ref{lemma:bbg}. It requires $\Tilde{O}(1)$ operations.
 Finally, Step $5$ is purely local.
This proves the lemma as there are at most $O(\log n)$ iterations, w.h.p.
\end{proof}

\section{Full Analysis of \autoref{lemma:backbone_clustering}}
\label{sec:appendix_refinement}

In this section, we prove Lemma and show that a backbone clustering can easily be extended to clustering with very favorable properties.
To be precise, we show the following lemma:

\refinement*

We begin by proving a helpful auxiliary lemma.
We show that we can efficiently compute so-called $\delta$-nets on paths with few minor aggregations.
In a $\delta$-net, we mark a set of nodes on the path such that each node on the path has a marked node in distance $\delta$ and two marked nodes have distance at least $\delta$..
These nets can be seen as a generalization of the Maximal Independent Set (MIS) or Ruling Sets for arbitrary distances and are formally defined as follows:
\begin{definition}[$\delta$-nets]
    Let $V$ be a set of nodes/points and let $d(\cdot,\cdot): V^2 \to \mathbb{R}$ be a distance metric on these nodes/points.
    Then, we define a $\delta$-net of $V$ as a set $\mathcal{N} \subseteq V$, s.t., it holds:
    \begin{itemize}
        \item For each node/point $v \in V$, there is a node/point $p \in \mathcal{N}$ with $d(v,p) \leq \delta$.
        \item For two nodes/points $p_1,p_2 \in \mathcal{N}$, it holds $d(p_1,p_2) > \delta$.
    \end{itemize}
\end{definition}
We are particularly interested in $\delta$-nets of paths.
In a sequential model, these nets can trivially be constructed by greedy algorithm that iterates over the nodes path.
In the following, we sketch an algorithm that efficiently computes these nets in our model.
The idea behind the algorithm is to build distance classes of length $\Theta(\delta)$ and mark one node per class. 
Note that the distance between net points is w.r.t. to the path, and they could be closer to each other when considering all paths in $G$ (which will be important later). 
\begin{lemma}
\label{lemma:path_portals}
Let $G = (V,E,w)$ be a weighted graph.
    Consider a set of $\ell$ (not necessarily shortest) paths of length $\Tilde{O}(\mathcal{D})$ and let $\epsilon > 0$ be a parameter. 
    Then, we can compute $\delta$-net with $\delta = \epsilon\mathcal{D}$ for all paths in $\Tilde{O}(\epsilon^{-1}\cdot\ell)$ minor aggregations.
\end{lemma}
\begin{proof}
We compute the net points on all paths sequentially.
In the following, we focus on a single path $P$ and show that the net points can be computed with $\Tilde{O}(1)$  minor aggregations.
Together with the fact that there are $\ell$ paths, this proves the lemma.
    We assume that w.l.o.g. each node knows whether it is the first or last node on a path.
    First, all nodes compute their distance to the first node on the path (if they do not already know it). 
    Since the subgraphs consist of single paths, this can be done using the $\textsc{AncestorSum}$ primitive.
    We simply need to sum up all distances on the unique path to the first/last node.
    Note that the distances are \emph{exact}. 
    Next, the last node on the path broadcasts its \emph{distance label}, i.e., the length of the path, to all nodes on the path. 
    Denote this distance as $\mathcal{D}_P$ in the following. 
    Each node $v \in P$ now locally computes the smallest integer $i \in [1,\frac{\mathcal{D}_P}{\delta}]$, s.t., it holds $d_P(v,v_1) \leq i\cdot\delta$.
    We say that $v$ is in \emph{distance class} $i$.
    Each node now exchanges its distance class with its neighbors. 
    This can be done in a single local aggregation as we can encode the distance class in $O(\log(nW))$ bits.
    Finally, all nodes with an \emph{even} distance class and a neighbor in a \emph{lower} distance class locally declare themselves net points.
    As a result of this procedure, the distance between a node and the next net point is at most $2\delta$ and the distance between two net points is at least $2\delta$.
    This proves the lemma for $\delta = \nicefrac{\epsilon}{2}$.
\end{proof}
Given these preliminaries, the high-level idea is to create a net on each backbone path and then use the resulting net nodes as a center for the algorithm from Theorem \ref{thm:genericldd}.
The algorithm works in three steps.
\begin{mdframed}
\vspace{-10pt}
\paragraph{\textbf{(Step 1) Create a Backbone Clustering:}} Execute the black-box algorithm $\mathcal{A}$ to obtain a $(\alpha,\beta,\kappa)$-backbone clustering $\mathcal{K} = K_1, \ldots K_N$ with pseudo-diamter $\mathcal{D}_{BC}$. 

\medskip

The following two steps are executed parallelly in each cluster $K_i \in \mathcal{K}$.

\paragraph{\textbf{(Step 2) Create a Net on Each Backbone Path:}} In each of the $O(\kappa)$ paths in the backbone $\mathcal{B}_i$ of the cluster $K_i$, create a ${\mathcal{D}_{BC}}$-net.
% This means that each node on the path has a center in distance ${\mathcal{D}_{BC}}$, but two centers are in distance at least ${\mathcal{D}_{BC}}$.
% For this, we use the algorithm promised by Lemma \ref{lemma:path_portals}.
Denote the net points created by this algorithm as $\mathcal{N}_i$.

\paragraph{\textbf{(Step 3) Build Clusters Around Net Points:}}  Apply \textbf{one iteration of} our algorithm from Theorem \ref{thm:genericldd}.
That means we do not reapply the main loop until all nodes are clustered; we apply it only once.
We choose the previously computed net points $\mathcal{N}_i$ as centers $\mathcal{X}_i$ in each $K_i$.
As distance parameter choose $\mathcal{D} = 2\mathcal{D}_{BC}$. 
This

\end{mdframed}
\bigskip

The lemma follows directly from the correctness of the subroutines that we use.
Consider a single cluster $K_i$ created by the black box algorithm $\mathcal{A}$.
First, we want to show that each node $v \in K_i$ has at least one net point $x \in \mathcal{N}_i$ in distance $\mathcal{D}$.
This follows because $v$ must be in the distance $\mathcal{D}_{BC}$ to some path, and there is a net point in the distance $\mathcal{D}_{BC}$ to the closest node on that path.
This holds in either of the two setups, as it follows directly from the properties of the backbone clustering.
Thus, we fulfill the \textbf{covering property} required by the Theorem \ref{thm:genericldd}.
In the \textbf{packing property}, we need to distinguish whether the path in the back are approximate shortest paths or not.
First, consider the case that we do \emph{not} have exact shortest paths in $G$. 
Recall that we need to bound the number of net points in distance $6\mathcal{D}$ to each node as this determines the cutting probability.
As the length of each path is bounded by $O(\mathcal{D}\beta)$, we obtain $O(\beta)$ net points per path and $O(\kappa\beta)$ net points in total as we have $\kappa$ paths.
Clearly, each node can have at most $O(k\beta)$ net points at \emph{any} distance.
Thus, if we run the algorithm from Theorem with parameter $\tau \in O(k\beta)$ and receive clusters with strong diameter $8\mathcal{D} = 8\cdot(2\mathcal{D}_{BC})= 16\mathcal{D}_{BC}$.
This proves the proclaimed diameter bound from Lemma.
Further, each node is contained in a cluster with probability at least $\nicefrac{1}{2}$.
This follows directly from Lemma in the analysis of Theorem \ref{thm:genericldd}.
Thereby, we showed the clustering probability of $\nicefrac{1}{2}$.
It remains to prove the cutting probability.
An edge is cut with probability $O(\frac{\alpha\ell}{\mathcal D})$ by the $(\alpha,\beta,\kappa)$-backbone clustering in Step $1$ and with probability $O\left(\frac{(\log(\kappa\beta))\ell}{\mathcal D}\right)$ in Step $3$. 
The former follows from the definition of the backbone clustering, and the latter follows from Lemma \ref{lemma:clustering_cut_fixed_iteration} in the analysis of Theorem \ref{thm:genericldd}.
By the union bound, the probability of $O\left(\frac{(\alpha+\log(\kappa\beta))\ell}{\mathcal D}\right)$ follows.

Further, it can be implemented with $\Tilde{O}(1)$ approximate \SetSSP computations and minor aggregations.

For the other bound, suppose each path is an $(1+\beta^{-1})$-approximate shortest path in $K_i$.
Note that we only need to consider the cutting probability because the proof of the other two properties is analogous.
Consider a node $v$ and consider all net points in the distance $6\mathcal D$ to $v$.
Recall that each path has length $\beta\mathcal{D}$, and we construct a $\mathcal{D}$-net.
If we compute net on approximate shortest paths \textbf{of bounded length}, they also have the following extremely useful property.
It holds:
\begin{lemma}
\label{lemma:net_distance}
    Let $\epsilon_p, \epsilon_s \leq 1$ be two arbitrary parameters.
    Let $P$ be a $(1+\epsilon_s)$-approximate shortest path of some graph $G := (V,E,w)$ of length $\mathcal{D} \cdot \epsilon_s^{-1}$.
    Further, let $\mathcal{N}$ be a $\epsilon_p\mathcal{D}$-net on $P$.
    Then, for each $c \geq 1$, every node $v \in V$ has at most $O(c\epsilon_p^{-1})$ net nodes in distance $c\mathcal{D}$ in $G$. 
\end{lemma}
\begin{proof}
    Let $s \in V$ be the first node on path $P$, i.e., the node from which the approximate shortest path $P$ was computed.
    We now divide path $P$ into distance classes with respect to the distances $G$.
    For each net node $p \in \mathcal{N}$, we say that $p$ is in distance class $i_p$ if $i_p \in \left[1,O(\epsilon_s^{-1})\right]$ is the biggest integer such that $d_G(s,p) \leq i \cdot \mathcal{D}$, i.e., it holds:
    \begin{align}
        i_p := \left\lceil \frac{d_G(s,p)}{\mathcal{D}} \right\rceil
    \end{align}
    Note that $d_G(\cdot,\cdot)$ denotes the true distance between nodes and \textbf{not} the distance on path $P$.
    The latter may be greater by $(1+\epsilon_s)$-factor.
    
    Define $p_1 \in \mathcal{N}$ to be the first net point (counting from $s$) on $P$ that is in distance $c\mathcal{D}$ to $v$, i.e., it holds:
    \begin{align}
        p_1 = \argmin_{p \in \mathcal{N} \cap B_G(v, c\mathcal{D})} d_P(p,s)   
    \end{align}
    
    Note that, by this definition, all net points in $B_G(v, c\mathcal{D}$ must lie \emph{behind} $p_1$ on the path $P$, i.e., have a greater distance to $s$ than $p_1$ (with respect to path $P$).
    In the following, we will only consider these net points and do not consider the earlier net points.

    Let $\mathcal{U} \subset \mathcal{N}$ contain the next $N = 10c\epsilon_p^{-1}$ net nodes of $P$. 
    Denote these nodes as $p_1, p_2, \ldots, p_N$.
    Further, as we only consider net points that are further away from $s$ than $p_1$, the distance between $p_1$ and any $p' \not\in \mathcal{U}$ is at least:
    \begin{align}
        d_P(p_1,p') &= \sum_{i=1}^{N} d_P(p_i,p_{i+1}) + d_P(p_{N},p')\\
        &\geq \sum_{i=1}^{N} d_P(p_i,p_{i+1})
        \geq \sum_{i=1}^{N} \epsilon_p \cdot \mathcal{D}
        = 10\cdot c\cdot\mathcal{D} \label{eq:distp1}
    \end{align}

    Further, let $p_1$ be in distance class $i_{p_1}$. 
    We now argue that each net node $p' \in \mathcal{N}$ that is not in $\mathcal{U}$ must be in distance class at least $i_{p_1}+10c-2$. 
    First, we consider the definition of distance class $i_{p'}$, add a dummy distance from $s$ to $v$, and rearrange.
    We get:
    \begin{align*}
        i_{p'} &= \left\lceil \frac{d_G(s,p')}{\mathcal{D}} \right\rceil = \left\lceil \frac{d_G(s,p')}{\mathcal{D}} \right\rceil + \left\lceil \frac{d_P(s,p')}{\mathcal{D}} \right\rceil - \left\lceil \frac{d_P(s,p')}{\mathcal{D}} \right\rceil\\
        &= \underbrace{\left\lceil \frac{d_P(s,p')}{\mathcal{D}} \right\rceil}_{(*)} - \underbrace{\left(\left\lceil \frac{d_P(s,p')}{\mathcal{D}} \right\rceil - \left\lceil \frac{d_G(s,p')}{\mathcal{D}} \right\rceil\right)}_{(**)}
    \end{align*}
    Now recall that the distance between $p_1$ and $p'$ \textbf{on the path} $P$ is at least $10\mathcal{D}$ by definition as $p' \not\in \mathcal{U}$. 
    Thus, for the first term, it holds:
    \begin{align}
        (*) &= \left\lceil \frac{d_P(s,p')}{\mathcal{D}} \right\rceil\\
        \text{As }d_P(s,p') = d_P(s,p_1) + d_P(p_1,p'):\\
        &= \left\lceil \frac{d_P(s,p_1) + d_P(p_1,p')}{\mathcal{D}}   \right\rceil\\
       \text{By Ineq. } \eqref{eq:distp1}: \\
        &\geq \left\lceil \frac{d_P(s,p_1) + 10c\mathcal{D}}{\mathcal{D}}\right\rceil\\
        &\geq \left\lceil \frac{d_P(s,p_1)}{\mathcal{D}}\right\rceil + \left\lceil \frac{10c\mathcal{D}}{\mathcal{D}}\right\rceil-1\\
        &= i_u + 10c - 1
    \end{align}
    For the second term, we use the fact that we chose $\epsilon_s$ to be very small. In particular, we chose it small enough that the approximation error for all nodes of the path (even the nodes close to the end) is smaller than $\mathcal{D}$.
    It holds:
    \begin{align}
        (**) &:= \left\lceil \frac{d_P(s,p')}{\mathcal{D}} \right\rceil - \left\lceil \frac{d_G(s,p')}{\mathcal{D}} \right\rceil\\
        &\leq \left\lceil \frac{d_G(s,p') + \epsilon_s\cdot d_G(s,p')}{\mathcal{D}} \right\rceil - \left\lceil \frac{d_G(s,p')}{\mathcal{D}} \right\rceil\\
        &\leq \left\lceil \frac{\epsilon_s\cdot d_G(s,p')}{\mathcal{D}} \right\rceil\\
        \text{As } d_G(s,p') \leq \epsilon^{-1}_s\cdot\mathcal{D}:\\
        &< \left\lceil \frac{\epsilon_s\cdot \epsilon_s^{-1} \mathcal{D}}{\mathcal{D}} \right\rceil < 1
    \end{align}
Combining our formulas, we get
\begin{align}
\label{eqn:i9_1}
    i_{p'} &= \left\lceil \frac{d_P(s,p')}{\mathcal{D}} \right\rceil - \left(\left\lceil \frac{d_P(s,p')}{\mathcal{D}} \right\rceil - \left\lceil \frac{d_G(s,p')}{\mathcal{D}} \right\rceil\right)\\
    &\geq (i_p + 10c -1) -1 = i_p+10c - 2
\end{align}
as claimed.

We now claim that only nodes in $\mathcal{U}$ may be close to $v$.
Now assume for contradiction that both $p_1$ and $p' \not\in \mathcal{U}$ are in distance $c\mathcal{D}$ to $v$.
However, this would imply that the actual distance between $p$ and $p'$ is at most $c\mathcal{D}$ by the triangle inequality.
It holds:
\begin{align}
    d(v,p') \leq c\mathcal{D}
\end{align}
Now we consider the path from $s$ to $p'$.
By \emph{not} taking the path $\mathcal{P}$ from $s$ to $p'$, but instead the path via $p_1$ and $v$, we see that:
    \begin{align}
        i_{p'} &:= \left\lceil\frac{d_G(s,p')}{\mathcal{D}} \right\rceil\\
        &\leq \left\lceil \frac{d_G(s,p_1) + d_G(p_1,v) +d_G(v,p')}{\mathcal{D}} \right\rceil\\
        &\leq \left\lceil \frac{d_G(s,p_1) + 2c\mathcal{D}}{\mathcal{D}} \right\rceil\\
        &\leq i_{p_1} + 2c + 1  \label{eqn:i5_1}
    \end{align}
    Therefore, it holds that 
    \begin{align}
        i_{p'} \underset{\eqref{eqn:i5_1}}{\leq} i_{p_1} + 2c+1  < i_{p_1} + 10c-2 \underset{\eqref{eqn:i9_1}}{>} i_{p'}
    \end{align}
    This is a contradiction.
    Thus, no node $w \not\in \mathcal{U}$ can be close to $v$, which implies that only nodes in $\mathcal{U}$ can be close.
    As the number nodes in $\mathcal{U}$ is bounded by $O(c\epsilon_p^{-1})$, the lemma follows.
\end{proof}

Thus, by using Lemma \ref{lemma:net_distance} with $\epsilon_s = \beta^{-1}$ and $\epsilon_p=1$, there are most $O(1)$ other points between $p_1$ and $p_2$.
As we have $\kappa$ paths, there are $O(\kappa)$ points in distance $\mathcal{D}_{BC}$ to node $v$. 
Thus, we run the algorithm with parameter $\tau \in O(\kappa)$, and we have an edge-cutting probability of only $O(\frac{\ell\cdot\log\kappa}{\mathcal{D}})$.
Again, the union bound yields the claim.

\section{The Tree Operations of Ghaffari and Zuzic}
\label{sec:tree_operations}

Throughout our algorithms, we will often need to work on a forest, i.e., sets of disjoint trees. 
We use the following lemma for recurring tasks on these trees:
\begin{restatable}[Tree Operations, Based on \cite{GZ22}]{lemma}{treeoperations}
\label{lemma:tree_operations}
Let $F := (T_1, \dots, T_m)$ be a subforest (each edge $e$ knows whether $e \in E(F)$
or not) of a graph and suppose that each tree $T_i$ has a unique root $r_i \in V$, i.e., each node knows whether it is the root and which of its neighbors are parent or children, if any. 
Now consider the following three computational tasks:
\begin{enumerate}
    \item \textbf{AncestorSum and SubtreeSum:} Suppose each node $v \in T_i$ has an $\Tilde{O}(1)$-bit private input $x_v$. Further, let $Anc(v)$ and $Dec(v)$ be the ancestors and descendants of $v$ w.r.t. to $r_i$, including $v$ itself. Each node computes $A(v) := \bigotimes_{w \in Anc(v)} x_w$ and $D(v) := \bigotimes_{w \in Dec(v)} x_w$.
%    \item \textbf{Heavy Light Decomposition:} Each node learns the identifier of its heavy child, i.e., its child of the largest subtree size. Additionally, each node learns the identifier of each non-heavy child of any node on its unique path to $r_i$.
    \item \textbf{Path Selection:} Given a node $w \in T_i$, each node $v \in T_i$ learns whether it is on the unique path from $r_i$ to $w$ in $T_i$.
    \item \textbf{Depth First Search Labels:} Each node $v \in T_i$ computes its unique entry and exit label of a depth first search started in $r_i$.
\end{enumerate}
All of these tasks can be implemented in $\Tilde{O}(\HD)$ time in \CONGEST, $\Tilde{O}(1)$ depth in \PRAM, and $\Tilde{O}(1)$ time in \HYBRID.
\end{restatable}
\begin{proof}
    Tasks 1-3 can be performed entirely in the minor aggregation model. For task 4, we require two rounds of \CONGEST, \HYBRID or \PRAM.
    \begin{enumerate}
        \item \textbf{Ancestor and Subtree Sum:} This was shown in \cite[Lemma 16]{GZ22}.
        \item \textbf{Path Selection:} We perform a single minor aggregation with $x_w=1$ and $x_v=0$ for $v\in T_i\backslash\{w\}$ where we contract the unique path from $w$ to $r_i$ performing no actions in the Consensus step. Every node with value $1$ then marks itself a part of the path.
        \item \textbf{Depth First Search Labels:} We perform Ancestor sum and Subtree sum to count the number of nodes on each node's root path and subtree. Each node informs its parent about its subtree size using a single round of \CONGEST, \HYBRID or \PRAM. We order the children by ascending subtree size, breaking ties arbitrarily. To obtain the Depth First Search labels of one of its child nodes, a parent combines the length of its root path with the sizes of the subtrees traversed before that node. After computing these values, each parent uses another single round of \CONGEST, \HYBRID or \PRAM to inform its children about the computed values.           
    \end{enumerate}
\end{proof}

\end{document}